\documentclass[a4paper, intlimits, 10pt]{amsart}
\RequirePackage[OT1]{fontenc}
\RequirePackage{amsthm,amsmath,amssymb}
\usepackage{mathrsfs}
\usepackage{dsfont}
\usepackage{enumerate}
\usepackage[english]{babel}
\usepackage{graphicx}
\usepackage{caption}
\RequirePackage[numbers]{natbib}
\usepackage[usenames, dvipsnames]{color}

\renewcommand{\epsilon}{\varepsilon}

\newcommand{\R}{\mathds{R}}
\newcommand{\Q}{\mathds{Q}}

\newcommand{\N}{\mathds{N}}

\newcommand{\E}{\mathds{E}}

\renewcommand{\P}{\mathds{P}}
\newcommand{\tdP}{\tilde{\P}}

\renewcommand{\Mc}{\mathcal{M}}
\newcommand{\Qc}{\mathcal{Q}}
\newcommand{\Wc}{\mathcal{W}}
\newcommand{\prob}{\mathcal{P}}
\newcommand{\pgn}{\hat{\pi}_N}
\newcommand{\supp}{\textrm{supp}}
\newcommand{\rr}{\textbf{r}}

\numberwithin{equation}{section}
\theoremstyle{plain}
\newtheorem{Theorem}{Theorem}[section]
\newtheorem{Cor}[Theorem]{Corollary}
\newtheorem{Defn}[Theorem]{Definition}
\newtheorem{Prop}[Theorem]{Proposition}
\newtheorem{Lem}[Theorem]{Lemma}
\newtheorem{Ass}[Theorem]{Assumption}

\theoremstyle{remark}
\newtheorem{Rem}[Theorem]{Remark}
\newtheorem{Ex}[Theorem]{Example}

\usepackage[foot]{amsaddr}
\makeatletter
\renewcommand{\email}[2][]{%
  \ifx\emails\@empty\relax\else{\g@addto@macro\emails{,\space}}\fi%
  \@ifnotempty{#1}{\g@addto@macro\emails{\textrm{(#1)}\space}}%
  \g@addto@macro\emails{#2}%
}
\makeatother

\begin{document}
\title[Robust estimation of superhedging prices]{Robust estimation of superhedging prices}

\date{\today}

\author{Jan Ob{\l}{\'o}j}
\email{jan.obloj@maths.ox.ac.uk}
\address{Mathematical Institute and St.\ John's College\newline
University of Oxford\newline
AWB, ROQ, Woodstock Road\\
Oxford, OX2 6GG, UK}

\author{Johannes Wiesel}
\email{johannes.wiesel@maths.ox.ac.uk}

\begin{abstract}
 We consider statistical estimation of superhedging prices using historical stock returns in a frictionless market with $d$ traded assets. We introduce a plugin estimator based on empirical measures and show it is consistent but lacks suitable robustness. To address this we propose novel estimators which use a larger set of martingale measures defined through a tradeoff between the radius of Wasserstein balls around the empirical measure and the allowed norm of martingale densities. We establish consistency and robustness of these estimators and argue that they offer a superior performance relative to the plugin estimator. We generalise the results by replacing the superhedging criterion with acceptance relative to a risk measure. We further extend our study, in part, to the case of markets with traded options, to a multiperiod setting and to settings with model uncertainty. We also study convergence rates of estimators and convergence of superhedging strategies.
\end{abstract}


\thanks{Support from the European Research Council under the European Union's Seventh Framework Programme (FP7/2007-2013) / ERC grant agreement no. 335421 and from St John's College in Oxford are gratefully acknowledged. We thank Daniel Bartl, Stephan Eckstein, Robert Stelzer, Ruodu Wang 
and three anonymous reviewers for their very helpful comments which allowed us to improve this paper. JW further acknowledges support from the German Academic Scholarship Foundation.}

\maketitle

\section{Introduction} 

Computation of risk associated to a given financial position is one of the fundamental operations market participants have to perform. For institutional players, like banks, it is regulated by the Basel Committee \cite{basel2013fundamental} which dictates rules and requirements for such risk assessments. A golden standard has long been given by Value-at-Risk (VaR), however more recently this is being replaced by convex risk measures like Average VaR (Expected Shortfall) or more sophisticated approaches which include market modelling. Consequently, there is an abundant literature on VaR estimation and some more recent works related to statistical estimation of law-invariant risk measures, see \cite{cont2010robustness,kratschmer2012qualitative,kratschmer2014comparative, kou2013external,pichler2013evaluations, embrechts2015aggregation, embrechts2018robustness}. 
All these works consider a static situation with no trading involved.

In contrast, in this paper we consider estimation of risk for an agent who can trade in the market to offset her risk exposure. To put in evidence the novelty and relevance of our setting, we concentrate on one,  simple but canonical, way to assess risk: the \emph{superhedging price}. Consider a one-period frictionless market with prices $(S_t,S_{t+1})$ denominated in units of a fixed numeraire. The current stock prices $S_t$ are known and the future prices $S_{t+1}$ are modelled as random variables, say with return $r:=S_{t+1}/S_t$ drawn from a distribution $\P$ on $\R_+^d$. For a payoff $g:\R_+^d\to \R$, its superhedging price is given by:
\begin{equation}\label{eq:sh_def}
 \pi^{\P}(g):=\inf \{x \in \R \ | \ \exists H \in \R^d \text{ s.t. } x+H(r-1) \ge g(r) \ \P\text{-a.s.}\}.
\end{equation}
In this simple setting, an arbitrage strategy is $H\in \R^d$ such that $\P(H(r-1)\geq 0)=1$ and $\P(H(r-1)>0)>0$ and if no such strategy exists we say that no-arbitrage NA($\P$) holds. 
By the Fundamental Theorem of Asset Pricing, absence of arbitrage is equivalent to existence of a probability measure $\Q$, equivalent to $\P$, under which $S$ is a martingale, i.e., $\E_\Q[r]=1$. There might be more than one such measure and they can all be used for pricing. Taking the supremum over $\E_\Q[g]$ enables to compute the maximal feasible price for $g$ and this, by the fundamental pricing-hedging duality, is the same as the superhedging price of $g$:
\begin{equation}\label{eq:sh_duality}
\pi^{\P}(g)=\sup_{\Q \sim \P, \ \E_{\Q}[r]=1} \E_{\Q}[g],
\end{equation}
for all Borel $g$, cf.\ \cite[Thm.\ 1.31]{follmer2011stochastic}. Despite its theoretical importance and practical relevance, to the best of our knowledge, there has been no attempt to study statistical estimation of the superhedging price. 
Our paper fills this important gap. Instead of postulating a measure $\P$, we build estimators of $\pi^{\P}(g)$ directly from historical observations of returns $r_1, \ldots, r_N$ and study their properties. Furthermore, we extend the estimators to take into account also the option price data. This is practically relevant and methodologically novel in that it allows a coherent and simultaneous use of historical time-series data with current option price data or, in mathematical finance jargon, the \emph{physical measure} data and the \emph{risk neutral measure} data. 

In contrast, in existing approaches historical returns are, if at all, only used indirectly  to compute $\pi^\P(g)$. In classical mathematical finance, one first postulates a family of plausible models $\{\P_\theta:\theta\in \Theta\}$. Such choice may be influenced by stylised features of historical returns, see \cite{Rough} for a recent example, as well as by other considerations, e.g., of computational tractability. Thereon, historical returns are not used and only the ``future facing" options price data is exploited to select a candidate pricing measure $\Q_\theta$. More recently, pioneered by Mykland \cite{mykland2000conservative, mykland2003interpolation, Mykland_prediction_set} in a continuous-time setting
and pursued within the so-called robust approach to pricing and hedging, it was suggested to use historical returns to select a \emph{prediction set}, i.e., the set of paths on which the superhedging property is required, and then to compute the resulting cheapest superhedge which trades in stocks and options, see \cite{,hou2015robust,bfhmo}. Our approach inherits from that perspective but takes a statistical viewpoint and evolves it into a dynamic and asymptotically consistent methodology.

To describe our approach, suppose we observe $d$-dimensional historical returns $r_1=S_1/S_0, \dots, r_N=S_N/S_{N-1}$ and for simplicity assume that these are non-negative i.i.d.\ realisations of a distribution $\P$ which satisfies the no-arbitrage condition. We can equivalently represent the observations through their associated empirical measures  
\begin{align*}
\hat{\P}_N= \frac{1}{N} \sum_{i=1}^N \delta_{r_i},
\end{align*}  
which are well known to converge weakly to $\P$ as $N\to \infty$, see \cite[Theorem 19.1, p. 266]{van1998asymptotic}. This suggests a very natural way to approximate the superhedging price by simply using $\hat{\P}_N$ in place of $\P$. We show in Theorem \ref{Thm. one-per} below that the resulting \emph{plugin estimator} $\hat{\pi}_N(g):=\pi^{\hat{\P}_N}(g)$ is asymptotically consistent:
\begin{align*}
\lim_{N \to \infty}\hat{\pi}_N(g)=\pi^\P(g),\qquad \P^{\infty}\text{-a.s.},
\end{align*}
where $\P^{\infty}$ denotes the law of the process $(r_N)_{N\geq 1}$.
However, we also show that $\hat{\pi}_N$ has serious shortcomings. First, it is not (statistically) robust: small perturbations of $\P$ can lead to large changes in the distribution of $\hat{\pi}_N$. We argue that the L\'evy-Prokhorov metric used in the classical definition of statistical robustness, Definition \ref{def:statrobust}, is not appropriate when looking at the financial context of derivatives pricing. We propose and study alternative metrics and ensuing notions of statistical robustness in Section \ref{sec:robust}. 

Second, the plugin estimator also lacks robustness from the financial point of view of risk management. In fact, $\hat{\pi}_N$ is monotone in $N$ and converges from below so it is always a lower estimate of the risk: $\hat{\pi}_N\leq \pi^\P$. In Theorem \ref{thm:cnv_rate_basic}, and in more detail in \cite[Section \ref{sec:cnvrates}]{stats2}, we study the convergence rates for the plugin estimators. This, in the one-dimensional case $d=1$, could be exploited to build conservative estimates for the superhedging price $\pi^\P$.

A first intuition to improve the plugin estimator could be to turn to estimators of the support of $\P$. Indeed, the superhedging price $\pi^\P(g)$, say for a continuous $g$, only depends on $\P$ via its support. We could thus replace the $\hat{\P}_N$-a.s.\ inequality in the plugin estimator by an inequality on an estimator of the support of $\P$. Such estimators are well studied in statistics, going back to \cite{geffroy,chevalier,devroye1980detection, grenander1981abstract}, see also \cite{cuevas1990pattern, korostelev1995efficient, mammen1995asymptotical, hardle1995estimation, polonik1995measuring, tsybakov1997nonparametric, cuevas2004boundary,casal2007set}. Unfortunately, this approach does not seem to hold any ground. First, convergence of support estimators usually imposes strong conditions on $\P$, e.g., compactness and convexity of the support and/or existence of a density. Second, for the convergence of the superhedging prices we would also need to impose some uniform continuity assumptions on $g$. However, under such conditions on $\P$ and $g$, we could directly improve the plugin estimator and consider a suitable $\hat{\pi}_N+a_N$, see Section \ref{sec:plugin_cnvrate}. 

Instead, to address the shortcomings of the plugin estimator, we propose novel estimators, which we introduce in Section \ref{sec. impr}. 
They exploit the dual formulation of the superhedging price in \eqref{eq:sh_duality}. 
In order to achieve financial robustness and to increase our point estimates we need to consider a larger class of martingale measures. Thus we consider
$$\pi_{\Qc_N}(g)=\sup_{\Q\in \Qc_N}\E_\Q[g],$$
where $\Qc_N$ is a subset of all martingale measures $\mathcal{M}$. The plugin estimator corresponds to taking $\Qc_N=\{\Q\in \Mc: \Q\sim \hat{\P}_N\}$ and it is natural to replace it with
$$\Qc_N=\{\Q\in \Mc: \exists \tilde\P\in B_N(\hat{\P}_N) \textrm{ s.t. }\Q\sim \tilde\P \},$$
where $B_N(\hat{\P}_N)$ is some ``ball" in the space of probability measures around the empirical measure $\hat{\P}_N$. 
We show that this can lead to a consistent estimator if we use a  sufficiently strong metric, e.g., the Wasserstein infinity metric $\Wc^\infty$. In general however such $\Qc_N$ is too large. Instead, our main insight is to consider a tradeoff between the radius of the balls and the behaviour of martingale densities:
$$\hat{\mathcal{Q}}_N := \{\Q \in \mathcal{M} \ | \ \|d\Q/d\tilde\P\|_{\infty}\le {k_N} \text{ for some } \tilde\P \in B^p_{\epsilon_N}(\hat{\P}_N) \},$$ where $B^p_{\epsilon_N}(\hat{\P}_N)$ denotes the $p$-Wasserstein ball of radius $\epsilon_N$ around $\hat{\P}_N$ and $\epsilon_N \to 0$ as well as $k_N \to \infty$. 
With a suitable choice of $\epsilon_N,k_N$, we establish consistency of $\pi_{\hat{\Qc}_N}(g)$ for a regular $g$, see Theorem \ref{Thm wasserstein}, and also their financial robustness, see Corollary \ref{cor. fin_rob_was}. 
This also allows us to study the cases when the estimator naturally extends to the setting of superhedging under model uncertainty about $\P$, see Corollary \ref{cor:modeluncertainty}. 
The statistical robustness of $\pi_{\hat{\Qc}_N}(g)$ is shown in Section \ref{sec:robust}, see Theorem \ref{thm rob_was}. \
In Section \ref{sec:riskmeasures} we extend our analysis to the case when risk is assessed not using the superhedging capital but rather via a generic risk measure $\rho$ admitting a Kusuoka representation (\cite{kusuoka2001law}, see \eqref{eq:kusuoka} for a definition). We stress that this is substantially different to all the works recalled at the beginning of this Introduction since we consider an agent who can trade and optimises her position to offset the risk. We propose an estimator, inspired by $\pi_{\hat{\Qc}_N}(g)$, and show its consistency.

Finally, we also propose another estimator:
$$\sup_{\Q\in \mathcal{M}} \left(\E_{\Q}[g]-C_N\left(\inf_{\hat{\Q} \sim \hat{\P}_N, \ \hat{\Q}\in \mathcal{M}} \left\| \frac{d\hat{\Q}}{d\Q} \right\| _{\infty}-1\right)\right),$$
which  is inspired by penalty methods used in risk measures and their representations as non-linear expectations. 
Asymptotic consistency of this estimator is shown in Thereom \ref{thm:penalty} and holds for an arbitrary measurable bounded $g$.

The rest of the paper is organised as follows. In Section \ref{sec. plugin} we study the plugin estimator $\hat{\pi}_N$: its consistency, convergence rates and robustness, both statistical and financial. In Section \ref{sec. impr} we propose improved estimators and establish consistency for all of them, under different sets of assumptions. Subsequently, in Section \ref{sec:robust}, we discuss statistical robustness of all the estimators. We show in particular that no estimator can be robust in the classical sense of Tukey-Huber-Hampel, and suggest ways to amend the classical definition to make it more appropriate to the superhedging price estimation. We then detail further applications of the main results. In particular we discuss convergence of superhedging strategies in Section \ref{sec. strat} and partially extend the results to a multiperiod setting in Section \ref{sec mult}. Section \ref{sec. arb} discusses estimators $\pi_{{\Qc}_N}$ for generic sets of martingale measures $\Qc_N$ and derives necessary and sufficient conditions for asymptotic consistency of estimators. In particular, it motivates the estimators studied in Section \ref{sec. impr}. Finally, Section \ref{sec:cnvrates}  studies in more detail convergence rates of the plugin estimator $\hat{\pi}_N$ when $d=1$ and is auxiliary to Section \ref{sec:plugin_cnvrate}. Lastly, Appendix \ref{sec:appendix} contains the remaining proofs, along with auxiliary results. 
\\
\textbf{Notation.} 
We write $\prob(A)$ for the set of probability measures on $A\subset \R^d$. 
$\P_n\Rightarrow \P$ denotes weak convergence of measures. 
$\P\in \prob(\R^d_+)$ is a generic distribution for returns $r$ so that $\E_\P[r]=\int_{\R^d_+} x\P(dx)$. We let $\Mc=\{\Q\in \prob(\R^d_+): \E_\Q[r]=1\}$ denote the set of martingale measures for the stock prices $(S_t,S_{t+1})$, where $S_t>0$ is fixed and $S_{t+1}=r S_t$. We write $\Mc_{A}$ for the set of martingale measures supported on $A$. 
We say that $\P\in \prob(\R^d_+)$ does not admit arbitrage, or that NA$(\P)$ holds, if $\{\Q\in \Mc: \Q\sim \P\}\neq \emptyset$.
Above, and throughout, $H$ is a row vector, $r$ is a column vector and $1$ denotes either a scalar or a column vector $(1,\ldots, 1)^{\mathrm{T}}$.

\section{The plugin estimator}\label{sec. plugin}

Recall that we want to build an estimator for the superhedging price $\pi^\P(g)$. The easiest and possibly most natural way to do this is simply to replace the measure $\P$ with the empirical measures $\hat{\P}_N$. This yields the \emph{plugin estimator}: 
\begin{equation}\label{eq:plugin}
\pgn(g):=\pi^{\hat{\P}_N}(g).
\end{equation}
 In this section we develop the necessary tools to show asymptotic consistency of this estimator and understand its properties. The proofs are reported in Appendix \ref{sec:appendix}.

\subsection{Consistency}
We now state the main result of this section:
\begin{Theorem}\label{Thm. one-per}
Let  $\P_1, \P \in \prob(\R_+^d)$ and $g: \R^d_+ \to \R$ be Borel measurable. Assume that $r_1, r_2, \dots$ are realisations of a time-homogeneous ergodic Markov chain with initial distribution $\P_1$ and unique invariant distribution $\P$ such that $\P_1\ll \P$. Then
\begin{equation}\label{eq:pgn_cnv}
\begin{split}
\lim_{N \to \infty} \pgn&(g)=\pi^\P(g) \qquad \P^{\infty}\text{-a.s.},
\end{split}
\end{equation}
where $\P^{\infty}$ denotes the law of the Markov process started from $\P_1$.
\end{Theorem}
\begin{Rem}\label{rem:applicable}
The assumptions in the above theorem are standard in econometric theory and cover a variety of models frequently used for modelling of financial returns data. We refer to Corollary \ref{cor:ergodicity} in \citep{stats2} for sufficient conditions for stationarity (with exponential decay rates) for various random coefficient autoregressive models, e.g., linear and power GARCH and stochastic autoregressive volatility models, which are frequently used for option pricing. Nevertheless we remark that this assumption rules out deterministic trends, structural breaks and seasonalities, which need to be treated separately.
\end{Rem}
The proof for a general $g$ follows by Lusin's theorem from the case of a continuous claim $g$ which in turn depends on the characterisation of the superhedging price using concave envelopes, which we now recall. 
\begin{Defn} Let $g: \R^d_+\to \R$ be Borel.
For $A \subseteq \R_+^d$ and $x \in A$ we define the pointwise concave envelope
\begin{align*}
\hat{g}_A(x)= \inf \{ u(x)\ | \ u:\R_+^d \to \R \text{ closed concave, }u \ge g \text{ on } A \}.
\end{align*}
We define the $\P$-a.s. concave envelope as
\begin{align*}
\hat{g}_{\P}(x)= \inf \{u(x) \ | \ u: \R_+^d \to \R \text{ closed concave, }  u\ge g \ \P\text{-a.s.} \}.
\end{align*}
\end{Defn}
It is well known that in the definition of concave envelopes above we could take infimum over affine functions instead of concave functions. It follows from the definition of the superhedging price in \eqref{eq:sh_def} that we have
\begin{equation}
\label{eq:cncenv}
\pi^{\P}(g)=\hat{g}_{\P}(1)\quad \textrm{and} \quad \pgn(g)=\hat{g}_{\hat{\P}_N}(1)=\hat{g}_{\{r_1,\ldots,r_N\}}(1).
\end{equation}


Properties and computational methods for concave envelopes, or more generally for convex hulls of a set of discrete points, have been studied in many applied sciences and there are a number of efficient numerical routines available for their calculation. Naturally computational complexity increases with higher dimensions. Nevertheless there exist algorithms determining approximative convex hulls, whose complexity is independent of the dimension, see for instance \cite{sartipizadeh2016computing}.

To establish a dual formulation for the plugin estimator, assume now that $\P_1\ll \P$ as well as no-$\P$-arbitrage, NA$(\P)$, holds and recall this implies the pricing-hedging duality, cf.\ \eqref{eq:sh_duality}.
It turns out that since $\text{supp}(\hat{\P}_N) \subseteq \text{supp}(\P)$ this already implies that NA$(\hat{\P}_N)$ holds for $N$ large enough. More generally we have:
\begin{Prop}\label{prop. supp}
Let  $\P\in \prob(\R_+^d)$ and $(\P^N)_{N \in \N}$ be a sequence of probability measures on $\R_+^d$ such that $\P^N \Rightarrow \P$ and $\text{supp}(\P^N) \subseteq \text{supp}(\P)$. Then 
\begin{align*}
\text{NA}(\P) \hspace{0.5cm} \Leftrightarrow \hspace{0.5cm} \exists N_0 \in \N \text{ s.t.  NA}(\P^N) \ \text{for all }N \ge N_0.
\end{align*}
In particular, if NA($\P$) holds then in the setup of Theorem \ref{Thm. one-per} we also have
\begin{equation}\label{eq. long}
\begin{split}
\lim_{N \to \infty} \pgn(g)
=\lim_{N \to \infty} \sup_{\Q \sim \hat{\P}^N, \ \Q \in \mathcal{M}} \E_{\Q}[g]=\pi^\P(g) \qquad {\P}^{\infty}\text{-a.s.}
\end{split}
\end{equation}
\end{Prop}
We close this section considering an extended setup where in addition to the traded assets $S$, whose historical prices we observe, there also exist options in the market, which can be used for hedging $g$. If the market enlarged with those options does not allow for an arbitrage, the superhedging price of $g$ in this market is again  approximated by the plugin estimator, which now also allows for trading in the options. More precisely, we have the following:

\begin{Cor}\label{Cor. options}
Let  $\P\in \prob(\R_+^d)$ and $g: \R^d_+ \to \R$ be Borel-measurable. In addition to the assets $S$, assume that there are $\tilde{d}$ traded options with continuous payoffs $f_1(r)$ and prices $f_0$ in the market. Define the evaluation map 
$$e(r)=\Big(r^1, \dots, r^d, f^1_1(r)/f^1_0, \dots,f^{\tilde{d}}_1(r)/f^{\tilde{d}}_0 \Big)^{\mathrm{T}}$$ 
and $\tilde{\P}:=\P\circ e^{-1}$. Finally assume no arbitrage, NA$(\tilde{\P})$, holds. Then, under the assumptions of Theorem \ref{Thm. one-per}, we have $\P^{\infty}$-a.s.
\begin{align*}
&\lim_{N \to \infty} \inf\{x \in \R \ | \ \exists H \in \R^{d+\tilde{d}} \text{ s.t. } x+H(e(r)-1)\ge g(r) \ \forall r\in \{r_1, \dots, r_N\}\}\\
=\ &\inf\{x \in \R \ | \ \exists H \in \R^{d+\tilde{d}} \text{ s.t. } x+H(e(r)-1)\ge g(r) \ \P \text{-a.s.}\}\\
=\ &\sup_{\Q \sim \P, \ \Q \in \mathcal{M}, \ \E_{\Q}[f_1]=f_0} \E_{\Q}[g].
\end{align*}
\end{Cor}
It is worth stressing that in the classical approach to pricing and hedging, the historical returns are seen as \emph{physical measure} inputs and might be used, e.g., for extracting stylised features which models should exhibit. In contrast, option prices $f_0$ are \emph{risk-neutral measure} inputs and would be used to calibrate the pricing measures. To the best of our knowledge consistent use of both in one estimator has not been achieved before.

\subsection{Statistical robustness}\label{sec:pgn_robust}
Robustness of estimators is concerned with their sensitivity to perturbation of the sampling measure $\P$. To formalise this, suppose we have a sequence of estimators $T_N$ which can be expressed as a fixed functional  $T: \mathcal{P}(\R_+^d) \to \R$ evaluated on the sequence of empirical measures, i.e., $T_N=T(\hat{\P}_N)$. This is clearly the case with the plugin estimator of the superhedging price in \eqref{eq:plugin}. Hampel \cite{hampel1971general} proposed the following definition of statistical robustness:
\begin{Defn}[\cite{huber2011robust}, p. 42]\label{def:statrobust}
Let $r_1, r_2, \dots$ be i.i.d.\ from $\P \in \prob(\R_+^d)$. The sequence of estimators $T_N=T(\hat{\P}_N)$ is said to be robust at $\P$ if for every $\epsilon>0$ there is $\delta>0$ and $N_0 \in \N$ such that for all $\tilde\P\in \prob(\R_+^d)$ and $N\ge N_0$ we have
\begin{align*}
d_L(\P, \tilde\P) \le \delta \quad\Longrightarrow \quad d_L(\mathcal{L}_{\P}(T_N), \mathcal{L}_{\tilde\P}(T_N)) \le \epsilon,
\end{align*}
where $d_L$ is the L\'{e}vy-Prokhorov metric 
\begin{align}\label{eq:LPdistance}
d_{L}(\P,\tilde{\P}):= \inf \{\delta>0 \ | \ \P(B) \le \tilde{\P}(B^{\delta})+\delta \ \text{for all } B \in \mathcal{B}(\R_+^d) \}.
\end{align}
\end{Defn}
We sometimes say that $T_N$ is robust with respect to $d_L$ to stress the dependency on the particular choice of the metric. 
A classical result of Hampel, see \cite[Thm.\ 2.21]{huber2011robust}, states that if $T$ is asymptotically consistent, i.e.
$$ T_N=T(\hat{\P}_N)\longrightarrow T(\P), \quad \textrm{for all }\P\in \prob(\R^d_+)$$
then $T_N$ is robust at $\P$ if and only if $T(\cdot)$ is continuous at $\P$. The following theorem characterises weak continuity of the superhedging price and hence also robustness of its estimators. In particular, it implies that even for i.i.d.\ returns $\pgn$ is robust only for special combinations of $g$ and $\P$.
\begin{Theorem}\label{thm. rob}
Let $g$ be continuous and $\P\in \prob(\R^d_+)$. Then the functional $\tilde\P \mapsto \pi^{\tilde\P}(g)$ is lower semicontinuous at $\P$. It is continuous if and only if
\begin{align}\label{eq:sh trivial}
\pi^{\P}(g)=\sup_{\Q \in \mathcal{M}} \E_{\Q}[g].
\end{align}
In consequence, any asymptotically consistent estimator $T_N$ 
is robust at $\P$ only if the above equality holds true. 
\end{Theorem}
In particular we see that, in general, the plugin estimator $\pgn(g)$ is not robust w.r.t.\ $d_L$. The fact that this holds for any asymptotically consistent estimator suggests strongly that the classical definition of robustness is not adequate in the present context. The superhedging price $\pi^\P(g)$ is concerned with the support of $\P$ in the sense that for $\P_1,\P_2\in \prob(\R^d_+)$ with equal supports, and for a continuous $g$, we have $\pi^{\P_1}(g)=\pi^{\P_2}(g)$. In contrast, any $\delta$-perturbation in the L\'evy-Prokhorov sense allows for arbitrary changes to the support, see Lemma \ref{lem. wasser v2}.
In particular, even if $\P$ satisfies no-arbitrage, measures in its neighbourhood may not and one may not employ \eqref{eq:sh_duality} for these. To control the support, we can consider $d_H(\text{supp}(\P), \text{supp}(\tilde\P))$, for $\P,\tilde \P\in \mathcal{P}(\R_+^d)$ and where $d_H$ denote the Hausdorff metric on closed subsets of $\R^d_+$.
\begin{Prop} \label{prop:d_H cont}
Let $g: \R_+^d \to \R$ be uniformly continuous and let $\P\in \mathcal{P}(\R_+^d)$ such that NA$(\P)$ holds. Then the functional $\prob(\R^d_+)\ni \tilde{\P}\to \pi^{\tilde{\P}}(g)$ 
is continuous w.r.t.\ the pseudo-metric $d_H(\text{supp}(\P), \text{supp}(\tilde{\P})).$
\end{Prop}
Alas, this does not allow us to recover statistical robustness of the plugin estimator as the pseudometric above does not admit control over the tails of $\P$. Instead, in Section \ref{sec:robust_pathwise}, we consider a stronger $\Wc^\infty$ metric which allows to obtain an analogue to Hampel's robustness result.

\subsection{Financial robustness}
The plugin estimator $\pgn$ not only lacks statistical robustness, as seen above, but is also not a financially robust estimate of risk. In fact, if $\P_1\ll \P$, it converges to the superhedging price from below, i.e., $\pgn\nearrow \pi^\P$. From a risk-management perspective one would like to find a consistent estimator for the $\P$-a.s.\ superhedging price converging from above. However, as we now show, this is not possible in general. As a direct consequence of the discontinuity of the superhedging functional with respect to the L\'{e}vy-Prokhorov metric $d_L$, the convergence from above at some confidence level cannot be achieved in practical applications. 
\begin{Prop}\label{prof:nocnvabove}
Let $\P\in \prob(\R^d_+)$ satisfy NA$(\P)$ and $g$ be bounded and Lipschitz continuous. Then, there exists no consistent estimator $T_N$ of $\pi^\P(g)$ such that for a confidence level $\alpha \in [0,1]$ there exists $N_0 \in \N$ and 
\begin{align}\label{eq. conf}
\P^{\infty}\left(T_N \ge \sup_{\Q \in \mathcal{M}, \ \Q \sim \P} \E_{\Q}[g] \text{ for all } N \ge N_0\right)\ge \alpha
\end{align}
for all $\P \in \P(\R_+^d)$.
\end{Prop}
Thus, in order to achieve the above property \eqref{eq. conf} it is necessary to make additional regularity assumptions on $\P$ and $g$. We show that this is possible for suitably conservative estimators, see Section \ref{subsec dist} below. 
In the case of the plugin estimator, we can never achieve convergence from above but we can develop an understanding of the order of magnitude of the difference $\pi^{\P}(g)-\hat{\pi}_N(g)$. We first do this by studying the convergence rates, see also \cite[Section \ref{sec:cnvrates}]{stats2}. Secondly, we achieve this via notions of statistical robustness suited for the plugin estimator, see Section \ref{sec:robust_pathwise}.

\subsection{Convergence rates}
\label{sec:plugin_cnvrate}
We now investigate the convergence rate in \eqref{eq. long}. While motivated by financial considerations, the question is of independent interest. We focus on the one-dimensional case. We let $F_\P$ be the cumulative distribution function of $\P\in \prob(\R_+)$ and $d_N=\sup_{r \in \R_+} |F_{\hat{\P}_N}(r)-F_\P(r)|$ denote the Kolmogorov-Smirnov distance between $\hat{\P}_N$ and $\P$. 
\begin{Defn}
For $N \in \N$ and $k=1, \dots, \lfloor 1/(3d_N)\rfloor$ we define the interquantile distance
\begin{align*}
\kappa_k^N= 
F^{-1}_{\P}(3kd_N)-F^{-1}_{\P}(3(k-1)d_N\vee 0+) & \text{ for }k=1, \dots \lfloor 1/(3d_N)\rfloor.
\end{align*}
Furthermore we set
\begin{align*}
\kappa^N_0= \begin{cases}
F^{-1}_{\P}(1)-F^{-1}_{\P}(1-d_N) & \text{if } \P\text{ has bounded support,}\\
0& \text{otherwise}.
\end{cases}
\end{align*}
and let $\kappa^N=\sup_{k\in \{0, \dots,\lfloor 1/(3d_N)\rfloor \}} \kappa_k^N.$
\end{Defn}
We can now establish the speed of convergence for the plugin estimator. 
\begin{Theorem}\label{thm:cnv_rate_basic}
In the setup of Theorem \ref{Thm. one-per} assume that NA$(\P)$ holds and $g$ is bounded and uniformly continuous with $|g(r)-g(\tilde{r})| \le \delta(|r-\tilde{r}|)$ for some $\delta: \R_+ \to \R_+$ such that $\delta(r) \to 0$ for $r \to 0$. Then, as $N\to \infty$, 
\begin{align*}
\pi^{\P}(g)-\hat{\pi}_N(g)&=\sup_{\Q \sim \P, \ \Q \in \mathcal{M}} \E_{\Q}[g]- \sup_{\Q \sim \hat{\P}_N, \ \Q \in \mathcal{M}} \E_{\Q}[g] \\
&= \begin{cases}
\mathcal{O}(\delta(\kappa^N))& \text{if } \P\text{ has bounded support,}\\
\mathcal{O}\left(\delta(\kappa^N)+\frac{1}{F^{-1}_{\P}\left(1-d_N\right)}\right)& \text{otherwise}.
\end{cases}
\end{align*}
\end{Theorem}
\begin{Rem}
When the support of $\P$ is bounded, the above result holds for all continuous $g$. 
Furthermore $\kappa^N$ tends to $0$ as $N \to \infty$. 
\end{Rem}
\begin{Lem}[{Dvoretzky-Kiefer-Wolfowitz, cf.\ \citep[Thm.~11.6]{kosorok2008introduction}}]\label{lem dkw} 
Suppose the returns $r_1, r_2, \ldots$ are i.i.d.\ samples from $\P \in \mathcal{P}(\R_+)$. Then for every $\epsilon>0$ 
\begin{align*}
\P^{\infty}(d_N > \epsilon) \le 2e^{-2N\epsilon^2}.
\end{align*}
\end{Lem}
Theorem \ref{thm:cnv_rate_basic} and Lemma \ref{lem dkw} yield probabilistic bounds on the distribution of  $\pi^{\P}(g)-\hat{\pi}_N(g)$. This is explored in \cite[Section \ref{sec:cnvrates}]{stats2}, where we also prove Theorem \ref{thm:cnv_rate_basic} and provide extensions of Lemma \ref{lem dkw}.

\section{Improved estimators for the $\P$-a.s. superhedging price}
\label{sec. impr}

In the last section we have seen that the plugin estimator is asymptotically consistent but has important shortcomings from a statistical and financial point of view. To address these, we propose now new estimators and investigate their asymptotic behaviour as well as their robustness. To construct these, we consider ``balls" around the empirical measure $\hat{\P}_N$ and we rely on recent convergence rate results of $\hat{\P}_N$ to $\P$ for the choice of the radii.

We start by considering balls in the Wasserstein$-\infty$ metric, which offers a very good control over the support but where we need to make strong assumptions on $\P$ to control the rate of convergence for $\hat{\P}_N$. Subsequently, in Section \ref{subsec dist}, we consider Wasserstein$-p$ metrics, $p\geq 1$. The use of weaker metrics allows us to treat all measures admitting suitable finite moments but requires a penalisation over the dual (pricing) measures.  
In fact, our estimators rely on a suitable combination of results on convergence of empirical measures with insights into pricing and control over martingale densities. Similarly to the spirit of Corollary \ref{Cor. options} above, we combine the \emph{physical measure}- and the \emph{risk neutral measure}- arguments, see \eqref{eq:WassBall}. 
We note that using Wasserstein metrics, as opposed to weaker metrics, allows us to control the first moment which is important for no-arbitrage reasons, see \cite[Section \ref{sec. arb}]{stats2}.
Finally, in Section \ref{sec:penaltyestimator}, we consider much larger balls, indeed all of $\Mc$, and let penalisation select the appropriate measures. 
Short proofs are given here, the proofs are reported in Appendix \ref{app. 1}.
\subsection{Wasserstein $\Wc^\infty$ balls}\label{sec. infi}
When considering robustness of the plugin estimator we saw that to consider measures in a ball around $\hat{\P}_N$ we have to consider a notion of distance which, unlike the L\'evy-Prokhorov metric, controls the supports. This is achieved by the Wasserstein-$\infty$ distance
\begin{align}\label{eq:Winfty}
\mathcal{W}^{\infty}(\P, \tilde{\P})&:= \inf_{\gamma\in \Pi(\P,\tilde{\P})} \gamma\text{-ess-sup }|x-y|\nonumber \\
&=\inf\left\{\epsilon >0 \ \Big| \ \P(B) \le \tilde{\P}(B^{\epsilon}), \ \tilde{\P}(B) \le \P(B^{\epsilon}) \ \forall B \in \mathcal{B}(\R_+^d) \right\},
\end{align}
where $\Pi(\P,\tilde{\P})$ denotes the set of all probability measures $\gamma$ with marginals $\P$ and $\tilde{\P}$ and where the equality between the definition and the second representation is a consequence of the Skorokhod representation theorem. A direct comparison of \eqref{eq:Winfty} with \eqref{eq:LPdistance} reveals that $\Wc^\infty$ controls support in the way that $d_L$ does not. However, one immediate issue with considering $\Wc^\infty$ is that if $\P$ has unbounded support
then $\mathcal{W}^{\infty}(\P, \hat{\P}_N)=\infty$ for all $N \in \N$ since $\hat{\P}_N$ are finitely supported. For this reason, and also to obtain appropriate confidence intervals, in order to build a good estimator using $\Wc^\infty$-balls we have to impose relatively strong assumptions on $\P$:
\begin{Ass}\label{Ass. bilip}
The measure $\P$ is an element of $ \prob(A)$ for a connected, open and bounded set $A\in \mathcal{B}(\R_+^d)$ with a Lipschitz boundary. Furthermore $\P$ admits a density $\rho: A \to (0, \infty)$ such that there exists $\alpha \ge 1$ for which $1/\alpha \le \rho(r)\le \alpha$ on $A$.
\end{Ass}
Under the above assumption, we have explicit bounds on $\Wc^\infty(\P,\hat{\P}_N)$. The case $d=1$ follows from Kiefer-Wolfowitz bounds while the case $d\geq 2$ was established in \cite[Thm.~1.1]{trillos2014rate}.
\begin{Lem}\label{Lem d=1}
Assume that $\P$ fulfils Assumption \ref{Ass. bilip} and NA$(\P)$ holds. Furthermore let $r_1, r_2, \dots$ be i.i.d. samples from $\P$. If $d=1$, then except on a set with probability $O(\exp(-2\sqrt{N}))$, $\mathcal{W}^{\infty}(\P, \hat{\P}_N) \le l_N(1,\alpha,A):=\alpha N^{-1/4}$. If $d\geq 2$, then except on a set with probability $O(N^{-2})$,
\begin{align*}
\mathcal{W}^{\infty}(\P, \hat{\P}_N) \le l_N(d,\alpha,A ):= C(\alpha,A)\begin{cases}
\frac{\log(N)^{3/4}}{N^{1/2}} & \text{if }d=2, \\
 \frac{\log(N)^{1/d}}{N^{1/d}} & \text{if }d\ge 3.
\end{cases}
\end{align*}
\end{Lem}

We let $B_\epsilon^\infty(\P)$ denote a $\Wc^\infty$-ball of radius $\epsilon$ around $\P$. The above lemma allows to deduce consistency of the estimator based on such $\Wc^\infty$ balls:
\begin{Prop}\label{prop:Wcinfty_est}
Consider $\P\in \prob(\R^d_+)$ satisfying NA$(\P)$ and Assumption \ref{Ass. bilip}, and let $\alpha$, $l_N:=l_N(d,\alpha, A)$ be as in Lemma \ref{Lem d=1}. Let $g$ be a continuous function and $r_1, r_2, \dots$ be i.i.d.\ samples from $\P$. Then  
\begin{align*}
\hat{\pi}_N^\infty(g):=\sup_{\tdP \in B^{\infty}_{l_N}(\hat{\P}_N)}\pi^{\tdP}(g) \searrow \pi^{\P}(g), \quad \text{as }N \to \infty,\quad \P^\infty-a.s.
\end{align*}
\end{Prop}
\begin{Rem}
The practical use of $\Wc^\infty$ estimators requires a good handle on $l_N$. Its dependence on the set $A$ is mild -- from \cite{trillos2014rate} we see that if a bi-Lipschitz homeomorphism $\phi: \overline{A}\to [0,1]^d$ exists, then the constant $C$ depends on the domain $A$ only via the Lipschitz constant of $\phi$. However, the knowledge of $\alpha$ requires uniform \emph{a priori} estimates 
on the density of $\P$ on $A$. This should be contrasted with $\Wc^p$ estimators below, which only require finiteness of certain moments.  
\end{Rem}
\begin{proof}[Proof of Proposition \ref{prop:Wcinfty_est}]
From Lemma \ref{Lem d=1}, an application of Borel-Cantelli shows that, $\P^\infty$-a.s., for $N$ large enough $\P\in B^{\infty}_{l_N}(\hat{\P}_N)$ and, in particular, $\hat{\pi}_N^\infty\geq \pi^\P$. Further, as $d_H(\text{supp}(\P),\text{supp}(\tilde{\P}))\leq \Wc^\infty(\P,\tilde{\P})$ by \eqref{eq:Winfty}, for all $\tdP \in B^{\infty}_{l_N}(\hat{\P}_N)$ we have $\text{supp}(\tdP) \subseteq \overline{\text{supp}(\P)^{2l_N}}$, a compact on which $g$ is uniformly continuous. For measures supported on this compact, Proposition \ref{prop:d_H cont} yields continuity of $\tilde{\P}\to \pi^{\tilde{\P}}$ with respect to $\Wc^\infty$, which in turn implies consistency of $\hat{\pi}^\infty_N$ and concludes the proof.
\end{proof}
Thus $\hat{\pi}_N^\infty$ is not only consistent but also financially robust. We shall see in Corollary \ref{cor:Wcinfty_rob} below, that it is also statistically robust with respect to $\Wc^\infty$. However, these results only hold for measures $\P$ which satisfy Assumption \ref{Ass. bilip}. In the next section we introduce a family of estimators which exhibit similar desirable properties for a much larger class of measures $\P$.

\subsection{Wasserstein $\Wc^p$ balls and martingale densities}
\label{subsec dist}

We assume no-arbitrage NA($\P$) holds and exploit \eqref{eq:sh_duality} to consider estimators of the form
\begin{equation}\label{eq:generic_estimator}
\pi_{\Qc_N}(g):=\sup_{\Q \in \mathcal{Q}_N} \E_{\Q}[g]
\end{equation}
for different specifications of the sets of martingale measures $\Qc_N$ based on ``balls" around $\hat{\P}_N$. 
In order to guarantee asymptotic consistency we have to ascertain that the true measure $\P$ is contained in these balls and that we have some control over the tails of the martingale measures in $\Qc_N$.
Our crucial insight, following recent work of \cite{esfahani2015data}, is to work with Wasserstein metrics defined, for $p\geq 1$ and $\P,\tilde\P\in \prob(\R_+^d)$ with a finite $p^{\textrm{th}}$ moment, by
\begin{align*}
\mathcal{W}^p(\P, \tilde{\P})&=\left( \inf \left\{\int_{\R_+^d \times \R_+^d}|r-s|^p \gamma(dr,ds) \ \bigg| \ \gamma \in \Pi(\P, \tilde{\P}) \right\}\right)^{1/p}
\end{align*}
where $\Pi(\P, \tilde{\P})$ is the set of probability measures on $\R_+^d \times \R_+^d$ with marginals $\P$ and $\tilde{\P}$. In case $p=1$, \cite{kellerer1982duality} showed that Kantorovitch-Rubinstein duality (see \citep[Theorem 11.8.2, p.421]{dudley2018real}) has a particularly nice expression:
\begin{align}\label{eq:KRduality}
\mathcal{W}^1(\P, \tilde{\P})
&= \sup_{f \in \mathcal{L}_1} \left(\, \int_{\R^d_+} f(y)d\mu(y)- \int_{\R^d_+} f(y) d\nu(y) \right),
\end{align}
where $\mathcal{L}_1$ denotes the 1-Lipschitz continuous functions $f: \R_+^d \to \R$.\footnote{We develop the theory for all $p\geq 1$. In practice, the choice of $p$ has to be made by the statistician. From Theorem \ref{thm rob_was}  and the equation above Corollary \ref{cor. fin_rob_was} it is apparent that, for robustness, one wants to take $p$ as large as possible. This however makes moment assumptions more restrictive. 
The cases $p = 1$ and $p = 2$ are the most popular in literature, given in particular the nice duality for $p = 1$ and the fact that $L^2$ is a Hilbert space.}
A Wasserstein ball around $\P$ is denoted
$$B^p_{\epsilon} (\P)=\{\tilde\P\in \prob(\R^d_+)| \Wc^p(\P,\tilde\P)\leq \epsilon\}.$$
For a given $\epsilon\geq 0$ and $k\in (0,\infty]$, let 
\begin{align}\label{eq:WassBall}
D_{\epsilon,k}^p(\P):= \left\{\Q\in \mathcal{M} \ \bigg| \ \left\|\frac{d\Q}{d\tdP} \right\|_{\infty} \le k \text{ for some } \tdP \in B_{\epsilon}^p(\P)\right\}.
\end{align}
One's first intuition might be to use $\Qc_N=D_{\epsilon_N,\infty}^p(\hat{\P}_N)$ in \eqref{eq:generic_estimator}. Interestingly, this does not work as the balls are too large. Indeed, Wasserstein distance metrises weak convergence and Lemma \ref{lem. wasser v2} shows that any ball around $\hat{\P}_N$ includes measures with full support. As it turns out, to obtain a consistent estimator a subtle interplay is required between $\epsilon$ and $k$ in \eqref{eq:WassBall}. 
\begin{Ass}\label{Ass 1}
\begin{enumerate}
\item\label{Ass 1.1} $r_1, r_2, \dots$ are realisations of a time-homogeneous ergodic Markov chain with initial distribution $\P_1$ and unique invariant distribution $\P$ such that $\P_1\ll \P$ and $\|d\P_1/d\P\|_{L^{2s}(\P)}<\infty$ for some $s>3$. Furthermore $\E_{\P}[|r|^q]<\infty$ for some $q>2ps/(s-2)$ and there exists a sequence $(\rho_N)_{N\in\N}$ with $\sum_{N\in \N}\rho_N<\infty$ such that if $r_1\sim \P$
\begin{align}\label{eq:poincare}
\E\left[\E[f(r_N)-m(f)|r_1]^2\right]\le \rho_N^2,
\end{align}
holds for all $\|f\|_{\infty}\le 1$, all $N\in \N$, where $m(f)=\E[f(r_1)]$.
\item\label{Ass 1.2} $r_1, r_2, \dots$ are i.i.d. samples of $\P$ and there exist $a,c >0$ such that $\E_{\P}[\exp(c|r|^a)] <\infty$. 
\end{enumerate}
\end{Ass}
We again refer to \cite[Corollary \ref{cor:ergodicity}]{stats2} for examples of processes, which satisfy Assumption \ref{Ass 1}. Clearly Assumption \ref{Ass 1}.\ref{Ass 1.2} implies \ref{Ass 1}.\ref{Ass 1.1}. Under this assumption \cite{fournier2015rate}, see also \cite{esfahani2015data}, used concentration of measure techniques to obtain rates of the decay for $\P^{\infty}(\mathcal{W}(\P, \hat{\P}_N) \ge \epsilon )$, see Lemma \ref{Lem. was} and \cite[Lemma \ref{lem:fournilin}]{stats2}. This allows to compute explicitly a function $\epsilon_N:(0,1)\to \R_+$ with $\epsilon_N(\beta)\searrow 0$ as $N\to \infty$, such that 
$$\P^{\infty} (\mathcal{W}^p(\P, \hat{\P}_N) \ge \epsilon_N(\beta)) \le \beta,\quad N\geq 1.$$
We say that Assumption \ref{Ass 1} holds if either Assumption \ref{Ass 1}.\ref{Ass 1.1} holds and then $\epsilon_N$ is given in \eqref{eq:fournilin} or Assumption 
 \ref{Ass 1}.\ref{Ass 1.2} holds and $\epsilon_N$ is then given in \eqref{eq:def_epsilonN}.
We state now the main result in this section. 
\begin{Theorem}\label{Thm wasserstein}
Let $g$ be either Lipschitz continuous and bounded from below or continuous and bounded, $p \ge 1 $ and $\P \in \prob(\R_+^d)$ satisfying NA($\P$). 
Suppose Assumption \ref{Ass 1} holds and $\beta_N \in (0,1)$ satisfy $\lim_{N\to \infty} \beta_N=0$ and $\lim_{N \to \infty} \epsilon_N(\beta_N)=0$. Pick a sequence  $k_N= o(1/\epsilon_N(\beta_N))$. 
Then, the limit in $\P^\infty$-probability
\begin{align}\label{eq:weak}
\lim_{N \to \infty} \pi_{\hat{\Qc}_N}=\pi^\P(g),\hspace{0.5cm}
\end{align}
holds, where $\hat{\Qc}_N:=D^p_{\epsilon_N(\beta_N),k_N}(\hat{\P}_N)$. Furthermore, if $(\beta_N)_{N\in \N}$ satisfies $\sum_{N=1}^{\infty}\beta_N < \infty$ then the limit \eqref{eq:weak} also holds $\P^\infty$-almost surely.

\end{Theorem}
The above result shows that $\pi_{\hat{\Qc}_N}$ is an asymptotically  consistent estimator of $\pi^\P$. Note that we assume no arbitrage NA$(\P)$ so that, using \eqref{eq:sh_duality}, the convergence above is equivalent to
\begin{align*}
\lim_{N \to \infty} \sup_{\Q \in \hat{\Qc}_N} \E_{\Q}[g]= \sup_{\Q \sim \P, \ \Q \in \mathcal{M}} \E_{\Q}[g] .
\end{align*}
We write $\hat{\Qc}_N^p=\hat{\Qc}_N$ when we want to stress the dependence on $p$. 
As mentioned above, the consistency depends crucially on the choice of $\hat{\Qc}_N$. We discuss this further and motivate the above choice in \cite[Section \ref{sec. arb}]{stats2}. For $p>1$, $D^p_{\epsilon,k}(\P)$ are weakly compact but $D^1_{\epsilon, k}(\P)$ is not even weakly closed in general, see Lemma \ref{lem. comp}. In case of $p=1$, taking weak closure of $\hat{\Qc}_N^1$ could destroy the asymptotic consistency of the estimator, e.g., taking $g(r)=(r-1)$ in the example in the proof of Lemma \ref{lem. comp}.

For the particular choice of $\beta_N=\exp(-\sqrt{N})$ under Assumption \ref{Ass 1}.\ref{Ass 1.2} an explicit computation yields that for $N$ large enough we have
\begin{align*}
\epsilon_N(\beta_N)=\left(\frac{\log(c_1\exp(\sqrt{N}))}{c_2 N}\right)^{1/\min(\max(d/p,2),a/(2p))} \sim \frac{1}{N^{{1/\min(\max(2d/p,4),a/p)}}}.
\end{align*}
However many other choices of $\beta_N$ are feasible. The essential point is that for summable $(\beta_N)$, a Borel-Cantelli argument implies that for $N$ large enough the true distribution $\P$ is within an $\epsilon_N(\beta_N)$-ball around $\hat{\P}_N$. This allows us to deduce a sufficient condition for financial robustness of our estimator:
\begin{Cor}\label{cor. fin_rob_was}
In the setup of Theorem \ref{Thm wasserstein} with $\sum_{N=1}^\infty \beta_N<\infty$, let $g$ be such that 
\begin{align}\label{eq:bounded_mart_dens}
\exists C\in \R_+\quad \textrm{s.t.}\quad \sup_{\Q \in \mathcal{M}, \ \| \frac{d\Q}{d\P}\|_{\infty} \le C} \E_{\Q}[g]=\sup_{\Q \sim \P, \ \Q \in \mathcal{M}} \E_{\Q}[g]=\pi^\P(g).
\end{align}
Then $\pi_{\hat{\Qc}_N}(g)\geq \pi^\P(g)$ for $N$ large enough so that the estimator is asymptotically consistent and converges from above.
\end{Cor}
The condition \eqref{eq:bounded_mart_dens} is motivated by an approximation result, see Lemma \ref{lem. rasonyi}. It allows us also to consider the case when we are unsure about the \emph{true} measure $\P$ and instead prefer to superhedge under all measures in its small neighbourhood. 
\begin{Cor}\label{cor:modeluncertainty}
In the setup of Theorem \ref{Thm wasserstein}, fix $C>0$ and assume there exists $\Q\in \Mc$, $\|d\Q/d\P\|_\infty< C$. 
Consider $C_N\to C$ and a fixed $\epsilon>0$. Then
\begin{align*}
\lim_{N \to \infty} \sup_{\Q \in D^p_{\epsilon+\epsilon_N,C_N}(\hat{\P}_N)} \E_{\Q}[g]
=\sup_{\Q \in D^p_{\epsilon,C}(\P)} \E_{\Q}[g]
\end{align*}
holds in $\P^{\infty}$-probability and $\P^{\infty}$\text{-a.s.} whenever $\sum_{N=1}^\infty \beta_N<\infty$.
\end{Cor}
We close this section with two examples illustrating that the assumptions on regularity of $g$ in Theorem \ref{Thm wasserstein} can not be easily relaxed.
\begin{Ex}[$g$ unbounded, not Lipschitz]\label{Ex. 1}
Set $g(r)=(r-1)^2$ and consider $(r_N)_{N\geq 1}$ i.i.d.\ from $\P=\delta_1$. For $r_N \ge 2$ consider the measures
\begin{align*}
\nu_N=\frac{\epsilon_N(\beta_N)}{2}\ \delta_0+\left(1-\frac{r_N \ \epsilon_N(\beta_N)}{2(r_N-1)}\right) \ \delta_1+\frac{\epsilon_N(\beta_N)}{2(r_N-1)}\ \delta_{r_N}
\end{align*}
and 
\begin{align*}
\Q_N=\frac{\epsilon_N(\beta_N)}{2\sqrt{\epsilon_N(\beta_N)}} \ \delta_0+\left(1-\frac{r_N \epsilon_N(\beta_N)}{2(r_N-1)\sqrt{\epsilon_N(\beta_N)}}\right)\ \delta_1+\frac{\epsilon_N(\beta_N)}{2(r_N-1)\sqrt{\epsilon_N(\beta_N)}} \ \delta_{r_N}.
\end{align*}
Then $\mathcal{W}^1(\nu_N, \delta_1) \le \epsilon_N(\beta_N)$,
\begin{align*}
\left\|\frac{d\Q_N}{d\nu_N} \right\|_{\infty}\le \frac{1}{\sqrt{\epsilon_N(\beta_N)}} 
\end{align*}
and choosing $r_N=1/\epsilon_N(\beta_N)$ we find
\begin{align*}
\E_{\Q_N}[g] \ge \frac{\sqrt{\epsilon_N(\beta_N)}}{2(1/\epsilon_N(\beta_N)-1)}(1/\epsilon_N(\beta_N)-1)^2 \to \infty,\quad \textrm{ as } N \to \infty.
\end{align*}
\end{Ex}

\begin{Ex}[$g$ bounded, discontinuous]
Set $g(r)=\mathds{1}_{\{r \neq 1\}}$ and consider $(r_N)_{N\geq 1}$ i.i.d.\ from $\P=\delta_1$. Let 
\begin{align*}
\nu_N=\frac{1}{2}\delta_{1-\epsilon_N(\beta_N)/2}+\frac{1}{2}\delta_{1+\epsilon_N(\beta_N)/2},
\end{align*}
then $\mathcal{W}^1(\nu_N,\delta_1)=\epsilon_N(\beta_N)/2$. We conclude
\begin{align*}
\lim_{N \to \infty} \sup_{\Q \in \hat{\mathcal{Q}}_N} \E_{\Q}[g]\ge \lim_{N \to \infty} \sup_{\Q \ll \nu_N, \ \Q \in \mathcal{M}} \E_{\Q}[g]=1 \neq 0 = \sup_{\Q \sim \P, \ \Q \in \mathcal{M}} \E_{\Q}[g].
\end{align*}
\end{Ex}

Let us remark that $\pi_{\hat{\Qc}_N}(g)$ acting on infinite dimensional spaces is bounded by a more sophisticated version of the plugin estimator. To see this define the Average Value at risk of $g$ at level $1/k$, for $k\geq 1$, by 
\begin{align*}
AV@R_{1/k}^{\P}(g)= \max_{\tdP \sim \P, \ \| d\tdP / d\P \|_{\infty} \le k} \E_{\tdP}[g].
\end{align*}
In dimension one, $d=1$, it can be re-expressed, see \cite[Thm.~4.47]{follmer2011stochastic}, as
$$AV@R_{1/k}^{\P}(g):=k\int_{1-1/k}^1 F_{\P \circ g^{-1}}^{-1}(x) dx,$$
which makes the link with the classical Value-at-Risk apparent. If we now include the ability to trade and optimise the final position, by the translation-invariance of $AV@R_{1/k}^{\P}(\cdot)$ we can write
\begin{align*}
&\inf_{H \in \R^d} AV@R_{1/k}^{{\P}}(g(r)-H(r-1))\\
&\qquad =\inf \left\{ x \in \R \ | \ \exists H \in \R^d \text{ s.t. }  AV@R_{1/k}^{\P}(g(r)-H(r-1)-x)\le 0\right\},
\end{align*}
which is a superhedging price, where the acceptance cone is given by an $AV@R$ constraint. An analogous representation and bounds for $\pi_{\hat{\Qc}_N}$ follow:
\begin{Cor}\label{Lem:AVaR} In the setup of Theorem \ref{Thm wasserstein}, let $g$ be $1$-Lipschitz and $\P\in \prob(\R^d_+)$ satisfying NA$(\P)$. Then  there exists $N_0\in \N$ such that for all $N\ge N_0$
\begin{align} \label{eq. avar1}
&\inf_{H \in \R^d} AV@R_{1/k_N}^{\hat{\P}_N}(g(r)-H(r-1))\nonumber\\
&\qquad \le \inf_{H \in \R^d} \sup_{\tilde{\P} \in B_{\epsilon_N(\beta_N)}^p(\hat{\P}_N)} AV@R_{1/k_N}^{\tilde{\P}}(g(r)-H(r-1)) =\pi_{\hat{\Qc}_N}(g)\nonumber\\
&\qquad=\inf \Bigg\{ x \in \R \ | \ \exists H \in \R^d \text{ s.t. }  \\
&\qquad\qquad\qquad\sup_{\tilde{\P} \in B^p_{\epsilon_N(\beta_N)}(\hat{\P}_N)}AV@R_{1/k_N}^{\tilde{\P}}(g(r)-H(r-1)-x)\le 0\Bigg\}\nonumber\\
&\qquad\qquad\le \inf_{|H|\le 1 } AV@R_{1/k_N}^{\hat{\P}_N}(g(r)-H(r-1))+2k_N \epsilon_N(\beta_N)\nonumber
%
\end{align}
on a set of probability greater or equal than $1-\beta_N$.
\end{Cor}
Note that 
\begin{align*}
\pi_{\hat{\mathcal{Q}}_N}(g) = \sup_{\tilde{\P} \in B^p_{\epsilon_N(\beta_N)}(\hat{\P}_N)} \sup_{\|d\nu/ d\tilde{\P}\|_{\infty} \le k_N} \inf_{H \in \R^d} \E_{\nu}\left[g-H(r-1)\right].
\end{align*}
The proof proceeds by using a min-max argument to interchange the two suprema and the infimum above and uses continuity of $\tilde{\P} \mapsto AV@R_{1/k_N}^{\tilde{\P}}(g-H(r-1))$ w.r.t. to $\Wc^1$, see \cite{pichler2013evaluations} and Appendix \ref{app. 1} for details.

We  provide a method for the direct calculation of the Wasserstein estimator implemented in TensorFlow\footnote{Our Python implementation for all of the numerical examples in the paper can be found at \url{https://github.com/johanneswiesel/Stats-for-superhedging}.}, which is based on recent duality results obtained in \citep{eckstein2018robust}. As this approximation is computationally quite costly when a large sample size is used, we opt to compute the upper bound of $\pi_{\hat{\Qc}_N}(g)$ given in \eqref{eq. avar1} instead. This is shown in Figure \ref{fig. wass1}. 

\begin{figure}
\centering
\begin{minipage}[b]{0.49\textwidth}
  \includegraphics[width=\textwidth]{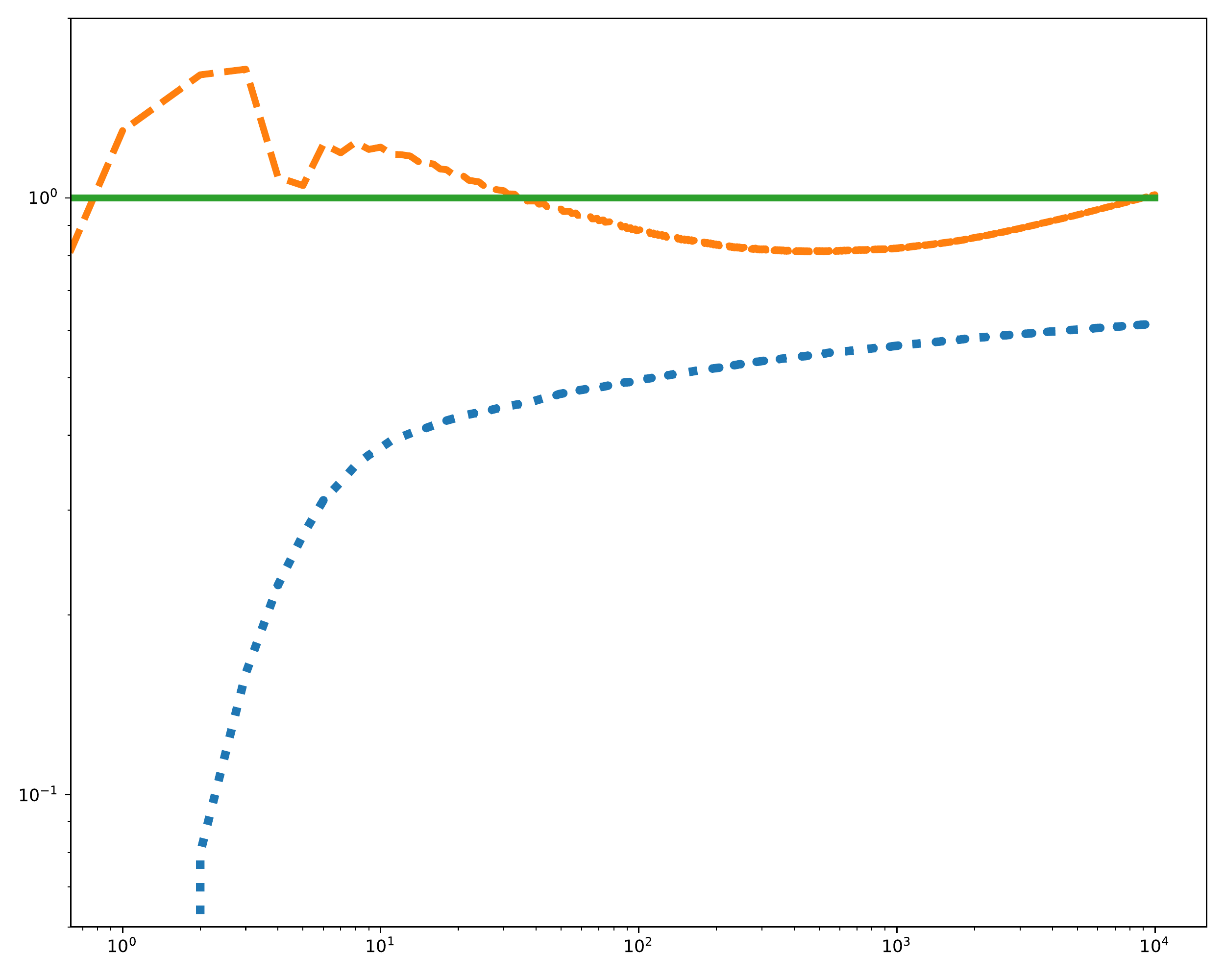}
\end{minipage}
\begin{minipage}[b]{0.49\textwidth}
  \includegraphics[width=\textwidth]{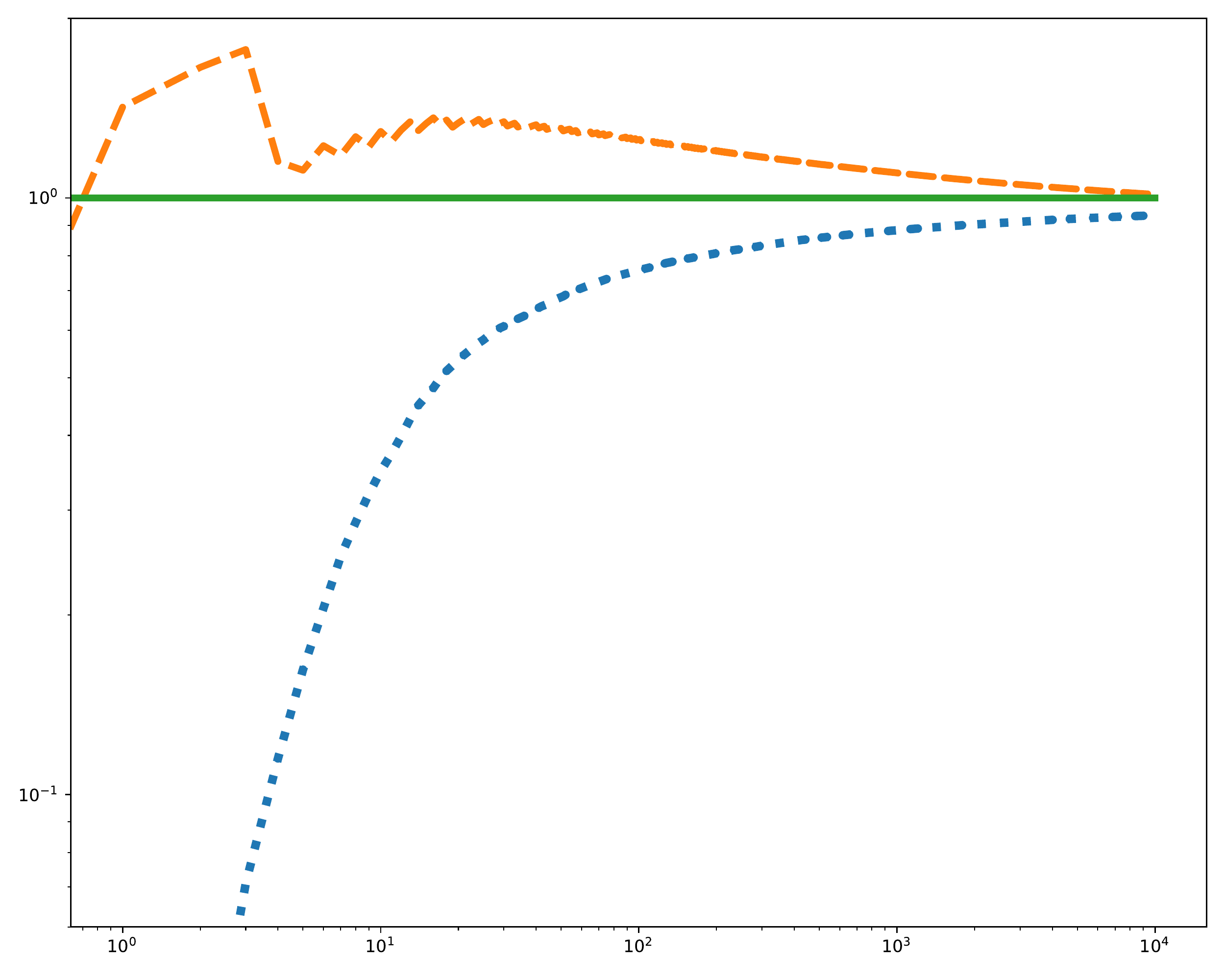}
\end{minipage}
 \captionsetup{width=\linewidth}
\caption{Convergence of the plugin estimator $\hat{\pi}_N$ (dotted) and the Wasserstein estimator $\pi_{\hat{\Qc}_N}$ (dashed) to the true value (solid) as $N\to \infty$ for $g(r)=(1-r)\mathds{1}_{\{r \le 1\}}-\sqrt{r-1}\mathds{1}_{\{r>1\}}$, $\P=\text{Exp}(1)$ (left) and $g(r)=(r-2)^+$, $\P=\exp(\mathcal{N}(0,1))$ (right). Results averaged over $10^3$ runs.}
\label{fig. wass1}
\end{figure}

\subsection{A penalty approach: estimator for discontinuous payoffs}\label{sec:penaltyestimator}
In the previous section we introduced the estimator $\pi_{\hat{\mathcal{Q}}_N}$ where $\hat{\mathcal{Q}}_N$ were based on Wasserstein balls around $\hat{\P}_N$. This estimator allowed us to address fundamental shortcomings of the plugin estimator but, as the counterexamples demonstrated, it is only asymptotically consistent under suitable regularity assumptions on $g$ and/or further assumptions on $\P$. To construct an estimator which would be consistent also for discontinuous payoffs while preserving some of the desirable robustness properties of $\pi_{\hat{\mathcal{Q}}_N}$, it is natural to turn to penalty methods used in risk measures and their representations as non-linear expectations. Namely we use the maximum norm of the Radon-Nikodym derivative, rather than the Wasserstein distance, in the penalisation term.
\begin{Theorem}\label{thm:penalty}
In the setting of Theorem \ref{Thm. one-per}, let NA$(\P)$ hold and let $g: \R^d_+ \to \R_+$ be Borel-measurable and bounded by some constant $C>0$. Then for any $C_N \stackrel{N\to \infty}{\longrightarrow} C$ we have
\begin{align}\label{eq. ball}
\lim_{N \to \infty} \sup_{\Q\in \mathcal{M}} \left(\E_{\Q}[g]-C_N\left(\inf_{\hat{\Q} \sim \hat{\P}_N, \ \hat{\Q}\in \mathcal{M}} \left\| \frac{d\hat{\Q}}{d\Q} \right\| _{\infty}-1\right)\right)=\sup_{\Q \sim \P, \ \Q \in \mathcal{M}} \E_{\Q}[g]
\end{align}
$\P^{\infty}{-a.s.}$, 
where for two probability measures $\Q, \hat{\Q}$ the expression $\left\| \frac{d\hat{\Q}}{d\Q}\right\|_{\infty} =\infty$ if $\hat{\Q}$ is not absolutely continuous w.r.t. $\Q$.
\end{Theorem}
The direct implementation of \eqref{eq. ball} proves numerically expensive and unstable due to the fraction $\|d\hat{\Q}/d\Q\|_{\infty}$ appearing in the penalisation term. Thus, in Figure \ref{fig. hyb1}, we show an upper bound on the penalty estimator derived in the proof of Theorem \ref{thm:penalty} in Appendix \ref{app. 1}. We focus on more tractable properties of the plugin and Wasserstein estimator for the rest of the paper.

\begin{figure}
\centering
\begin{minipage}[b]{0.49\textwidth}
  \includegraphics[width=\textwidth]{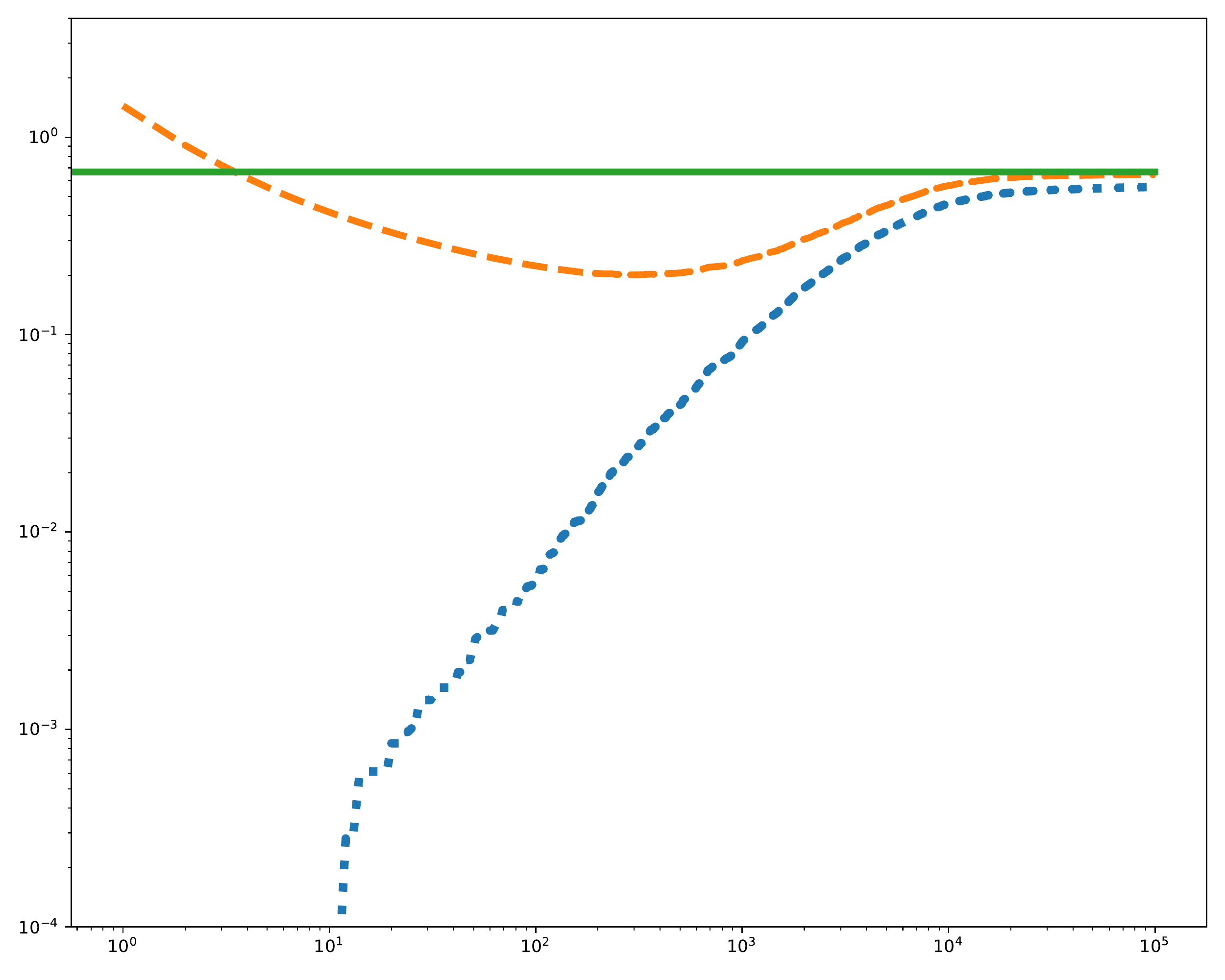}
\end{minipage}
\begin{minipage}[b]{0.49\textwidth}
  \includegraphics[width=\textwidth]{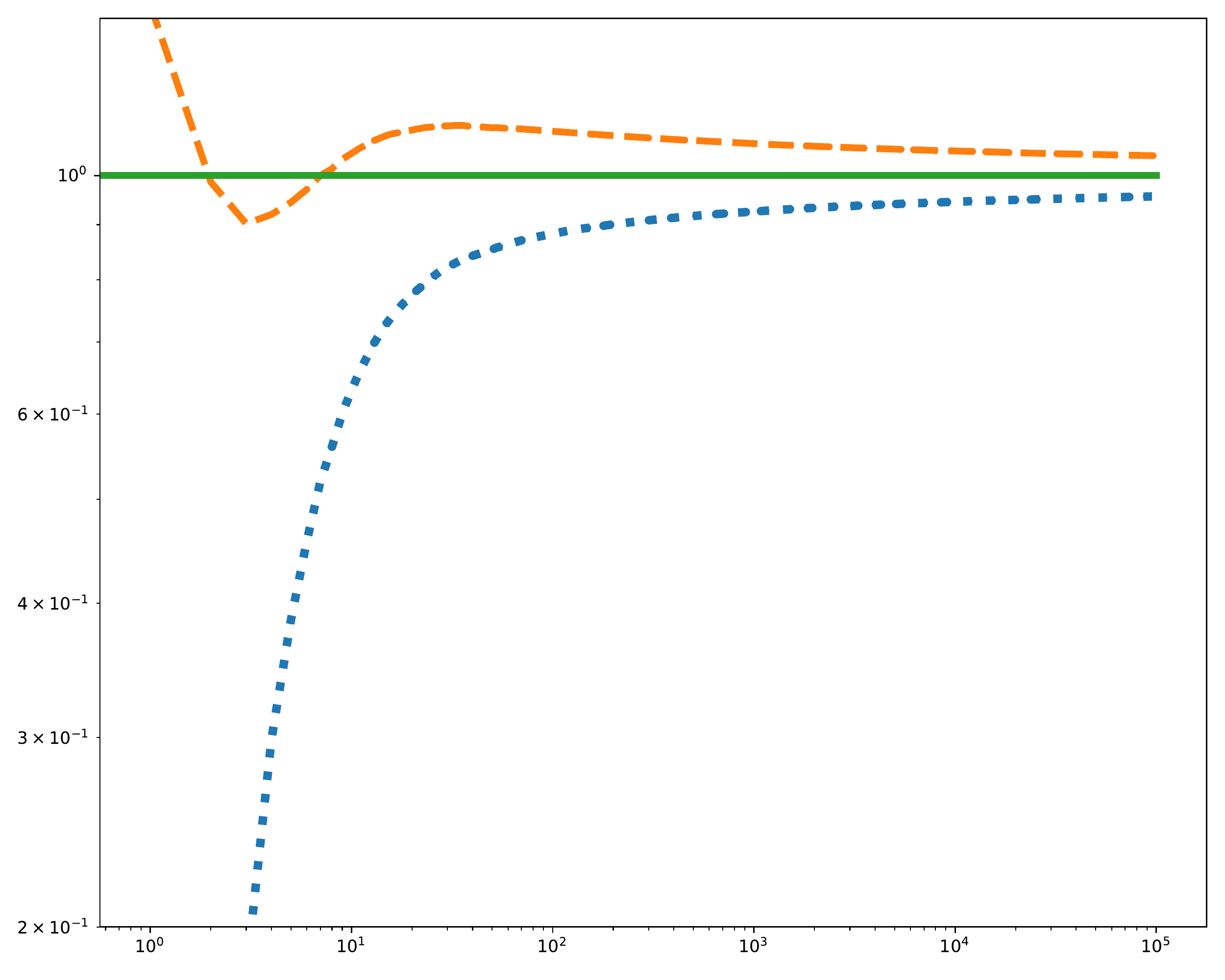}
\end{minipage}
\captionsetup{width=\linewidth}
\caption{
Convergence of the penalty estimator (dashed) in Theorem \ref{thm:penalty} and the plugin estimator $\hat{\pi}_N$ (dotted) to the true value (solid) as $N\to \infty$ for $g(r)=\mathds{1}_{\{r \le 0.5\}}$, $\P=\P^{10}$ from \cite[Example \ref{Ex. 2}]{stats2} (left) and $g(r)=\mathds{1}_{\{r \le 0.5\}}$, $\P=\text{Exp}(1)$ (right). Results averaged over $10^3$ runs.} \label{fig. hyb1}
\end{figure}

\section{Statistical robustness of superhedging price estimators}
\label{sec:robust}

Recall that in Theorem \ref{thm. rob} we showed that classical robustness in the sense of Hampel can not hold unless the superhedging price is trivial in that \eqref{eq:sh trivial} holds. This is closely related to properties of L\'evy-Prokhorov metric, see Lemma \ref{lem. wasser v2}, and we are naturally led to consider stronger metrics than $d_L$, which offer a better control on the support. Below, we investigate the use of Wasserstein distances. First we consider $\Wc^p$ for $p\geq 1$ which is sufficient to establish robustness of the estimator $\pi_{\hat{\Qc}_N}$ from Section \ref{subsec dist}. Then we turn to an even stronger metric $\Wc^\infty$ which is needed to study the plugin estimator.

\subsection{Robustness with respect to the Wasserstein-Hausdorff metric}
\label{sec:robust_wassp} 

Following \cite{li2017generalized} we consider Wasserstein-Hausdorff distance, i.e., a Hausdorff distance between subsets of $\prob(\R_+^d)$ equipped with $\Wc^p$:
\begin{Defn}
Let $\mathfrak{P}, \tilde{\mathfrak{P}} \subseteq \mathcal{P}(\R_+^d)$. The $p$-Wasserstein-Hausdorff metric between sets $\mathfrak{P}$ and $\tilde{\mathfrak{P}}$ is given by
\begin{align*}
\mathcal{W}^p(\mathfrak{P},\tilde{\mathfrak{P}}):= \max \left( \sup_{\P \in \mathfrak{P}} \inf_{\tilde{\P} \in \tilde{\mathfrak{P}}} \mathcal{W}^p(\P, \tilde{\P}), \sup_{\tilde{\P} \in \tilde{\mathfrak{P}}} \inf_{\P \in \mathfrak{P}} \mathcal{W}^p(\P, \tilde{\P}) \right).
\end{align*}
\end{Defn}
In this generality $\mathcal{W}^p(\mathfrak{P}, \tilde{\mathfrak{P}})$ can take the value infinity. Properties of this quantity are discussed in \cite{li2017generalized} assuming compactness and uniform integrability of $\mathfrak{P}$ and $\tilde{\mathfrak{P}}$. We apply this distance to the sets of the form $\hat{\Qc}_N=D^p_{\epsilon_N(\beta_N),k_N}(\hat{\P}_N)$, see \eqref{eq:WassBall}, and we note that $D^1_{\epsilon_N(\beta_N),k_N}(\hat{\P}_N)$ is neither compact nor uniformly integrable, see Lemma \ref{lem. comp}. We used these sets in Section \ref{subsec dist} to define consistent estimators $\pi_{\hat{\Qc}_N}$, see \eqref{eq:generic_estimator} and Theorem \ref{Thm wasserstein}. The following establishes their robustness:
\begin{Theorem} \label{thm rob_was}
Fix $p\geq 1$. The estimator $\pi_{\hat{\Qc}_N}$ studied in Theorem \ref{Thm wasserstein} is robust with respect to the $p$-Wasserstein-Hausdorff metric in the sense that 
\begin{align*}
\sup_{g \in \mathcal{L}_1}\left| \pi_{\hat{\Qc}_N^1}(g)-\pi_{\hat{\Qc}_N^2}(g)
\right| 
\le  \mathcal{W}^p(\hat{\Qc}^1_N,\hat{\Qc}^2_N),
\end{align*}
where $\hat{\Qc}^i_N=D^p_{\epsilon_N(\beta_N),k_N}(\hat{\P}^i_N)$ for $\P^i\in \prob(\R^d_+)$, $i=1,2$.
\end{Theorem}

\begin{proof}
Note that for all $g \in \mathcal{L}_1$ and $\Q^i \in \hat{\Qc}^i_N$, $i=1,2$, we have
\begin{align*}
\E_{\Q^1}[g]  - \E_{\Q^2}[g]=\int_{\R_+^d \times \R_+^d} g(r)-g(s) d\gamma(r,s) \le \int_{\R_+^d \times \R_+^d} |r-s| d\gamma(r,s),
\end{align*}
where $\gamma\in \Pi(\Q^1,\Q^2)$ is a probability measure on $\R_+^d \times \R_+^d$ with marginals $\Q^1$ and $\Q^2$. 
Taking the infimum over all these probability measures $\gamma$ yields
\begin{align*}
\E_{\Q^1}[g]  - \E_{\Q^2}[g] \le \mathcal{W}^p(\Q^1, \Q^2)
\end{align*}
for all $p\ge1$. The claim follows.
\end{proof}

\begin{Rem}
It follows in particular that if $\P^1,\P^2$ admit no arbitrage then \\
\mbox{$\lim_{N \to \infty} \mathcal{W}^p(\hat{\Qc}^1_N,\hat{\Qc}^2_N)=0$} implies $\text{supp}(\P^1)=\text{supp}(\P^2)$. Indeed, otherwise there exists a Lipschitz continuous function $g$ such that
\begin{align*}
\sup_{\Q \sim \P^1, \ \Q \in \mathcal{M}} \E_{\Q}[g] \neq \sup_{\Q \sim \P^2, \ \Q \in \mathcal{M}} \E_{\Q}[g],
\end{align*}
so, by consistency, $\lim_{N \to \infty}\mathcal{W}^p(\hat{\Qc}^1_N,\hat{\Qc}^2_N)>0$, $\P^{\infty}$-a.s.
\end{Rem}

\subsection{Robustness with respect to $\Wc^\infty$ and perturbations of the support}
\label{sec:robust_pathwise}
We reconsider now robustness of the plugin estimator from Section \ref{sec. plugin}. In analogy to the previous section, it seems natural to simply consider the Hausdorff distance between the supports of $\hat{\P}^1_N$ and $\hat{\P}^2_N$. 
In Proposition \ref{prop:d_H cont} we established continuity of $\prob(\R^d_+)\ni \tilde{\P}\to \pi^{\tilde{\P}}(g)$ in the pseudometric $d_H(\text{supp}(\P),\text{supp}(\tilde{\P}))$ but noted that it was not sufficient for a robustness result. Recalling the L\'evy metric in \eqref{eq:LPdistance}, if $d_L(\P, \tilde{\P}) \le \epsilon$ then $\tilde{\P}$ can be obtained from $\P$ by redistributing $\epsilon$ mass to arbitrary points on $\R^d_+$, while $(1- \epsilon)$ mass can only be moved in an $\epsilon$-neighbourhood (in the Euclidean distance) of where $\P$ allocated mass. 
As we have observed before, the former operation causes problems, as it changes the null sets of the measure uncontrollably. This is no longer possible under our pseudometric. However, to obtain robustness, we have to restrict redistribution of mass to an $\epsilon$-neighbourhood for all sets and not only for the whole support. This is achieved by the $\Wc^\infty$ metric as is clear from the second representation in \eqref{eq:Winfty}. This leads to the following extended notion of robustness:
\begin{Defn}\label{Def robust}
Let $\mathfrak{P}\subseteq \mathcal{P}(\R_+^d)$ and $r_1, r_2, \dots$ be i.i.d. with distribution $\P\in \mathfrak{P}$. The sequence of estimators $T_N=T(\hat{\P}_N)$ is said to be robust at $\P \in \mathfrak{P}$ w.r.t.\ $\mathcal{W}^{\infty}$ on $\mathfrak{P}$, if for all $\epsilon >0$ there exist $\delta >0$ and $N_0 \in \N$ such that for all $N \ge N_0$ and all $\tilde\P \in \mathfrak{P}$
\begin{align*}
\mathcal{W}^{\infty}(\tilde\P, \P)\le \delta \hspace{0.5cm} \Longrightarrow \hspace{0.5cm} d_L(\mathcal{L}_{\tilde\P}(T_N), \mathcal{L}_{\P} (T_N)) \le \epsilon.
\end{align*}
\end{Defn}
The following asserts robustness of the plugin estimator in the above sense and is the main result in this section. 
\begin{Theorem}\label{Thm. infi}
Let $\P \in \mathcal{P}(\R_+^d)$ such that NA$(\P)$ holds. Then, for a uniformly continuous $g$, the plugin estimator $\hat{\pi}_N(g)$ is robust at $\P$ w.r.t.\ $\mathcal{W}^{\infty}$ on $\mathcal{P}(\R_+^d)$.
\end{Theorem}
The proof of Theorem \ref{Thm. infi} is reported in \cite[Section \ref{app:robust}]{stats2}. There are ways to weaken the continuity assumption on $g$ and obtain robustness on some $\mathfrak{P}\subseteq \prob(\R^d_+)$, see \cite[Corollary \ref{cor:Winfinity}]{stats2}. We close this section with a result on robustness of $\hat{\pi}^\infty_N(g)$ from Proposition \ref{prop:Wcinfty_est}.
\begin{Cor}\label{cor:Wcinfty_rob}
Let $g$ be a continuous function and $\P\in \prob(\R^d_+)$ satisfying NA$(\P)$ and Assumption \ref{Ass. bilip}. Then, the estimator $\hat{\pi}^\infty_N(g)$ from Proposition \ref{prop:Wcinfty_est} is robust at $\P$ w.r.t.\ $\mathcal{W}^{\infty}$ on $\mathcal{P}(\R_+^d)$.
\end{Cor}

 \section{Risk measurement estimation}\label{sec:riskmeasures}
The $\P$-a.s. superhedging price $\pi^{\P}(g)$ is a very conservative assessment of risk of a short position in a liability with payoff $g$. Instead, we could use a risk measure $\rho_{\P}$ for such an assessment, as proposed by \cite{cheridito2017duality}, leading to 
\begin{align*}
\pi^{\rho_{\P}}(g):=\inf\{x \in \R \ | \ \exists H \in \R^d \text{ such that }\rho_{\P}(g-x-H(r-1)) \le 0 \}.
\end{align*}
Note that we include above the ability to trade in the market in order to (optimally) reduce the risk of $g$. We  consider $\rho_{\P}$ with Kusuoka's representation
\begin{align}\label{eq:kusuoka}
\rho_{\P}(g)=\sup_{\mu\in \mathfrak{P}}\int_0^1 \text{AV@R}_{\alpha}^{\P}(g)d\mu(\alpha),
\end{align}
for a set $\mathfrak{P}$ of probability measures on $[0,1]$. This is not very restrictive since this representation, first obtained in \cite{kusuoka2001law}, holds for any law invariant coherent risk measure, see \cite{jouini2006law}. Importantly, it enables us to think of $\rho_{\P}(g)$ as a function of the underlying measure $\P$. Much like we did for $\pi^{\P}(g)$, we would like to estimate $\pi^{\rho_{\P}}(g)$ directly from the observed stock returns. To this end we introduce the following estimator
\begin{align*}
\pi^{\rho}_{B_{\epsilon_N(\beta_N)}^p(\hat{\P}_N)}(g):=\inf \Big\{x \in \R \ | \ &\exists H \in \R^d \text{ such that }\\
&\sup_{\tilde{\P} \in B_{\epsilon_N(\beta_N)}^p(\hat{\P}_N)}\rho_{\tilde{\P}}(g-x-H(r-1)) \le 0 \Big\}
\end{align*}
as the natural equivalent to $\pi_{\hat{\mathcal{Q}}_N}(g)$. In particular, if $\mathfrak{P}=\{\delta_{\alpha}\}$ where $\alpha\in[0,1]$, we simply have $\rho_{\P}(g)=$ AV@R$_{\alpha}^{\P}(g)$ and the corresponding estimator $\pi^{\rho}_{B_{\epsilon_N(\beta_N)}^p(\hat{\P}_N)}(g)$ resembles the Wasserstein $\Wc^p$ estimator for fixed level $1/k_N:=\alpha$.
We have the following consistency result:
\begin{Prop}
Assume $g$ satisfies $|g(r)-g(\tilde{r})| \le L_{\gamma}|r -\tilde{r}|^{\gamma}$ for some $\gamma \le 1$ and $L_{\gamma} \in \R$ and that $\sup_{\mu \in \mathfrak{P}}\int_0^1 \mu(d\alpha)/\alpha^{\gamma/p}<\infty$. Then for any $\P$ satisfying NA$(\P)$ and Assumption \ref{Ass 1} the limit in $\P^\infty$-probability
\begin{align*}
\lim_{n \to \infty}\pi^{\rho}_{B_{\epsilon_N(\beta_N)}^p(\hat{\P}_N)}(g)= \pi^{\rho_{\P}}(g) 
\end{align*}
holds. If Assumption \ref{Ass 1}.\ref{Ass 1.2} is satisfied, then the limit also holds $\P^\infty$-a.s.
\end{Prop}
\begin{proof}
The ``$\ge$"-inequality follows in the proof of Theorem \ref{Thm wasserstein}. We now prove the opposite inequality using \citep{pichler2013evaluations}[Corollary 11, p.538].
Fix $\epsilon>0$. Note that there exists $H \in \R^d$ such that $\rho_{\P}(g- \pi^{\rho_{\P}}(g)-\epsilon-H(r-1)) \le 0$. Then for all $\tilde{\P} \in B^p_{\epsilon_N(\beta_N)}(\hat{\P}_N)$
\begin{align*}
\pi^{\rho_{\tilde{\P}}}(g)&\le\rho_{\tilde{\P}}(g-H(r-1))
=\rho_{\tilde{\P}}(g-\pi^{\rho_{\P}}(g)-\epsilon-H(r-1))+\pi^{\rho_{\P}}(g)+\epsilon \\
&=[\rho_{\tilde{\P}}(g-\pi^{\rho_{\P}}(g)-\epsilon-H(r-1))-\rho_{\P}(g-\pi^{\rho_{\P}}(g)-\epsilon-H(r-1))]\\
&+\rho_{\P}(g-\pi^{\rho_{\P}}(g)-\epsilon-H(r-1))
+\pi^{\rho_{\P}}(g)+\epsilon \\
&\le \tilde{L}_{\gamma}\mathcal{W}^p(\tilde{\P}, \P)^{\gamma} \sup_{\mu \in \mathfrak{P}}\int_0^1 \frac{\mu(d\alpha)}{\alpha^{\gamma/p}}  +\pi^{\rho_{\P}}(g)+\epsilon.
\end{align*}
As $\epsilon>0$ was arbitrary the claim follows.
\end{proof}
Finally, we present a simple empirical test of the performance of our estimators. 
We simulate weekly returns according to a $\mathrm{GARCH}(1,1)$ model:
\begin{align*}
r_n =  \sqrt{\frac{\mu-2}{\mu}}\eta_n\sqrt{h_n},\quad 
h_n = \omega + \beta h_{n-1}+\alpha r_{n-1}^2,
\end{align*}
where $\omega=0.02, \beta=0.1, \alpha =0.8$
and $\eta_n$ is standard student-t distributed with $\mu=5$ degrees of freedom.
 We take $1000$ samples from the above $\P\sim \mathrm{GARCH}(1,1)$ and calculate the plugin estimator $\pi^{\mathrm{AV@R}_{0.95}}_{\hat{\P}_N}((r-1)^+)$ and the Wasserstein  estimator $\pi^{\mathrm{AV@R}_{0.95}}_{B_{\epsilon_N(\beta_N)}^p(\hat{\P}_N)}((r-1)^+)$. We compare this to a parametric estimator of $\pi^{\mathrm{AV@R}_{0.95}^\P}((r-1)^+)$, where we first estimate the parameters of the GARCH(1,1) model and then compute $\pi^{\mathrm{AV@R}_{0.95}^{\tilde{\P}}}((r-1)^+)$ given the estimated model $\tilde{\P}$. For each of these estimates, we use a running window of 50 weeks, which is in line with the Basel III regulations for calculating the 10-day AV@R (see \cite[MAR33, p.89]{basel2019minimum}), which set the minimum length of the historical observation period to be one year. We also consider the case when the parameters of the model change for observations $330-670$. The behaviour of the three estimators is shown in Figure \ref{fig. avar}. Both the Wasserstein and plugin estimator approximate the true value reasonably well -- the Wasserstein estimator being the most conservative estimator. The parametric estimator exhibits the most erratic behaviour which is due to the unstable estimation of the GARCH(1,1) parameters with only $50$ data points. This shows advantages of our proposed estimators when compared with a parametric approach, even when the model is correctly selected. The last pane in Figure \ref{fig. avar} uses S\&P500 weekly returns data from 01/01/2006- 01/01/2015 with a moving window of $50$ weeks. The GARCH(1,1) estimator, for which the model is mis-specified, does not pick up any markets movements, while both the Plugin and Wasserstein estimator detect the financial crisis and its aftermath. For similar plots but with GARCH(1,1) using log-returns, we refer to \cite[Section \ref{app:riskmeasures}]{stats2}. While preliminary, we believe that this simple empirical study points to clear advantages of our approach. In particular, it is encouraging to see that in the last pane, the Wasserstein estimators clearly picks up the financial crisis and the Eurozone debt crisis periods. A further in-depth study of comparative performance of different estimators is clearly needed and is left for future research.
\begin{figure}
\centering
  \includegraphics[width=\textwidth]{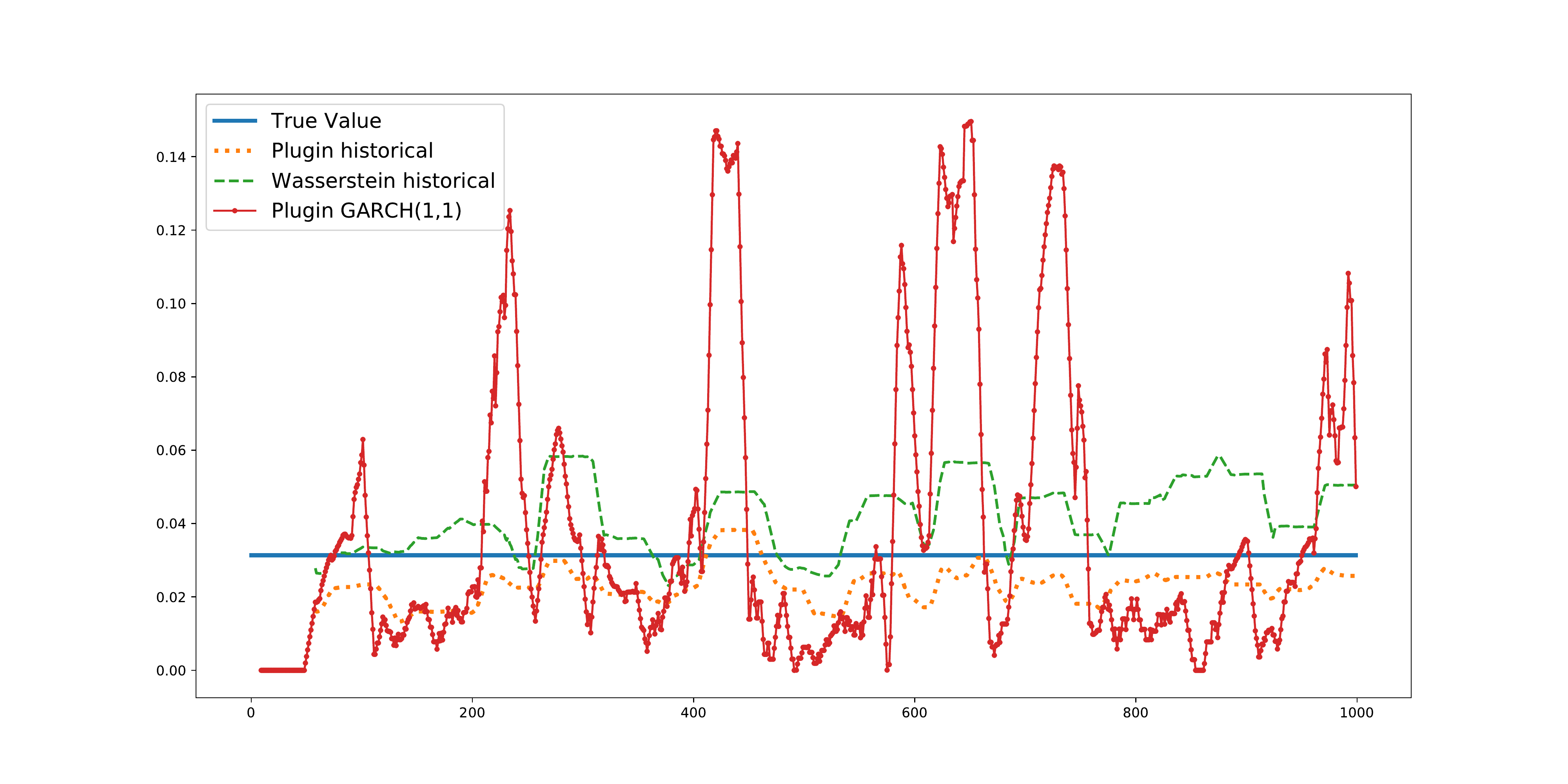}
  \includegraphics[width=\textwidth]{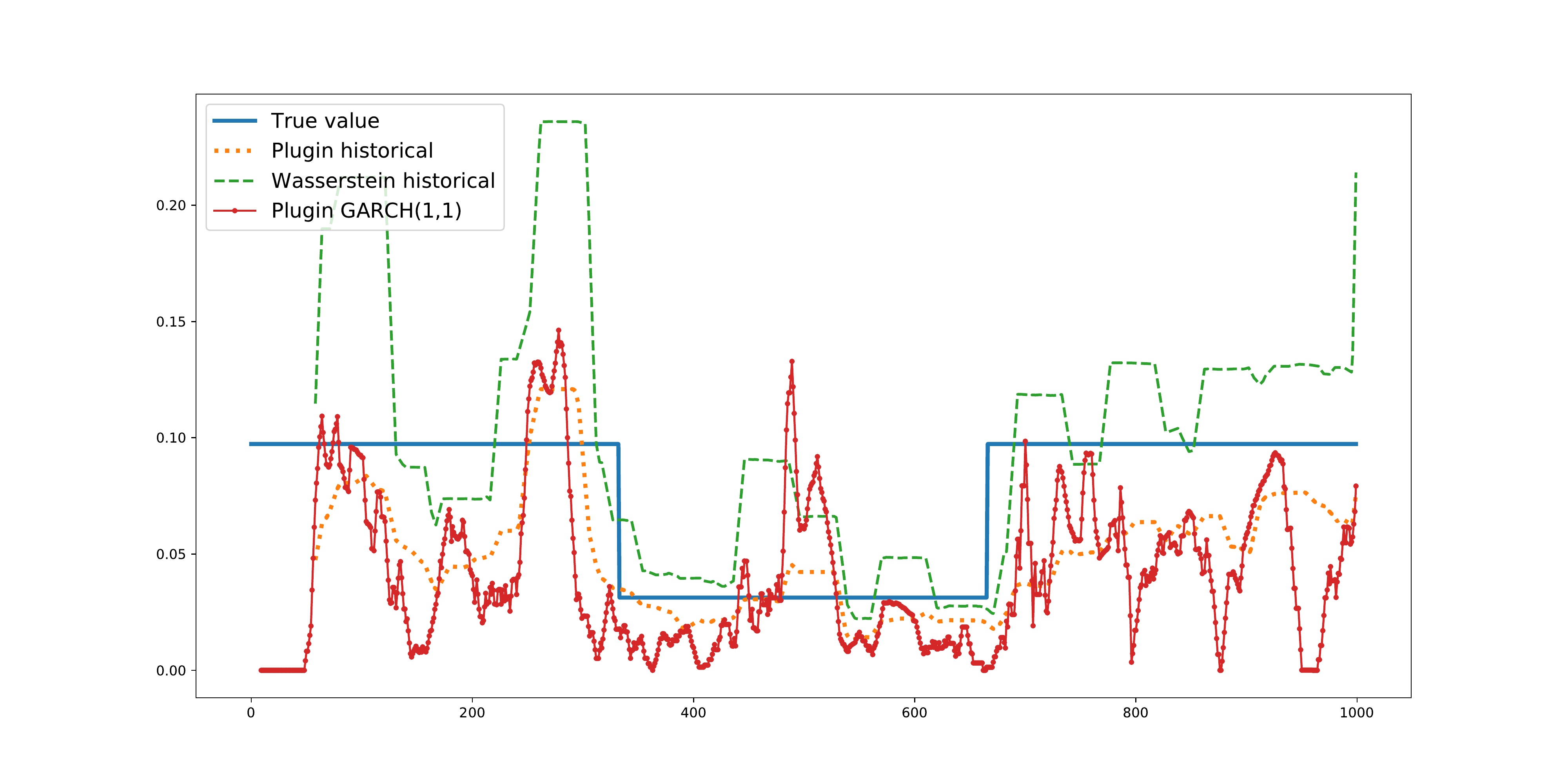}
  \includegraphics[width=\textwidth]{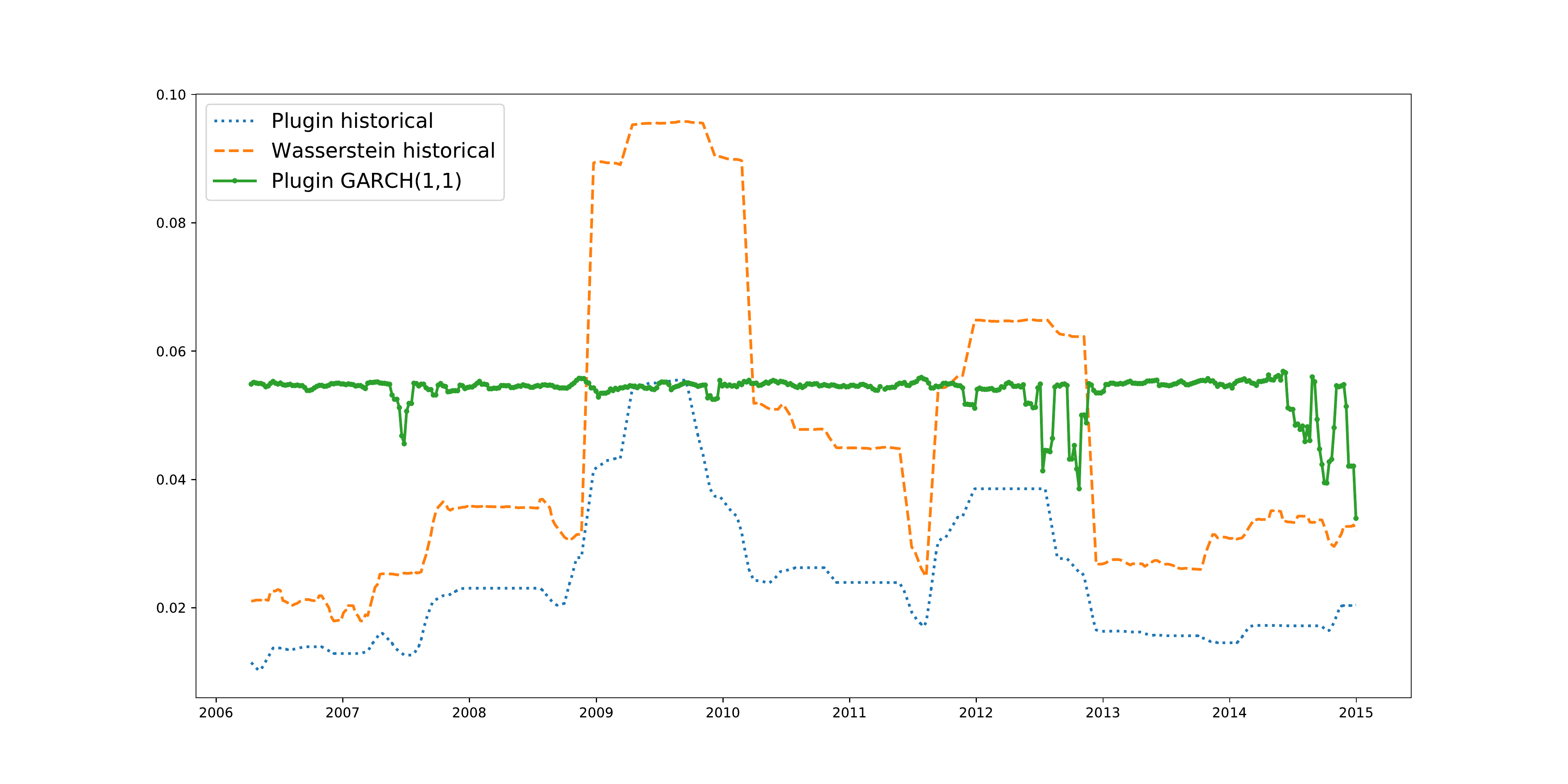}
\caption{Comparison of estimates for $\pi^{\mathrm{AV@R}_{0.95}^{\tilde{\P}}}((r-1)^+)$. Estimates use a rolling window of $50$ data points and we plot the average of the last $10$ (first two panes) or $5$ (last pane) estimates. The data is from $\P\sim \mathrm{GARCH}(1,1)$ (first pane) and its variant with a change in the parameters for the middle third of the time series. The last pane uses S\&P500 returns.} \label{fig. avar}
\end{figure}

\section{Convergence of Superhedging strategies} \label{sec. strat}
Given the consistency results of Sections \ref{sec. plugin} \& \ref{sec. impr} 
establishing convergence of superhedging prices we now turn to the convergence of the corresponding superhedging strategies. We start with the case of the plugin estimator.

\begin{Theorem}\label{Thm. strat}
Consider $\P\in \mathcal{P}(\R_+^d)$ satisfying NA$(\P)$. Suppose\\ $\overline{\bigcup_{N\geq 1}\text{supp}(\hat{\P}_N)}=\text{supp}(\P)$. Then for any sequence of trading strategies $(H_N)_{N \in \N}$ satisfying
\begin{align*}
\hat{\pi}_N(g)+H_N(r-1) \ge g(r) \quad \hat{\P}_N\text{-a.s.}
\end{align*}
there exists a subsequence $(H_{N_k})_{k \in \N}$ converging to some $H\in \R^d$ which satisfies 
\begin{align*}
\pi^{\P}(g)+H(r-1) \ge g(r) \quad \P\text{-a.s.}
\end{align*}
\end{Theorem}

\begin{proof}
Note that we can assume without loss of generality that $H_N \in \text{lin}(\text{supp}(\P)-1)$ for all $N \in \N$. NA$(\P)$ implies that the sequence $(H_N)_{N \in \N}$ is bounded. Indeed, assume towards a contradiction that $H_N \to \infty$. 
Clearly $H_N/|H_N| \to \tilde{H} \in \R^d$ with $|\tilde{H}|=1$. 
Also since
\begin{align*}
\hat{\pi}_{N}(g)+H_N(r-1) -g(r) \ge 0 \quad\hat{\P}_N\text{-a.s.}
\end{align*}
we conclude $\tilde{H}(r-1) \ge 0$ on $\text{supp}(\P)$. By NA$(\P)$ $\tilde{H}(r-1) =0$ follows $\P$-a.s. and as $\tilde{H} \in \text{lin}(\text{supp}(\P)-1))$ we have $\tilde{H}=0$, which contradicts $|\tilde{H}|=1$. This shows that $(H_N)_{N \in \N}$ is bounded. Thus there exists a subsequence $(H_{N_k})_{k \in \N}$ of $(H_N)_{N \in \N}$ such that $H_{N_k} \to H$. Lastly we note that for each $r \in \{r_1, \dots, \}$ we have
\begin{align*}
\pi^{\P}(g)+H(r-1)=\lim_{k \to \infty} (\hat{\pi}_{N_k}(g)+H_{N_k}(r-1))\ge g(r).
\end{align*}
Thus the claim follows for continuous $g$. As in the proof of Theorem \ref{Thm. one-per} we can then extend the result to general $g$ using Lusin's theorem. This concludes the proof.
\end{proof}

\begin{Rem}
The above claim remains true if we replace $\hat{\pi}_N$ by any consistent estimator which dominates $\hat{\pi}_N$. In particular it is valid for the $\mathcal{W}^{\infty}$-estimator $\hat{\pi}_N^\infty$ introduced in Proposition \ref{prop:Wcinfty_est} and the penalty estimator introduced in Theorem \ref{thm:penalty}.
\end{Rem}
In general one cannot replace the claim $H_{N_k} \to H$ in Theorem \ref{Thm. strat} by $H_N \to H \in \R^d$ as the following example shows.
\begin{Ex}
Take $d=1$ and 
$$\hat{\P}_N= \frac{1}{1-2^N}\sum_{k=1}^N 2^{-k} \delta_{1-(-2)^{-k+1}}.$$ 
Let $g(r)=1-|r-1|$. Note that $\P^N \Rightarrow \P:=\sum_{k=1}^{\infty} 2^{-k} \delta_{1-(-2)^{-k+1}}$ and $$\pi^{\P_N}(g)=1-\frac{2}{3}2^{-N+2} \uparrow \pi^{\P}(g)=1.$$ Then the sequence of superhedging strategies $(H_N)_{N \in \N}$ is uniquely determined as the slope of the line through the points $(1-(-2)^{-N+2},g(1-(-2)^{-N+2})), \  (1-(-2)^{-N+1},g(1-(-2)^{-N+1}))$, i.e. 
\begin{align*}
H_N&= \frac{|1-(-2)^{-N+2}-1|-|1-(-2)^{-N+1}-1|}{1-(-2)^{-N+1}-1+(-2)^{-N+2}}=\frac{2^{-N+2}-2^{-N+1}}{-(-2)^{-N+1}+(-2)^{-N+2}}\\
&=\frac{1}{3(-1)^{N+2}}
=\frac{(-1)^{N}}{3},
\end{align*}
which diverges.
\end{Ex}

Next we establish a corresponding result for the Wasserstein estimator $\pi_{\hat{\Qc}_N}(g)$. Recall the definition of AV@R given in Section \ref{sec. impr}.

\begin{Theorem}
Let $g$ be Lipschitz continuous and bounded from below. Also let $\mathcal{W}^1(\hat{\P}_N,\P) \le \epsilon_N(\beta_N)$ for large $N \in \N$, $k_N= o(1/\epsilon_N(\beta_N))$ and $\lim_{N \to \infty}\pi_{\hat{\Qc}_N}(g)= \pi^{\P}(g)$. Then for every sequence of trading strategies $(H_N)_{N \in \N}$ satisfying
\begin{align}\label{eq. avar}
AV@R_{1/k_N}^{\hat{\P}_N}\left(g(r)-H_N(r-1)-\pi_{\hat{\Qc}_N}(g)-1/N \right) \le 0
\end{align}
there exists a subsequence $(H_{N_k})_{k \in \N}$ converging to some $H\in \R^d$ which satisfies 
\begin{align*}
 \pi^{\P}(g)+H(r-1) \ge g(r) \quad \P\text{-a.s.}
\end{align*}
\end{Theorem}

\begin{proof}
As in Corollary \ref{Lem:AVaR}, we see that 
\begin{align*}
\pi_{\hat{\Qc}_N}(g) \ge 
\inf \left\{x \in \R \ | \ \exists H \in \R^d \text{ s.t. } AV@R_{1/k_N}^{\hat{\P}_N}(g-H(r-1)-x) \le 0\right\}.
\end{align*}
Recall $\lim_{N \to \infty}\pi_{\hat{\Qc}_N}(g)= \pi^\P(g)$. As $g$ is Lipschitz continuous Lemma \ref{Lem. Pichler} shows that 
\begin{align}\label{eq. pichler}
&\left|AV@R_{1/k_N}^{\hat{\P}_N}\left( g(r)-H(r-1)-\pi_{\hat{\Qc}_N}(g)\right)-AV@R_{1/k_N}^{\P}\left(g(r)-H(r-1)-\pi_{\hat{\Qc}_N}(g)\right)\right|\\
& \le \tilde{C}k_N\mathcal{W}^1(\hat{\P}_N, \P)\le \tilde{C}k_N\epsilon_N(\beta_N)\nonumber
\end{align}
for some $\tilde{C}>0$. Note that similarly to the arguments in the proof of Theorem \ref{Thm. strat} we can assume that $(H_N)_{N \in \N}$ is bounded, so a subsequence of $(H_N)_{N \in \N}$ (which we also denote by $(H_N)_{N \in \N}$ for convenience) converges to some $H \in \R^d$. Fix $N \in \N$. Then for all $M \ge N$ we have
\begin{align*}
&\sup_{\|d\tilde{\P}/d\hat{\P}_N\|_{\infty} \le k_N} \E_{\tilde{\P}}[g(r)-H_M(r-1)-\pi_{\hat{\Qc}_M}(g)-1/M] \\
\le\  &2\tilde{C}k_N\epsilon_N(\beta_N)+\sup_{\|d\tilde{\P}/d\hat{\P}_M\|_{\infty} \le k_N} \E_{\tilde{\P}}[g(r)-H_M(r-1)-\pi_{\hat{\Qc}_M}(g)-1/M]\\
\le\ &2\tilde{C}k_N\epsilon_N(\beta_N)+\sup_{\|d\tilde{\P}/d\hat{\P}_M\|_{\infty} \le k_M} \E_{\tilde{\P}}[g(r)-H_M(r-1)-\pi_{\hat{\Qc}_M}(g)-1/M]
\\\le\ &2k_N\tilde{C}\epsilon_N(\beta_N).
\end{align*}
Let $\epsilon>0$. Then
\begin{align} \label{eq. opti}
&\Big| \sup_{\|d\tilde{\P}/d\hat{\P}_N\|_{\infty} \le k_N} \E_{\tilde{\P}} \left[g(r)-H_M(r-1)-\pi_{\hat{\Qc}_M}(g)-1/M\right]\\
&\quad -\sup_{\|d\tilde{\P}/d\P\|_{\infty} \le k_N} \E_{\tilde{\P}}\left[g(r)-H(r-1)-\pi^{\P}(g)\right] \Big|\nonumber\\
&\le \tilde{C} k_N\mathcal{W}^1(\hat{\P}_N, \P)
+\Big| \sup_{\|d\tilde{\P}/d\P\|_{\infty} \le k_N} \E_{\tilde{\P}} \left[g(r)-H_M(r-1)-\pi_{\hat{\Qc}_M}(g)-1/M\right]\nonumber\\
&\quad -\sup_{\|d\tilde{\P}/d\P\|_{\infty} \le k_N} \E_{\tilde{\P}}\left[g(r)-H(r-1)-\pi^{\P}(g)\right] \Big|\nonumber
\le\epsilon\nonumber
\end{align}
by \eqref{eq. pichler} and $N \le M$ large enough since we have
\begin{align*}
&\sup_{\|d\tilde{\P}/d\P\|_{\infty} \le k_N}\left| \E_{\tilde{\P}} \left[g(r)-H_M(r-1)-\pi_{\hat{\Qc}_M}(g)-1/M\right]-\E_{\tilde{\P}} \left[g(r)-H(r-1)-\pi^{\P}(g)\right]\right| \\
\le\ &\sup_{\|d\tilde{\P}/d\P\|_{\infty} \le k_N} |H_M-H| |\E_{\tilde{\P}}[r-1]|+|\pi^{\P}(g)-\pi_{\hat{\Qc}_M}(g)-1/M| \to 0 \quad (M \to \infty).
\end{align*}
As $\epsilon >0$ was arbitrary, this implies  
\begin{align*}
\sup_{\|d\tilde{\P}/d\P\|_{\infty} < \infty} \E_{\tilde{\P}}\left[g(r)-H(r-1)-\pi^{\P}(g)\right] \le 0.
\end{align*}
This concludes the proof.
\end{proof}

\section{Multiperiod results}\label{sec mult}
In this section we partly extend results from Sections \ref{sec. plugin} and \ref{sec. impr} to the case $T>1$. As before we assume $g: (\R_+^d)^{T} \to \R$ is Borel. Let $\mathbb{F}$ be generated by the coordinate mappings $\textbf{r}_t(x)=x_t$, $x\in (\R_+^d)^{T}$.
We write $\rr_{i:j}$ for the vector $(\rr_{i},\ldots,\rr_j)$, $1\leq i<j\leq T$ and $\rr=\rr_{1:T}$. The martingale measures $\mathcal{M}^T$ are now defined via $$\mathcal{M}^T= \{\Q \in \mathcal{P}((\R_+^d)^T) \ | \ \E_{\Q}[\textbf{r}_t|\mathcal{F}_{t-1}]=1 \ \text{ for all }t=1, \dots, T\}.$$ We only consider $\R_+^d$-valued i.i.d. observations $r_1, \dots,r_N$ here\footnote{ The more general Markovian case is likely to require a genuinely novel approach, including using a modified definition of the empirical measure, see \cite{BBBW:20}. We leave it for future research.}. We now concatenate concave envelopes via the following procedure:
\begin{Defn} We define
\begin{align*}
g(\rr_{1:t}, \cdot): \R_+^d \to \R, \ \ \tilde{\rr} \mapsto g(\rr_{1:t}, \tilde{\rr}).
\end{align*} 
Then we recursively set for $\Omega \subseteq (\R_+^d)$
\begin{align*}
g^{\Omega}_{T,T}(\rr) &:= g(\rr) \\
g^{\Omega}_{t,T}(\rr_{1:t}) &:= (\widehat{g_{t+1}^{\Omega}}(\rr_{1:t}, \cdot))_{\Omega} (1), \qquad t=1, \dots, T-1,\\
g^{\Omega}_{0,T}&:= (\widehat{g_{1}^{\Omega}}(\cdot))_{\Omega}(1).
\end{align*}
In analogy with the one-period case we set $ \pi^{\hat{\P}_N}_T(g):=g_{0,T}^{\{r_1, \dots, r_N\}}$ for $r_i\in \R_+^d$, $i=1, \dots, N$.
\end{Defn}

\begin{Defn}For $\P \in \mathcal{P}(\R_{+}^d)$ we define
\begin{align*}
g^{\P}_{T,T}(\rr) &:= g(\rr) \\
g^{\P}_{t,T}(\rr_{1:t}) &:= (\widehat{g_{t+1}^{\P}}(\rr_{1:t}, \cdot))_{\P} (1), \hspace{0.5cm} t=1, \dots, T-1,\\
g^{\P}_{0,T}&:= (\widehat{g_{1}^{\P}}(\cdot))_{\P}(1)
\end{align*}
and define $\P^{\otimes T}= \underbrace{\P \otimes \cdots \otimes \P}_{\text{T times}}$. Furthermore we set $ \pi^{\P}_T(g):=g_{0,T}^{\P}$.
\end{Defn}
We quote the following result:

\begin{Lem}[ref. \cite{follmer2011stochastic}[Theorem 1.31, p.19]
Let $g: (\R_+^d)^T \to \R$ be Borel-measurable. Then under NA$(\P)$
\begin{align*}
g_{0,T}^{\P}&=\sup_{\Q \sim \P^{\otimes T}} \E_{\Q}[g(\rr)]\\
&=\inf\left\{x \in \R\ \bigg| \ \exists H \in \mathcal{H}(\mathbb{F})\text{ s.t. } x+\sum_{s=1}^TH_s(\rr_s-1)\ge g(\rr) \ \P^{\otimes T}\text{-a.s.} \right\}.
\end{align*}

\end{Lem}
We can now formulate the following multiperiod analogue of Theorem \ref{Thm. one-per}:

\begin{Theorem}
Let  $\P_1, \P$ be probability measures on $\R_+^d$ and $g: (\R_+^d)^T \to \R$ be Borel-measurable. Assume that the observations $r_1, r_2, \dots$ are  i.i.d. samples of $\P$. Then $\P^\infty$-a.s.
\begin{align*}
\lim_{N \to \infty} \pi_T^{\hat{\P}_N} (g)&=\lim_{N \to \infty} g_{0,T}^{\{r_1, \dots, r_N\}}= g_{0,T}^{\P}= \pi^{\P}_T(g).
\end{align*}
\end{Theorem}

\begin{proof}
We prove the claim by induction over $T\in \N$. The case $T=1$ follows from Theorem \ref{Thm. one-per}. Thus we assume that we have shown that for each $\rr_1 \in \R_+^d$ $$\lim_{N \to \infty}  g_{1,2}^{\{r_1, \dots, r_N\}}(\rr_1)=g_{1,2}^{\P}(\rr_1), \quad \P^{\infty}\text{-a.s.}$$ 
By Lusin's theorem there exists a sequence of  compact sets $K_n \subseteq \text{supp}(\P)$ such that $\P(\R_+^d\setminus K_n) \le 1/n$ and $g_{1,2}^{\P}|_{K_n}$ is continuous. As in the proof of Theorem \ref{Thm. one-per} we have 
$$g_{0,2}^{\P}= \widehat{(g_{1,2}^{\P}(\cdot))}_{\P}(1)= \lim_{n \to \infty} \widehat{(g^{\P}_{1,2}(\cdot))}_{K_n}(1).$$
As the concave envelope of pointwise increasing functions is increasing, we conclude that $ g_{1,2}^{\{r_1, \dots, r_N\}}(1)$ is increasing in $N \in \N$. Thus $\P^{\infty}$-a.s.
\begin{align*}
\widehat{\left(g_{1,2}^{\P}(\cdot)\right)}_{K_n}(1)=\widehat{\left(\lim_{N \to \infty} g_{1,2}^{\{r_1, \dots, r_N\}}(\cdot)\right)}_{K_n}(1)
&= \lim_{N \to \infty} \widehat{\left( g_{1,2}^{\{r_1, \dots, r_N\}}(\cdot)\right)}_{K_n}(1).
\end{align*}
Interchanging limits (as they are suprema) yields
\begin{align*}
g_{0,2}^{\P}&=\lim_{N \to \infty} \lim_{n \to \infty} \widehat{\left( g_{1,2}^{\{r_1, \dots, r_N\}}(\cdot)\right)}_{K_n}(1)\\
&=\lim_{N \to \infty} \lim_{n \to \infty} \lim_{M \to \infty}
\widehat{\left( g_{1,2}^{\{r_1, \dots, r_N\}}(\cdot)\right)}_{K_n \cap \{r_1, \dots, r_M\}}(1) \\
&= \lim_{N \to \infty} \lim_{M \to \infty}
\widehat{\left( g_{1,2}^{\{r_1, \dots, r_N\}}(\cdot)\right)}_{ \{r_1, \dots, r_M\}}(1) 
=\lim_{N \to \infty} g_{0,2}^{\{r_1, \dots, r_N\}}.
\end{align*}
This concludes the proof for $T=2$. The general induction step follows analogously.
\end{proof}
A similar reasoning can be applied for the estimator  $\pi_{\hat{\Qc}_N}(g)$ from Theorem \ref{Thm wasserstein} with $\sum_{N=1}^\infty \beta_N<\infty$.
\begin{Cor}
Let  $\P$ be a probability measure on $\R_+^d$ and $g: (\R_+^d)^T \to \R$ be $1$-Lipschitz and bounded from below. Assume that the observations $r_1, r_2, \dots$ are i.i.d.\ with respect to $\P$. Then under NA$(\P)$
\begin{align*}
&\lim_{N \to \infty} \sup_{q_1, \dots, q_T }\int_{(\R_+^d)^T}g(\rr)q_T(\rr_{1:T-1}; d\rr_T)\ldots q_2(\rr_1;d\rr_2)q_1(d\rr_1)\\
=\ &\sup_{\Q\in \Mc^T,\ \Q \sim \P ^{\otimes T}} \E_{\Q}[g(\rr)],\qquad \P^\infty\text{-a.s.,}
\end{align*}
where $(\rr_1, \dots, \rr_t) \mapsto q_{t+1}(\rr_1, \dots, \rr_t;\cdot)$ are Borel measurable mappings from $(\R_+^d)^t$ to $\hat{\Qc}_N$, $t=0,\ldots,T-1$.
\end{Cor}

\begin{proof}
We show the claim by backwards induction. Fix $\rr_{1:T-1} \in (\R_+^d)^{T-1}$. 
As in the proof of Theorem \ref{Thm wasserstein}, we have $\P\in B^p_{\epsilon_N}(\hat{\P}_N)$ for $N$ large enough $\P^\infty$-a.s. 
The ``$\ge"$-inequality follows as in the proof of Theorem $\ref{Thm wasserstein}$. Indeed,
\begin{align*}
&\sup_{q_T \in \hat{\mathcal{Q}}_N}\int_{\R_+^d}g(\rr)q_T(\rr_{1:T-1}; d\rr_T)\ldots q_2(\rr_1;d\rr_2)q_1(d\rr_1)\\
 &\qquad \qquad \ge \sup_{q_T\in \Mc, \ \|dq_T/d\P\|_{\infty} \le k_N}\int_{\R_+^d}g(\rr) q_T(\rr_{1:T-1}; d\rr_T)\ldots q_2(\rr_1;d\rr_2)q_1(d\rr_1).
\end{align*} 
As $g$ is bounded from below, passing to the limit with $N\to\infty$ gives the result. 
Now we show the``$\le$"-inequality. This follows directly from the fact that $\rr_T \mapsto g(\rr_{1:T-1}, \rr_T)$ is $1$-Lipschitz, so by the proof of Theorem \ref{Thm wasserstein}
\begin{align*}
&\sup_{q_T\in \hat{\mathcal{Q}}_N} \int_{\R_+^d}g(\rr)q^N_T(\rr_{1:T-1};d\rr_T) -\sup_{q_T\in \mathcal{M},\ q_T\sim \P } \int_{\R_+^d} g(\rr)dq_{T}(\rr_{1:T-1}; d\rr_T)\\
\le\ & 2k_N \epsilon_N(\beta_N).
\end{align*}
For the induction step we note that for some set $C \subseteq \mathcal{P}(\R_+^d)$, $t \in {1, \dots T}$ and some $1$-Lipschitz function $f: (\R_+^d)^{t} \to \R$ we have
\begin{align*}
&\left|\sup_{q_t \in C}\int_{\R_+^d} f(\rr) dq_t(\rr_{1:t-1}, \rr_t)- \sup_{q_t \in C} \int_{\R_+^d} f(\tilde{\rr}_{1:t-1}, \rr_t)dq_t(\tilde{\rr}_{1:t-1}, \rr_t)\right|\\
\le\ &\left|\rr_{1:t-1}-\tilde{\rr}_{1:t-1}\right|.
\end{align*}
In particular the functions 
\begin{align*}
&\sup_{q_{T}\in \mathcal{M},\ q_T \sim \P } \int_{\R_+^d}g(\rr)q_T(\rr_{1:T-1}; d\rr_T)\quad \text{ and } \sup_{q_T \in \hat{\Qc}_N} \int_{\R_+^d}g(\rr)q_T(\rr_{1:T-1}; d\rr_T)
\end{align*}
are $1$-Lipschitz continuous. The claim now follows using \cite[Prop. 7.34, p.154]{bertsekas1978stochastic}. 
\end{proof}
\begin{Rem}
Similar statements are valid for the penalty estimator and bounded functions $g$ as well as for the $\mathcal{W}^{\infty}$-estimator with continuous $g$.
\end{Rem}

\section{Asymptotic Consistency, Arbitrage and Contiguity} \label{sec. arb}
 
We discuss now the consistency of $\pi_{\Qc_N}(g)$ in \eqref{eq:generic_estimator} for general sets of martingale measures $\Qc_N$. This, in particular, provides a detailed motivation for the construction of our improved estimator in Theorem \ref{Thm wasserstein}.
Throughout, we assume NA($\P$) to be able to use the dual formulation \eqref{eq:sh_duality} and recall that NA($\P)$, by the First Fundamental Theorem of Asset pricing, is equivalent to existence of $\Q\in \Mc$, $\Q\sim\P$. 

Clearly, a first necessary condition for asymptotic consistency of $\pi_{\Qc_N}(g)$ is that $\Qc_N \neq \emptyset$ for $N$ large enough. For the plugin estimator, when $\Qc_N= \{\Q \in \mathcal{M}: \Q \sim \hat{\P}_N \}$, this was established in Proposition \ref{prop. supp} which outlined the relationship between NA$(\P)$ and NA$(\hat{\P}_N)$. The proof crucially relied on the fact that $\text{supp}(\hat{\P}_N) \subseteq \text{supp}(\P)$. Indeed, for general measures $\nu_N \Rightarrow \P$ the affine hull $\text{aff}(\text{supp}(\nu_N))$ could be of higher dimension than $\text{aff}(\text{supp}(\P))$ and thus there is no relationship between NA$(\P)$ and NA$(\nu_N)$ as the following example shows:
\begin{Ex}\label{Ex coun}
Take $d=1$ and $\P= \delta_{1}$. Obviously $\nu_N := \delta_{2}/N+(1-1/N) \delta_{1}$ converges weakly to $\P$. While NA$(\P)$ holds, NA$(\nu_N)$ is never fulfilled. Conversely $\nu_N := \delta_0/(2N)+\delta_1/(2N)+(1-1/N)\delta_2 \Rightarrow \delta_2$, so there exists a $\P$-Arbitrage while there is no $\nu_N$-Arbitrage.
\end{Ex}

It is thus both natural and necessary to maintain a relationship between $\Qc_N$ and $\{\Q \in \mathcal{M}: \Q \sim \hat{\P}_N \}$. A minimal property seems to be one of an asymptotic inclusion in the sense that for some sequence $(k_N)_{N \in \N}$ with $\lim_{N \to \infty}k_N=\infty$ we have 
\begin{equation}\label{eq:nec_QN}
\P^{\infty}(\{\Q \in \mathcal{M}: \Q \sim \hat{\P}_N\  \&\  \|d\Q/ d\hat{\P}_N \|_{\infty}\le k_N \} \subseteq \Qc_N \text{ for large } N) = 1.
\end{equation}
Then the consistency of the plugin estimator in Theorem \ref{Thm. one-per} implies that 
\begin{align*}
\liminf_{N \to \infty} \pi_{\Qc_N}(g) \ge \lim_{N \to \infty} \hat{\pi}_N(g) \ge \pi^{\P}(g) \quad \P^{\infty}\text{-a.s.}
\end{align*} 
The main task now is to identify sequences of sets $\Qc_N$ which satisfy \eqref{eq:nec_QN} and for which the reverse inequality
\begin{align}\label{eq. goal}
\limsup_{N \to \infty} \pi_{\Qc_N}(g) \le \pi^{\P}(g)=\sup_{\Q \in \mathcal{M}: \Q\sim \P} \E_\Q[g] \quad \P^{\infty}\text{-a.s.}
\end{align}
holds. For this, we need the $\Qc_N$ to asymptotically decrease to the set $\{\Q \in \mathcal{M}: \Q \sim \P\}$. More formally, denoting the $\epsilon$-neighbourhood of a set $A$ by $A^{\epsilon}$, the following can be seen to be a necessary condition for \eqref{eq. goal} when $g$ is continuous and bounded:
\begin{align}\label{eq. as}
\liminf_{N \to \infty }\nu_N(A^{\epsilon})&=0 \text{ for some sequence } (\nu_N)_{N \in \N} \text{ such that } \nu_N \Rightarrow\P \text{ implies } \nonumber \\
 \lim_{N \to \infty}\Q_N(A) &= 0 \text{ for all }(\Q_N)_{N \in \N}\text{, such that }\Q_N \in \mathcal{Q}_N \text{ for all }N \in \N,
\end{align}
for all $A \in \mathcal{B}(\R^d_+)$ and all $\epsilon>0$.
This condition, which is trivially satisfied if $\mathcal{Q}_N=\{ \Q \in \mathcal{M} \ | \ \Q \sim \hat{\P}_N \}$, can be construed as a contiguity condition on $(\nu_N)_{N \in \N}$ and $(\Q_N)_{N \in \N}$. The notion of contiguity goes back to Le Cam (\cite[Chp.~5]{le1990locally}) and was used by Kabanov and Kramkow (\cite{kabanov1995large}) in a continuous-time setting to describe large financial markets. We introduce a weaker version here, which is sufficient for our setting:
\begin{Defn}
A sequence of probability measures $(\Q_N)_{N \in \N}$ is o-contiguous wrt.\ a sequence $(\nu_N)_{N \in \N}$, if for all $\epsilon>0$ and for all open sets $O\in \mathcal{B}(\R_+^d)$ $\lim_{N \to \infty}\nu_N(O^{\epsilon}) = 0$ implies $\lim_{N \to \infty}\Q_N(O) = 0$. Sequences $(\nu_N)_{N \in \N}$ and $(\Q_N)_{N \in \N}$ are mutually o-contiguous if $(\nu_N)_{N \in \N}$ is o-contiguous wrt. $(\Q_N)_{N \in \N}$ and $(\Q_N)_{N \in \N}$ is o-contiguous wrt. $(\nu_N)_{N \in \N}$.\\
Given two sequences of sets of probability measures $(\mathfrak{P}_N)_{N \in \N}$ and $(\mathcal{Q}_N)_{N \in \N}$ we say that $(\mathcal{Q}_N)_{N \in \N}$ is o-contiguous wrt. $(\mathfrak{P}_N)_{N \in \N}$ if for every sequence $\Q_N \in \mathcal{Q}_N$ there exists a sequence $\nu_N \in \mathfrak{P}_N$ such that $(\Q_N)_{N \in \N}$  is o-contiguous wrt. $(\nu_N)_{N \in \N}$.
\end{Defn}
Though \eqref{eq. as} is a necessary condition for asymptotic consistency of $\pi_{\Qc_N}(g)$ it is not sufficient as the following example shows:
\begin{Ex}
Let $\P=1/3(\delta_0+\delta_1+\delta_2)$, then $\sup_{\Q \sim \P, \ \Q \in \mathcal{M}} \E_{\Q}((1-r)^+)=1/2$. Note that for $\nu_N =1/3(1-1/N)(\delta_0+\delta_1+\delta_2)+\delta_N/N$ we have $\nu_N \Rightarrow \P$ and $\Q_N=(1-3/N) \delta_0+1/N(\delta_1+\delta_2+\delta_N) \sim \nu_N$. Furthermore $\Q_N$ is o-contiguous with respect to $\nu_N$. Nevertheless $\Q_N \Rightarrow \delta_0$ and thus $E_{\Q_N}((1-r)^+)>3/4$ for $N >12$, so asymptotic consistency is not satisfied.
\end{Ex}

The above example shows that even though $(\Q_N)_{N \in \N}$ is o-contiguous wrt. $(\nu_N)_{N \in \N}$, $\Q_N$ converges weakly to a measure $\Q$, which is not a martingale measure. A way to resolve this problem, is to demand uniform integrability of $(\Q_N)_{N \in \N}$ as was done in \cite{hubalek1998does} in a continuous-time setting. This allows to establish the following:
\begin{Lem}\label{Cor. boundsupp}
Assume NA$(\P)$ holds. Let $\mathfrak{P}_N$ be such that for every sequence $(\nu_N)_{N \in \N}$ with $\nu_N \in \mathfrak{P}_N$ for all $N \in \N$ we have $\nu_N \Rightarrow \P$. If $k_N \to \infty$ and
\begin{enumerate}
\item $\{\Q \sim \hat{\P}_N,\ \Q \in \mathcal{M}, \ \|d\Q/ d\P \|_{\infty}\le k_N \} \subseteq \Qc_N$ for large $N$,
\item  $\Qc_N$ is o-contiguous wrt. $\mathfrak{P}_N$,
\item $\mathcal{Q}_N$ is such that every sequence $(\Q_N)_{N \in \N}$ with $\Q_N \in \mathcal{Q}_N$ for all $N \in \N$ is uniformly integrable,
\end{enumerate}
then $\pi_{\Qc_N}(g)$ is asymptotically consistent for all bounded and continuous $g$.
\end{Lem}
We omit the proof since relying on uniform integrability is not possible in general: if $\text{supp}(\P)$ is unbounded (2) and (3) of Lemma \ref{Cor. boundsupp} cannot be fulfilled at the same time. The following example illustrates this:
\begin{Ex}
Consider $\P$ with $\text{supp}(\P)=\R_+$ and $g(r)=(1-r)^+$. Then $\lim_{N \to \infty}\Q_N(\{0\})=1$ holds for every sequence $(\Q_N)_{N \in \N}$ such that 
\begin{align*}
\lim_{N \to \infty} \E_{\Q_N}[(1-r)^+]= \sup_{\Q \in \mathcal{M}_{\R_+}} \E_{\Q}[(1-r)^+]=1.
\end{align*}
Obviously $\delta_0$ is not a martingale measure, so $(\Q_N)_{N \in \N}$ cannot be uniformly integrable.
\end{Ex}

Thus uniform integrability of $(\Q_N)_{N \in \N}$ is in general too strict a requirement. In order to resolve this problem we strengthen the assumption, that $\nu_N$ converge weakly to the true measure $\P$. Instead we look at convergence in Wasserstein distance, which is known to metrize the weak convergence (cf. \cite[Theorem 6.9, p. 96]{villani2008optimal}). 

\section{Convergence rates for the plugin estimator}
\label{sec:cnvrates}
We study now in more detail the convergence rates for the plugin estimator $\hat{\pi}_N$ in the one-dimensional case. In particular, we prove Theorem \ref{thm:cnv_rate_basic} and also state some extensions. 

Following the approach of Hampel, it would be natural to consider the \emph{influence curve} of the superhedging functional, which is simply its Gateaux derivative at $\P$ in the direction of $\delta_{r}$ and represents the marginal influence of an additional observation with value $r$ when the sample size goes to infinity. This however produces trivial results: the function $\epsilon \mapsto \pi^{(1-\epsilon)\P+\epsilon\delta_r}(g)$ is constant on $(0,1)$ so the influence curve at $(r,\P)$ will be equal to zero or infinity. This happens because $(1-\epsilon)\P+\epsilon\delta_r$ have the same support for all $\epsilon\in (0,1)$. Clearly, to assess sensitivity of $\pi^\P$ to changes in $\P$ we have to vary the support. To this end we consider
\begin{align*}
\pi^\P(g)-\inf_{A\subseteq\R^d_+: \P(A)\geq 1-\epsilon} \pi^{\P|_A}(g)
\end{align*}
as $\epsilon\to 0$ and study its natural normalisation. The following examples shows that $\epsilon$ is not the correct normalisation.

\begin{Ex}\label{Ex. 2}
Take again $g(r)=|r-1|\wedge 1$ and note that instead of considering $\P=\lambda_{[0,2]}/2$ we can thin out the tails of $\P$ by setting 
\begin{align*}
\frac{d\P^n}{d\P}\bigg|_{[0,1]}=\frac{r^n}{n+1} \hspace{0.5cm}\text{and}\hspace{0.5cm} \frac{d\P^n}{d\P}\bigg|_{[1,2]}=\frac{(2-r)^n}{n+1}
\end{align*}
Naturally, $\pi^{\P}=\pi^{\P_n}$ but with increasing $n$ the probability mass is less well spread over the support of $\P^n$. We calculate for $n\ge 2$ that $F_{\P^n}(r)=r^{n+1}/2$ for $r \le 1$ and thus $F_{\P^n}^{-1}(p)=\sqrt[n+1]{2p}$ for $p<1/2$. This readily implies
\begin{align*}
\pi^\P(g)-\inf_{A\subseteq\R^d_+: \P(A)\geq 1-\epsilon} \pi^{\P|_A}(g) =1-g(\sqrt[n+1]{\epsilon})=\sqrt[n+1]{\epsilon}.
\end{align*}
\end{Ex}
The above examples motivates using quantile functions for normalisation as stated in Theorem \ref{thm:cnv_rate_basic} which we now prove.
\begin{proof}[Proof of Theorem \ref{thm:cnv_rate_basic}]
As we have noted before $\hat{\pi}_N(g)\leq \pi^{\P}(g)$ and $\hat{\pi}_N(g)$ is non-decreasing in $N$. Let us first consider $\P$ with bounded support. It suffices to show that
\begin{align*}
\sup_{\Q \sim \P, \ \Q \in \mathcal{M}}\E_{\Q}[g] \le \sup_{\Q \sim \hat{\P}_N, \ \Q \in \mathcal{M}}\E_{\Q_N}[g]+\mathcal{O}(\delta(\kappa^N)).
\end{align*}
Without loss of generality we assume $r_1 \le r_2 \le \dots \le r_N$ for the rest of the proof. By the definition of $d_N$ 
\begin{align*}
F^{-1}_{\P}((p-d_N) \vee 0+)\le F^{-1}_{\hat{\P}_N}(p)\le F^{-1}_{\P}(p+d_N)
\end{align*} 
holds for all $p\in [0,1]$. Thus we note that for all $i=2, \dots ,N$
\begin{align*}
r_{i}-r_{i-1}&=F^{-1}_{\hat{\P}_N}\left(\frac{i}{N}\right)-F^{-1}_{\hat{\P}_N}\left(\frac{i-1}{N}\right) \\
&\le \sup_{k=1, \dots \lfloor 1/(3d_N)\rfloor} F^{-1}_{\P}(3kd_N)-F^{-1}_{\P}(3(k-1)d_N \vee 0+)\le\kappa_N.
\end{align*}
Next we remark that by definition of $\pi^{\P}(g)$ and Proposition \ref{prop:d_H cont} there exists $C>0$ such that $$ \pi^{\P}(g)-\pi^{\P(\cdot|[r_1, r_N])}(g) \le C\delta(F^{-1}_{\P}(d_N)- F^{-1}_{\P}(0+))+C\delta(F^{-1}_{\P}(1)-F^{-1}_{\P}(1-d_N)\rfloor)).$$ 
Take $\Q \sim \P(\cdot|[r_1, r_N])$, $\Q \in \mathcal{M}$. We want to apply a simple Balayage construction to redistribute the mass of $\Q$ on the support of $\hat{\P}_N$. Thus we set
\begin{align*}
\Q_N(\{r_1\})&=\int_{[r_1,r_2)} \frac{r-r_2}{r_1-r_2}\Q(dr),\\
\Q_N(\{r_i\})&= \int_{[r_{i-1},r_i)} \frac{r-r_{i-1}}{r_i-r_{i-1}} \Q(dr)+\int_{[r_i,r_{i+1})} \left(1-\frac{r-r_i}{r_{i+1}-r_i} \right) \Q(dr)
\\&\hspace{7cm}\text{for }i=2, \dots, N-1.\\
\Q_N(\{r_N\})&= \int_{[r_{N-1},r_N]} \frac{r-r_N}{r_{N-1}-r_N}.
\end{align*}
A straightforward calculation shows $\Q_N(\{r_1, \dots, r_N\})=1$ and $E_{\Q_N}[r]=1$. We have
\begin{align}\label{eq. conv}
\nonumber &\left|\int_{r_1}^{r_N}g d\Q_N-\int_{r_1}^{r_N}gd\Q\right| \\
& = \sum_{k=2}^{N} \int_{r_{k-1}}^{r_k} \left|\frac{g(r_k)(r-r_{k-1})-g(r_{k-1})(r-r_k)}{r_{k}-r_{k-1}}-g(r)\right|\Q(dr)\nonumber\\
& \le \delta(\kappa^N).
\end{align}
For $\P \in \mathcal{P}(\R_+)$ with unbounded support and $g$ bounded by $D>0$ we note that
\begin{align*}
\pi^{\P}(g)-\pi^{\P(\cdot|[0,r_N])}(g) \le \frac{2D}{r_N} \le \frac{2D}{F^{-1}_{\P}(1-d_N)}.
\end{align*}
This concludes the proof.
\end{proof}
In order to improve upon, and further specify the results in Theorem \ref{thm:cnv_rate_basic}, we recall Lemma \ref{lem. rasonyi} and make the following easy observation:
\begin{Lem}
Let $C >0$ and $\P \in \mathcal{P}(\R_+^d)$. Then the set $\{\Q \in \mathcal{M} \ | \ \|d\Q / d\P \| \le C \}$ is weakly compact.
\end{Lem}
\begin{proof}
As
\begin{align*}
\limsup_{K \to \infty}\sup_{n \in \N}\E_{\Q_n}[|r| \mathds{1}_{\{|r|\ge K\}}] \le \limsup_{K \to \infty} C \E_{\P}[|r| \mathds{1}_{\{|r|\ge K\}}]=0,
\end{align*}
the claim follows.
\end{proof}
\begin{Cor}\label{Thm. conv. speed ext}
Let $\P \in \mathcal{P}(\R_+)$ have bounded support and let $g:\text{supp}(\P) \to \R$ be continuous such that $|g(r)-g(\tilde{r})| \le \delta(|r-\tilde{r}|)$ for some monotone $\delta: \R_+ \to \R_+$ with $\delta(r) \to 0$ for $r \to 0$. If \eqref{eq:bounded_mart_dens} holds then 
\begin{align*}
\sup_{\Q \sim \P, \ \Q \in \mathcal{M}} \E_{\Q}[g]- \sup_{\Q \sim \hat{\P}_N, \ \Q \in \mathcal{M}} \E_{\Q}[g] = \mathcal{O}(d_N^{1/2}+\delta(d_N^{1/2})) \hspace{0.5cm}\P^{\infty}\text{-a.s.}
\end{align*}
\end{Cor}
\begin{proof}
Note that by assumption it is sufficient to consider martingale measures $\Q \sim \P$ such that $\|d\Q/d\P\|_{\infty}\le C$ for some $C>0$. Similarly to the proof of Theorem \ref{thm:cnv_rate_basic} we set
\begin{align*}
\Q_N(\{r_1\})&=\int_{[r_1,r_2)} \frac{r-r_2}{r_1-r_2}\Q(dr)+\Q([0,r_1)),\\
\Q_N(\{r_i\})&= \int_{[r_{i-1},r_i)} \frac{r-r_{i-1}}{r_i-r_{i-1}} \Q(dr)+\int_{[r_i,r_{i+1})} \left(1-\frac{r-r_i}{r_{i+1}-r_i} \right) \Q(dr)
\\ &\hspace{7cm}\text{for }i=2, \dots, N-1.\\
\Q_N(\{r_N\})&= \int_{[r_{N-1},r_N]} \frac{r-r_N}{r_{N-1}-r_N}+\Q((r_N, \infty)).
\end{align*}
and note that $\Q_N(\{r_1, \dots, r_N\})=1$ and $$|E_{\Q_N}[r]-1|=|E_{\Q_N}[r]-\E_{\Q}[r]|=\int_0^{r_1}\left|r_1-r\right| d\Q(r)+\int_{r_n}^{\infty}|r_n-r|d\Q(r)\le Kd_N$$ for some $K>0$. Also
\begin{align*}
\pi^{\P}(g)- \hat{\pi}_N(g) \le \sup_{\Q \sim \P, \ \Q \in \mathcal{M}} \E_{\Q}[g]-\sup_{\Q \sim \hat{\P}_N, \ |\E_{\Q}(r-1)|\le Kd_N}\E_{\Q}[g]+2\delta(Kd_N)
\end{align*}
We further assume that $g$ is bounded by $D$. Then \eqref{eq. conv} becomes
\begin{align*}
\nonumber\left|\int_{r_1}^{r_N}g d\Q_N-\int_{r_1}^{r_N}gd\Q\right| &=  
\sum_{k=2}^{N} \int_{r_{k-1}}^{r_k} \left|\frac{g(r_k)(r-r_{k-1})-g(r_{k-1})(r-r_k)}{r_{k}-r_{k-1}}-g(r)\right|\Q(dr)\nonumber\\
&\le \delta(d_N^{1/2})+2CD\left|\left\{k\in \{1, \dots, \lfloor 1/(3d_N)\rfloor\} \ \big| \ |\kappa_k^N| \ge d_N^{1/2}\right\}\right|3d_N \\
&\le \delta(d_N^{1/2})+\frac{6CDF^{-1}_{\P}(1)d_N}{d_N^{1/2}}=\mathcal{O}(d_N^{1/2}+\delta(d_N^{1/2})).
\end{align*}
This concludes the proof.
\end{proof}

\begin{Rem} 
We note that the above asymptotic result can be used to set up a utility based hedging problem to approximate $\pi^\P$. If we let 
\begin{align*}
\alpha_N:= U\left(C(d_N^{1/2}+\delta(d_N^{1/2}))\right)
\end{align*}
for a concave and strictly increasing $U$ and some $C \in \R_+$ then 
$$\pi^\P(g)\approx \sup_{\Q \sim \hat{\P}_N, \ \Q \in \mathcal{M}} \E_{\Q}[g] + U^{-1}(\alpha_N)$$
is the value of the utility based hedging problem under $\hat{\P}_N$
\begin{align*}
\inf \{x \in \R \ | \ \exists H \in \R \text{ s.t. } U(x+H(r_1)-g(r)) \ge \alpha_N \text{ for } r=r_1, \dots, r_N \}.
\end{align*}
\end{Rem}
Let us now state a generalisation of Lemma \ref{lem dkw} to Markov chains, where we recall Definition \ref{def:markov} and notation from Section \ref{sec:plugin_cnvrate}:
\begin{Lem}[cf. {\cite[Theorem 11.24, p.228]{kosorok2008introduction}}]
Assume that $r_1, r_2, \dots$ are realisations of a stationary $\beta$-mixing Markov chain with exponential decay and invariant measure $\P$ as its initial distribution. Then, as $N \to \infty$
\begin{align*}
\sqrt{N} \sup_{x\in \R_+} \left|F_{\hat{\P}_N}(x)-F_{\P}(x)\right| \Rightarrow \sup_{x\in \R_+} |G(F_{\P}(x))|,
\end{align*}
where $G$ is a standard Brownian bridge on $[0,1]$.
\end{Lem}
\begin{proof}
Setting $\mathcal{F}=\{\mathds{1}_{(-\infty, x]}\  :\  x\in \R_+\}$, we can choose $p=4$  in \cite[Theorem 11.24, p.228]{kosorok2008introduction} and immediately check that for $\P$ with geometric mixing rate $\rho$ we have
\begin{align*}
\sum_{k=1}^\infty k \rho^k <\infty 
\end{align*}
as well as $J_{[]}(\infty, \mathcal{F},L_4(\P))<\infty$, where $J_{[]}(\infty, \mathcal{F},L_4(\P))$ denotes the bracketing integral for $\mathcal{F}$ and $p=4$ (see e.g. \cite[p. 17]{kosorok2008introduction}). This concludes the proof.
\end{proof}
For examples of stationary $\beta$-mixing Markov chains with exponential decay we refer to Lemma \ref{lem:ergodicity} and Corollary \ref{cor:ergodicity}.
The above lemma can be used to obtain asymptotic confidence bounds for Corollary \ref{Thm. conv. speed ext}. More precisely, for $N\to \infty$, we obtain
\begin{align*}
\P^{\infty}(|\pi^\P(g)-\hat{\pi}_N(g)|\ge \epsilon)&\le \P^\infty (d_N^{1/2}+\delta(d_N^{1/2} )\ge \epsilon/C)\\
&\le \P^\infty (d_N \ge f(\epsilon/C))\\
&\lesssim \P^\infty\left(\sup_{x\in \R_+} |G(F_{\P}(x))|\ge f(\epsilon/C)\sqrt{N}\right)
\end{align*}
for some constant $C>0$ and $f:\R_+\to \R_+$ being the square of the inverse of $x\to x+\delta(x)$. Similarly, under some regularity  assumptions on $F_{\P}^{-1}$, an analogous but non-asymptotic estimate can be given for Theorem \ref{thm:cnv_rate_basic} using Lemma \ref{lem dkw}. 

To close this section we consider the convergence rate in some cases when Theorem \ref{thm:cnv_rate_basic} does not apply. Note that Theorem \ref{thm:cnv_rate_basic} applies for a continuous $g$ whenever $\P$ has bounded support or if there exists $K>0$ such that 
\begin{align}\label{eq:cnvrate_restrict}
\sup_{\Q \sim \P|_{[0,K]}, \ \Q \in \mathcal{M}} \E_{\Q}[g]= \sup_{\Q \sim \P, \ \Q \in \mathcal{M}} \E_{\Q}[g].
\end{align}
Suppose now that no such $K$ exists. Clearly if $g(r)/r\to \infty$ as $r\to \infty$ then $\pi^\P(g)=\infty$ so consider $g$ with linear growth: $g(r)/r\to c\in \R$ as $r\to \infty$. As $\P$ necessarily has unbounded support we can take a sequence $K_n\to \infty$ and, by the above condition, some $(H_n)_{n \in \N}$ such that 
$$\pi^{\P(\cdot|[0,K_n])}(g)+H_n(K_n-1)=g(K_n).$$ 
As $\pi^{\P(\cdot|[0,K_n])}(g) \to \pi^{\P}(g)$ we conclude $H_n \uparrow c$. 
Thus $\pi^{\P}(g)=\max_{\lambda \in [0,1]\cap \text{supp}(\P)}(g(\lambda)-c(\lambda-1))=g(\tilde{r})-c(\tilde{r}-1)$ for some $\tilde{r}\in [0,1]\cap \text{supp}(\P)$. In particular
$$ \pi^{\P(\cdot|[0,K_n])}(g) \ge \frac{K_n-1}{K_n-\tilde{r}}g(\tilde{r})+\frac{1-\tilde{r}}{K_n-\tilde{r}}g(K_n).$$ 
Clearly $\hat{\pi}_N(g)\leq \pi^{\P(\cdot|[0,r_n])}(g)$ so the above, using that $g$ is bounded on $[0,1]$, implies that the convergence rate is at most of the order of 
\begin{equation}\label{eq:cnv_rate_bespoke}
\frac{1}{\max_{i=1, \dots, N} r_i}+\left(c-\frac{g(\max_{i=1, \dots, N} r_i)}{\max_{i=1, \dots, N} r_i}\right)
\end{equation}
but could be slower. Typically, e.g., if $\P$ has a density bounded from below in the neighbourhood of $\tilde{r}$, the tails of the distribution $\P$ are the decisive feature for the convergence rate for $\hat{\pi}_N$ and \eqref{eq:cnv_rate_bespoke} holds. We discuss this in more detail in the examples below. 
\begin{figure}[h!]
  \centering
  \captionsetup{width=.6\linewidth}
  \begin{minipage}[b]{0.45\textwidth}
    \includegraphics[width=\textwidth]{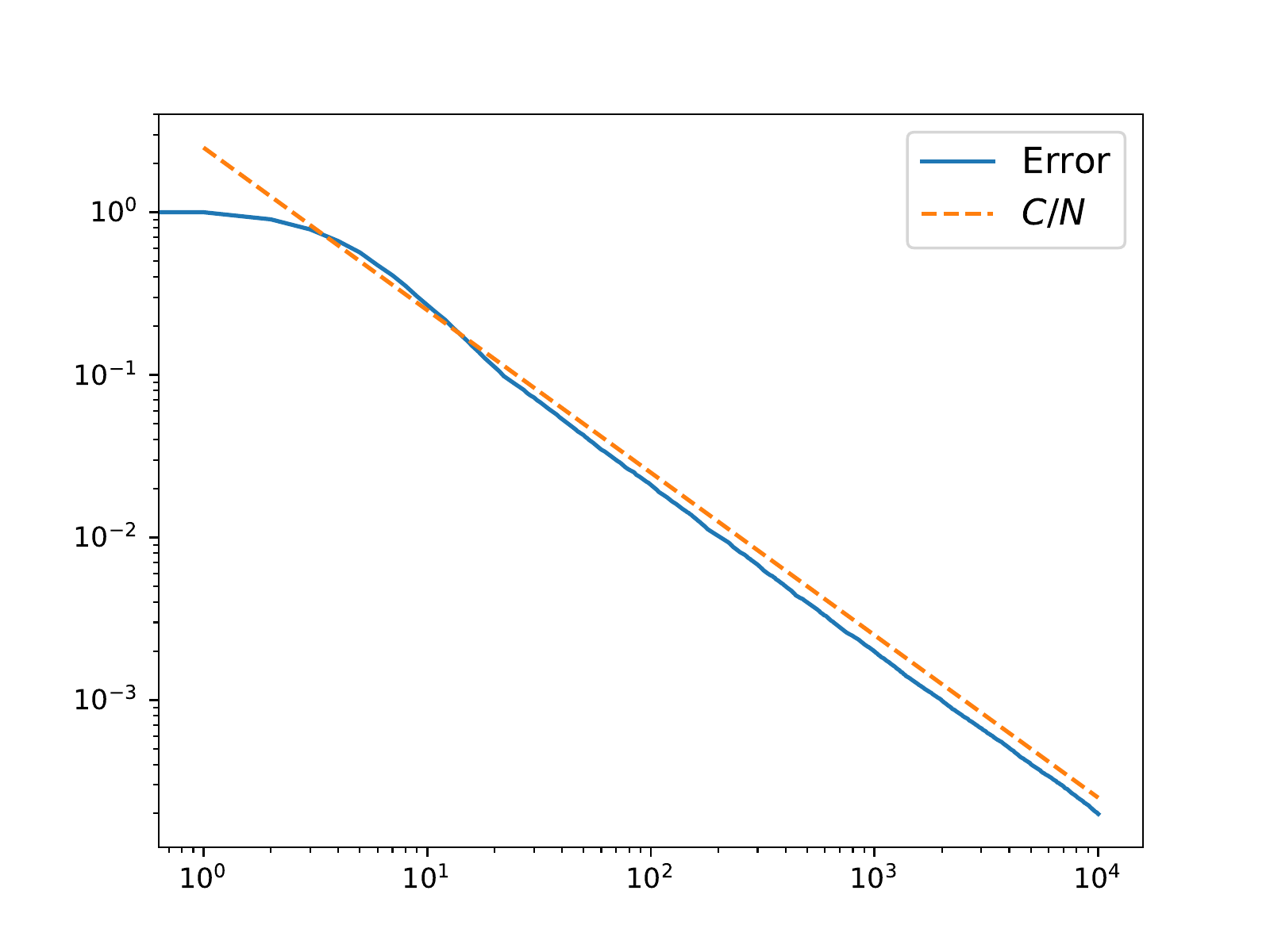}
    \caption{$g(r)=|r-1|$, $\P=\lambda|_{[0,2]}/2$}
    \label{fig. min1}
  \end{minipage}
  \hfill
  \begin{minipage}[b]{0.45\textwidth}
    \includegraphics[width=\textwidth]{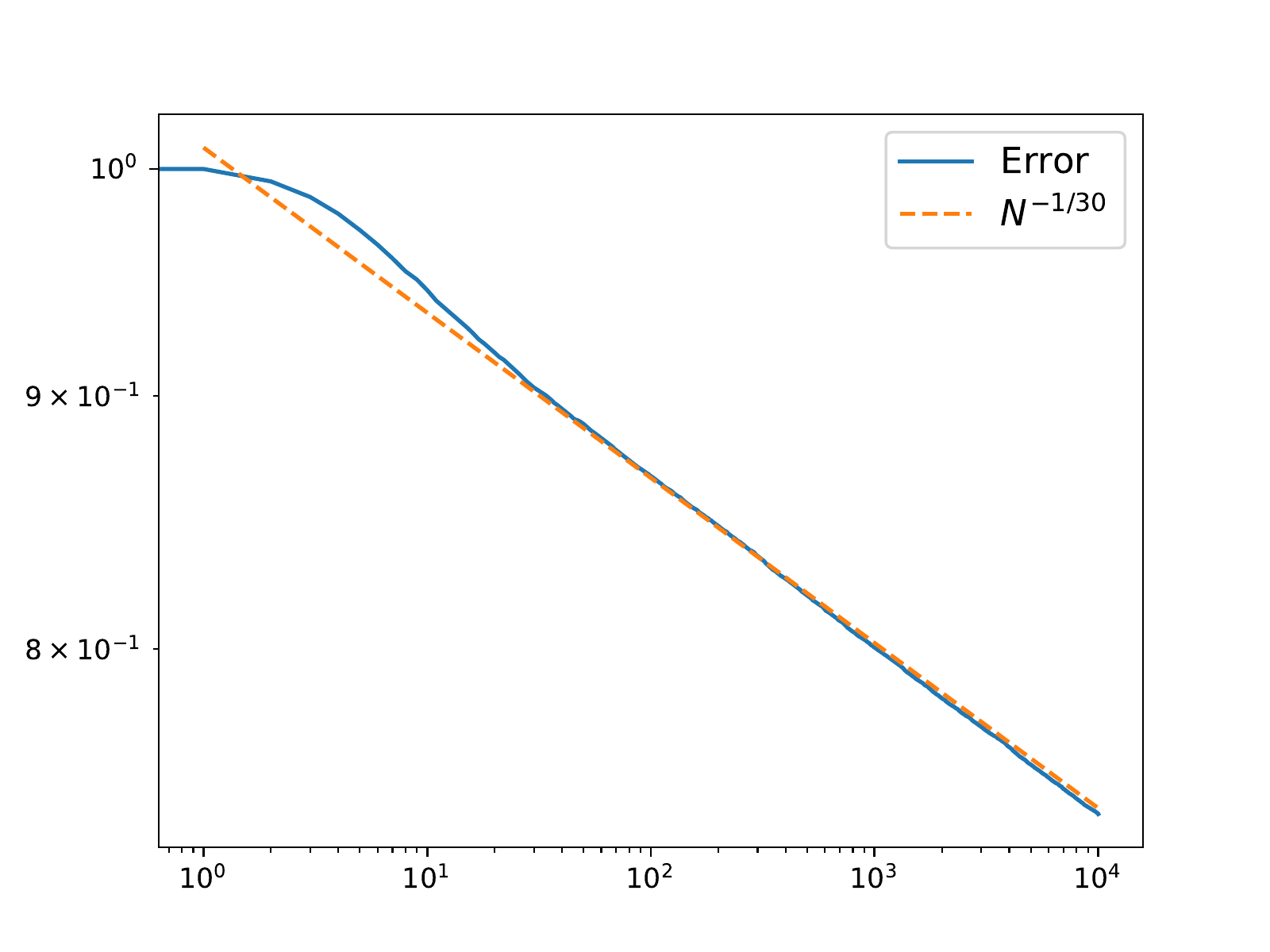}
    \caption{$g(r)=|r-1|$, $\P=\P^{29}$}
    \label{fig. min2}
  \end{minipage}
  \hfill
  \begin{minipage}[b]{0.45\textwidth}
    \includegraphics[width=\textwidth]{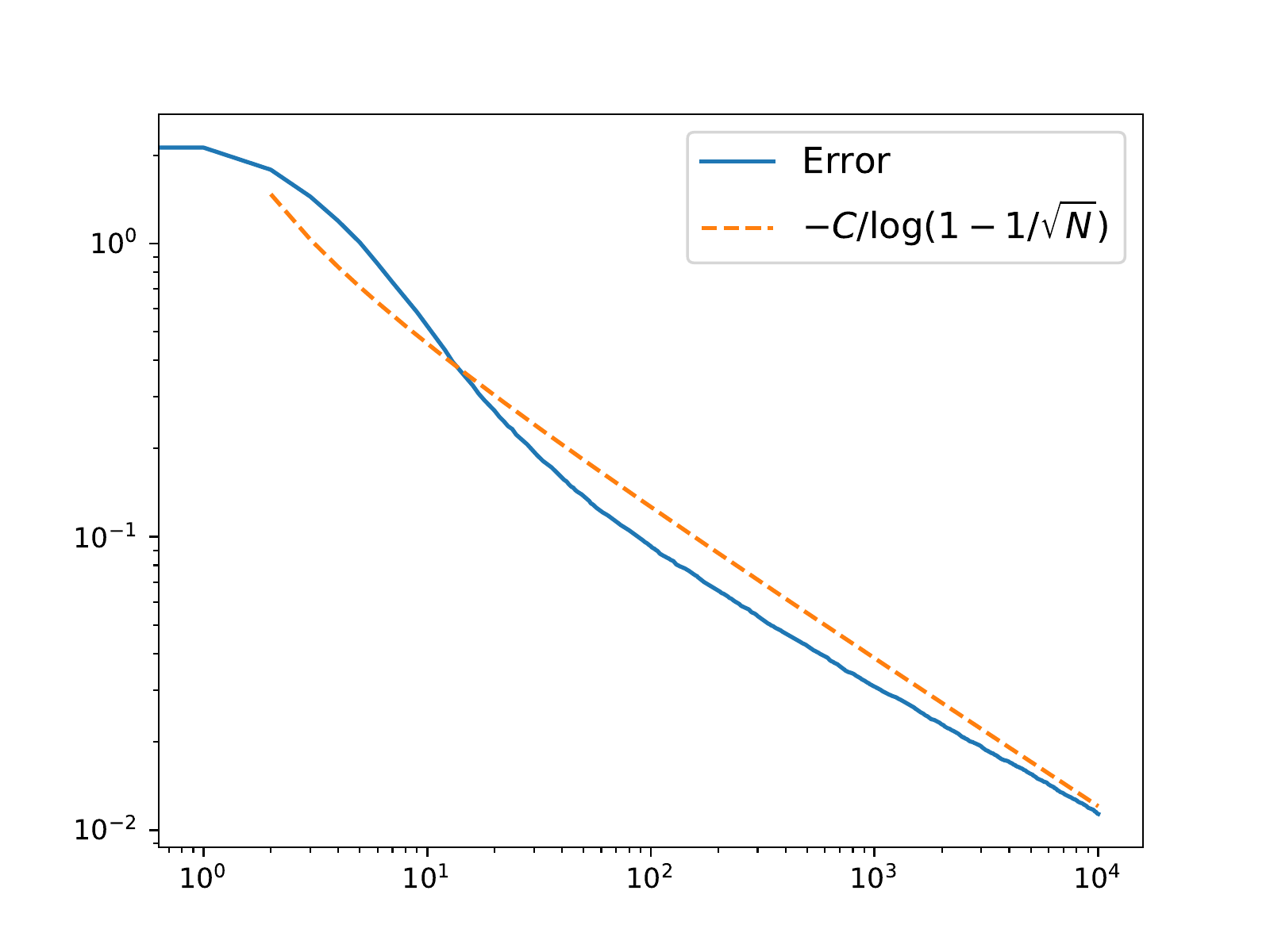}
    \caption{$g(r)=(2-r)\mathds{1}_{\{r\le1\}}$\\$+\sqrt{r}\mathds{1}_{\{r \ge 1\}}$, $\P=\text{Exp}(1)$}
    \label{fig. min3}
  \end{minipage}
  \hfill
  \begin{minipage}[b]{0.45\textwidth}
    \includegraphics[width=\textwidth]{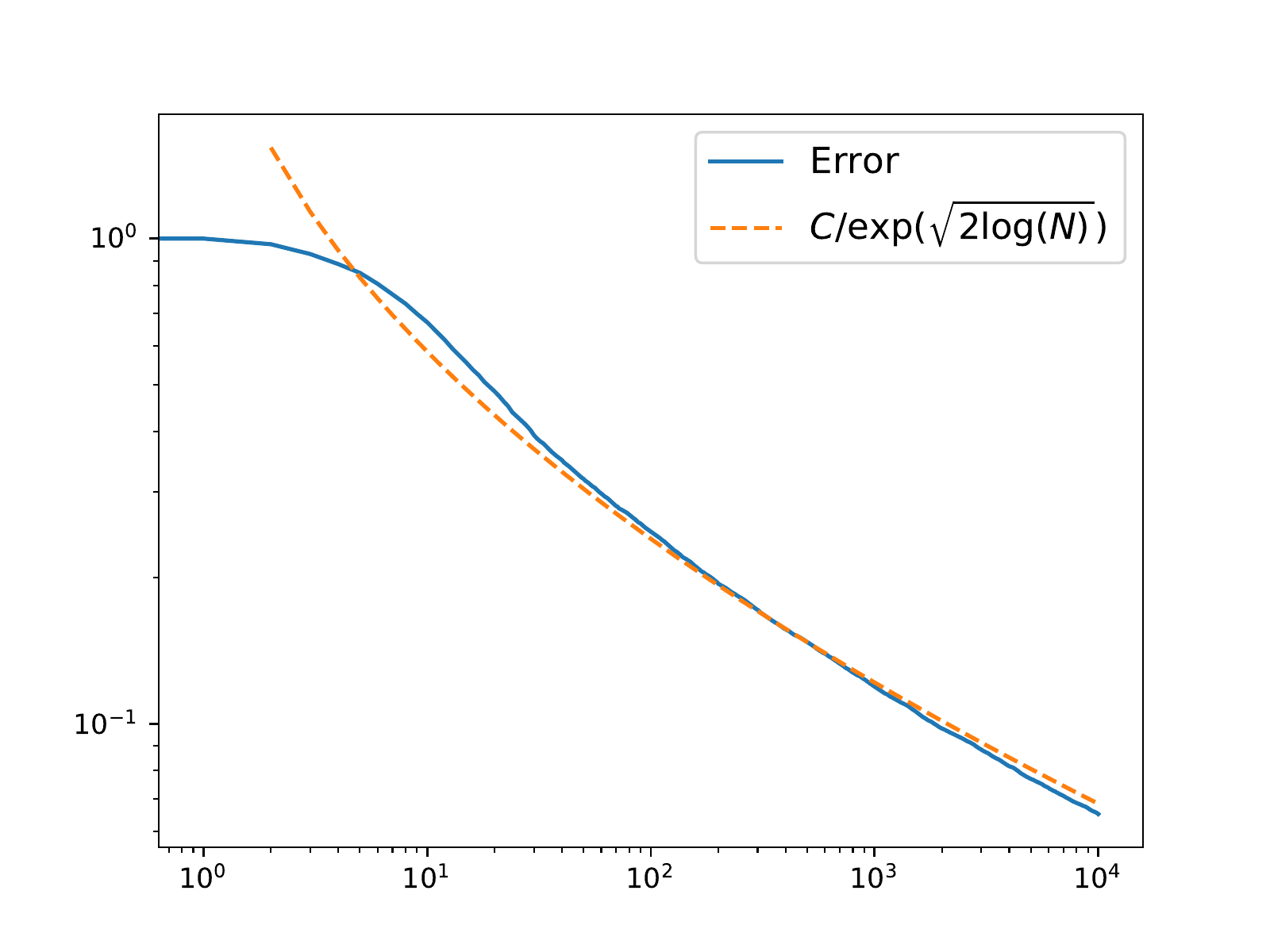}
    \caption{$g(r)=(r-2)^+$, \\$\P=\exp(\mathcal{N}(0,1))$}
    \label{fig. min6}
  \end{minipage}
  \hfill
  \begin{minipage}[b]{0.45\textwidth}
    \includegraphics[width=\textwidth]{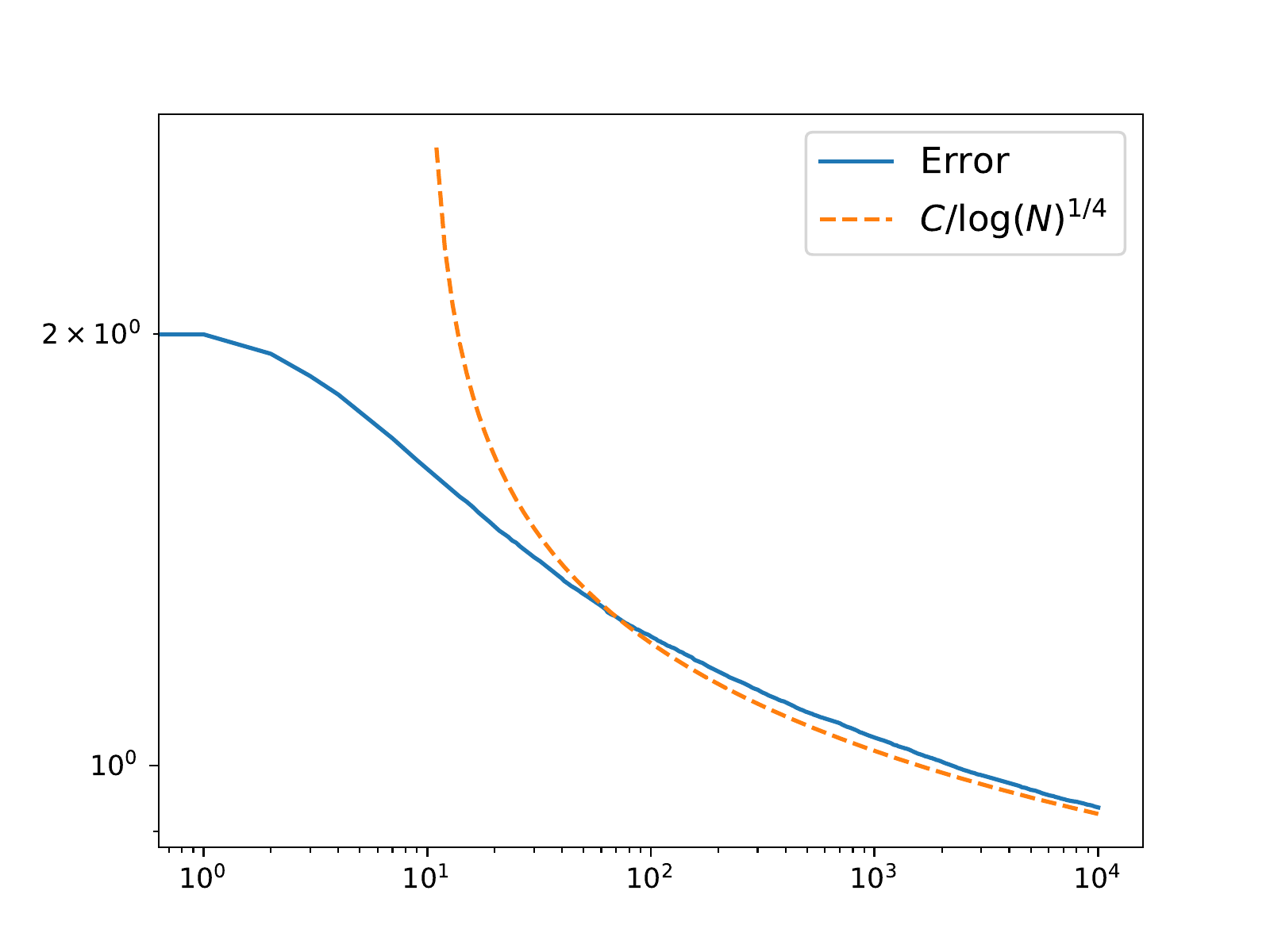}
    \caption{$g(r)=|r-1|+\sqrt{r-1}\mathds{1}_{\{r \ge 1\}}$,\\ $\P=|\mathcal{N}(0,1)|$}
    \label{fig. min4}
  \end{minipage}
  \hfill
  \begin{minipage}[b]{0.45\textwidth}
    \includegraphics[width=\textwidth]{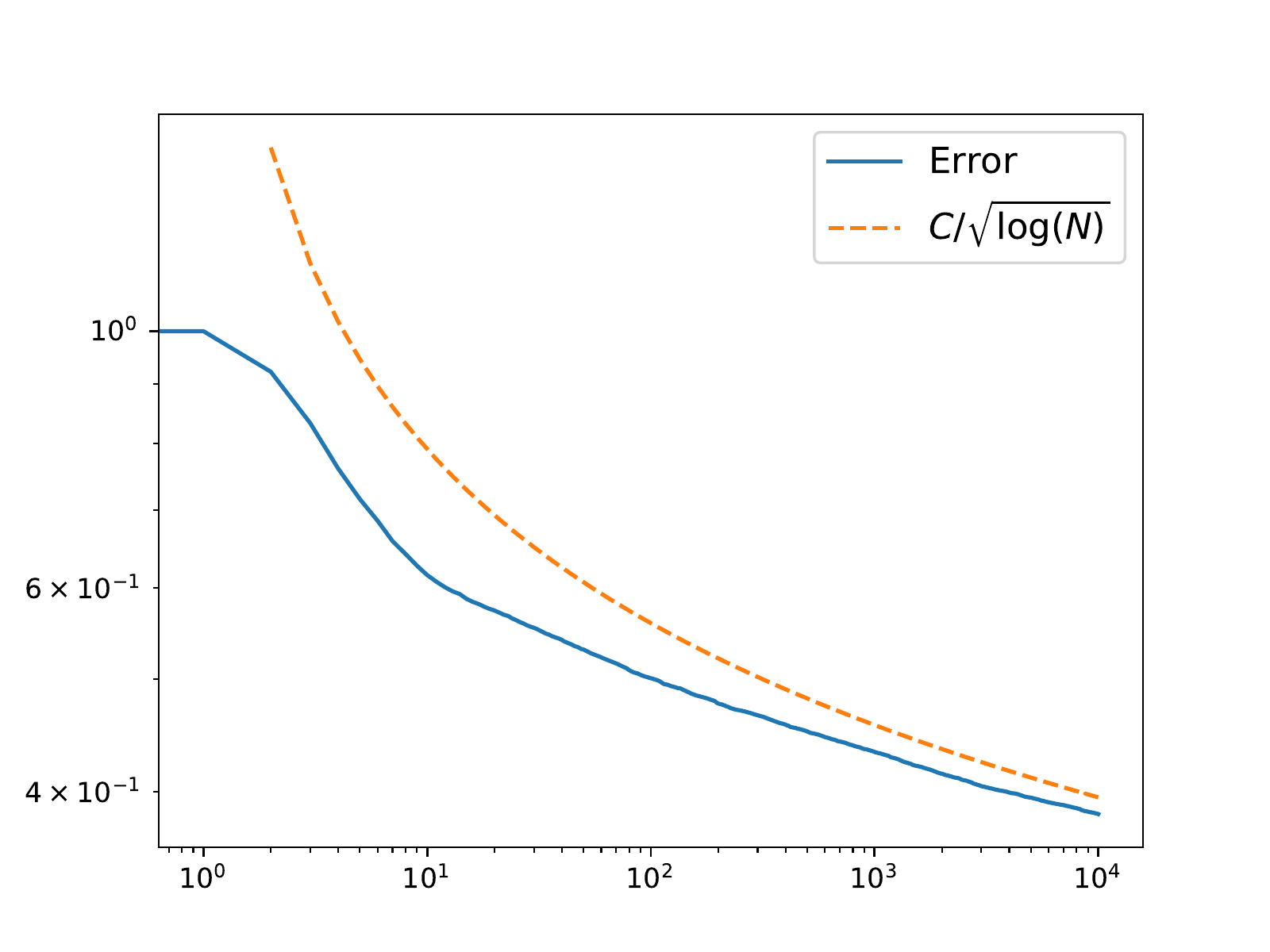}
    \caption{$g(r)=(1-r)\mathds{1}_{\{r\le 1\}}-\sqrt{r-1}\mathds{1}_{\{r>1\}}$, $\P=\text{Exp}(1)$}
    \label{fig. min5}
  \end{minipage}
\end{figure}

\begin{Ex}
We provide now some examples to illustrate different convergence rates which might be observed in the context of Theorem \ref{thm:cnv_rate_basic} or \eqref{eq:cnv_rate_bespoke} above. 
We use $N=10^5$ realisations of $1000$ runs for our numerical illustrations. 
We start with some examples when Theorem \ref{thm:cnv_rate_basic} applies either directly or because \eqref{eq:cnvrate_restrict} holds. 
\begin{itemize}
\item $g(r)=|r-1|$, $\P=\lambda|_{[0,2]}/2$. Note that $\P(\max_{i=1, \dots, N}r_i\le x)=(r/2)^N$. Thus $$\E\left[\max_{i=1, \dots, N}r_i-1 \right]=\int_{0}^2 \frac{Nx^{N-1}(x-1)}{2^N}dx=\frac{2N}{N+1}-1=1-\frac{1}{N+1}.$$ Convergence rate $\mathcal{O}(1/N)$ (Figure \ref{fig. min1}).
\item $g(r)=|r-1|$, $\P=\P^{29}$ from Example \ref{Ex. 2}. Convergence rate $\mathcal{O}(1/N^{1/30})$ (Figure \ref{fig. min2}).
\item $g(r)=(2-r)\mathds{1}_{\{r\le1\}}+\sqrt{r}\mathds{1}_{\{r \ge 1\}}$, $\P=\text{Exp}(1)$. Note that $2+x/8$ is tangential to $\sqrt{x}$ in $x=16$. In particular $\pi^{\P}(g)=2.125=\pi^{\P(\cdot|[0,16])}(g)$. Convergence rate $\mathcal{O}(F^{-1}_{\P}(d_N))=\mathcal{O}(-\log(1-1/\sqrt{N}))$ (Figure \ref{fig. min3}).
\end{itemize}
We move now to examples where we can not rely on Theorem \ref{thm:cnv_rate_basic} but use \eqref{eq:cnv_rate_bespoke} instead and show the bound may be sharp.
The asymptotic distribution of $\max_{i=1, \dots, N} r_i$ can be determined using classical results from extreme value theory. In particular, the scaled maximum of exponential/normal/lognormal random variables converges weakly to a Gumbel distributed random variable $Y$ (see \cite[Examples 2,3, p.199]{takahashi1987normalizing}). 
\begin{itemize}
\item $g(r)=(r-2)^+$, $\P=\exp(\mathcal{N}(0,1))$.
Here $$\max_{i=1, \dots, N} r_i \overset{d}{=}\frac{Y\exp(\sqrt{2\log N})}{\sqrt{2\log N}}+\exp(\sqrt{2 \log N}) \sim \exp(\sqrt{2 \log N}).$$ Convergence rate $\mathcal{O}(1/\max_{i=1, \dots,N}r_i)=\mathcal{O}(1/\exp{(\sqrt{2\log{N}})})$ (Figure \ref{fig. min6}).
\item $g(r)=|r-1|-\sqrt{r-1}\mathds{1}_{\{r \ge 1\}}$, $\P=|\mathcal{N}(0,1)|$. Here $$\max_{i=1, \dots, N}r_i\overset{d}{\sim}\frac{Y}{\sqrt{2 \log N}}+ \left(\sqrt{2 \log N}- \frac{\log (4\pi \log N)}{2 \sqrt{2\log{N}}} \right)\sim \sqrt{ \log N}.$$ Convergence rate $\mathcal{O}(1/(\max_{i=1, \dots, N}r_i)^{1/2})=\mathcal{O}(1/\sqrt[4]{\log{N}})$ (Figure \ref{fig. min4}).
\item $g(r)=(1-r)\mathds{1}_{\{r\le 1\}}-\sqrt{r-1}\mathds{1}_{\{r>1\}}$, $\P=\text{Exp}(1)$. Here $$\max_{i=1, \dots, N}r_i\overset{d}{=}Y+\log N\sim \log N.$$
Convergence rate $\mathcal{O}(1/(\max_{i=1, \dots, N}r_i)^{1/2})=\mathcal{O}(1/\sqrt{\log{N}})$ (Figure \ref{fig. min5}).
\end{itemize}
\end{Ex}

\appendix
\section{Additional results and proofs}\label{sec:appendix}

\subsection{Additional results and proofs for Section \ref{sec. plugin}}

\begin{proof}[Proof of Theorem \ref{Thm. one-per}]
First note that if $A_n$ is a non-decreasing sequence of sets with $A=\lim_n A_n=\cup_n A_n$ then $\hat{g}_A=\lim_n \hat{g}_{A_n}$. The ``$\geq$" inequality is obvious and the reverse follows since $\hat{g}_{A_n}$ is a non-decreasing sequence of concave functions thus its limit is a concave function dominating $g$ on $A$.

Using Lusin's theorem (see \cite[Theorem 7.4.3, page 227]{cohn2013measure}) we can find an increasing sequence $K_n$ of compact subsets of $\supp(\P)$ such that $\P(\R_+^d \setminus K_n) \le 1/n$ and $g|_{K_n}$ is continuous. Continuity of $g$ on $K_n$ implies that $\hat{g}_{K_n}=\hat{g}_{\P_{|K_n}}\leq \hat{g}_{\P}$. 
On the other hand, by the argument above, $\lim_n \hat{g}_{K_n}=\hat{g}_{\cup_n K_n}\geq \hat{g}_{\P}$ since $\P(\cup_n K_n)=1$. We conclude that $\lim_n \hat{g}_{K_n}= \hat{g}_{\P}$.
Further, by Birkhoff's ergodic theorem (see \citep[Theorem 9.6, p.159]{kallenberg2002foundations}) and $\P_1\ll \P$ we have
$$\bigcup_N \supp(\hat{\P}^N)=\{r_1,r_2,\ldots\}\quad \text{ is a.s. dense in }\supp(\P)$$
and hence $\hat{g}_{K_n\cap \{r_1,r_2,\ldots\}}=\hat{g}_{K_n}$ a.s. By the argument above, we thus have
$$\lim_{N\to \infty} \hat{g}_{\hat{\P}_N}=\hat{g}_{\{r_1,\ldots \}} = \hat{g}_{\cup_n K_n \cap \{r_1,\ldots\}}
=\lim_{n\to \infty} \hat{g}_{K_n}=\hat{g}_{\P},\quad \P^{\infty}\text{-a.s.}$$
where the second equality follows since the inclusion
$\{r_1,r_2,\ldots\}\subset \cup_n K_n$ holds $\P^{\infty}$-a.s.
We conclude using \eqref{eq:cncenv}.
\end{proof}

\begin{proof}[Proof of Proposition \ref{prop. supp}]
Assume first that NA($\P$) holds. Denote the relative interior of the convex hull of a set $A \subseteq \R^d$ by $\text{ri}(A)$. Furthermore write $\text{lin}(A)$ for the linear hull of $A$ and $\text{aff}(A)$ for the affine hull. We recall that $1 \in \text{ri}(\text{supp}(\P))$ if and only if NA$(\P)$ holds. Consequently there exist $\epsilon >0$, $0 \le k \le d$ and  $r_i^{\pm} \in \text{conv}(\text{supp}(\P))$ with $r_i^{\pm}=1\pm \epsilon e_i$ for $i=1, \dots, k$, where $e_1, \dots e_k$ are an orthonormal basis of the space $\text{lin}(\text{supp}(\P)-1)$. Fix some $r_i^{\pm}$. Then there exists $n \in \N$ and $\tilde{r}_1, \dots, \tilde{r}_n \in \text{supp}(\P)$ such that $r_i^{\pm}$ can be written as a convex combination of $\tilde{r}_1, \dots, \tilde{r}_n$. Denote by $A$ the finite collection of $\tilde{r}$ for all $r_i^{\pm}$, $i=1, \dots, k$. Choosing $0< \delta < \epsilon$ sufficiently small we have $1 \in \text{ri}(\{\tilde{r}_i^{\pm} \ | \ i=1, \dots, k \})$ for each choice of $\tilde{r}_i^{\pm}\in B_{\delta}(r_i^{\pm})\cap \text{aff}(\text{supp}(\P))$. As $\text{supp}(\P^N) \subseteq \text{supp}(\P)$ and $\P^N \Rightarrow \P$ there exists $N_0 \in \N$ such that 
\begin{align*}
\P^N(B_{\delta}(\tilde{r})\cap \text{aff}(\text{supp}(\P)))&\ge \P^N(B_{\delta}(\tilde{r})\cap \text{aff}(\text{supp}(\P^N)))\\
&=\P^N(B_{\delta}(\tilde{r}))>0
\end{align*}
for all $\tilde{r} \in A$ and $N \ge N_0$. Thus $1 \in \text{ri}(\text{supp}(\P^N))$ and NA($\P^N)$ holds.\\
Conversely if there exists $H \in \R^d$ such that $\P( H(r-1) >0)>0$ and $\P(H(r-1) \ge 0)=1$ then again by continuity $H(r-1) \ge 0$ on $\text{supp}(\P) \supseteq \text{supp}(\P^N)$ and the set $\{r \ | \ H(r-1)>0 \}$ is open. Thus there exists $N_0 \in \N$ such that $\P^N(H(r-1)>0)>0$ for all $N \ge N_0$. 
\end{proof}

\begin{proof}[Proof of Corollary \ref{Cor. options}]
By NA$(\tilde{\P})$ we have
\begin{align*}
&\inf\{x \in \R \ | \ \exists H \in \R^{d+\tilde{d}} \text{ s.t. } x+H(e(r)-1)\ge g(r) \ \P\text{-a.s.}\}\\
& =\ \sup_{\Q \sim \P, \ \Q \in \mathcal{M}, \ \E_{\Q}(f_1)=f_0} \E_{\Q}[g].
\end{align*}
As $e$ is continuous $\hat{\P}_N \Rightarrow \P$ implies $\widehat{(\tilde{\P})}_N \Rightarrow \tilde{\P}$. Next assume $\text{supp}(\hat{\P}_N) \subseteq \text{supp}(\P)$. Again by continuity of $f_1$ we conclude that $$e(\text{supp}(\P))=\text{graph}(f_1) \cap (\text{supp}(\P) \times \R^{\tilde{d}})$$ is closed. Thus we clearly have $\text{supp}(\tilde{\P}) \subseteq e(\text{supp}(\P))$. We show $e(\text{supp}(\P)) \subseteq \text{supp}(\tilde{\P})$. Assume towards a contradiction there exists $e(r) \in e(\text{supp}(\P))\setminus \text{supp}(\tilde{\P})$. Note that for any sequence $(r_n)_{n \in \N}$ such that $\lim_{n \to \infty} r_n=r$ we have $\lim_{n \to \infty} e(r_n)=e(r)$. Thus there exists $\epsilon>0$ such that $e(B_{\epsilon}(r)) \cap \text{supp}(\tilde{\P})=\emptyset.$ But $\tilde{\P}(e(B_{\epsilon}(r)))=\P (B_{\epsilon}(r))>0$, a contradiction. Thus
$\text{supp}(\widehat{(\tilde{\P}})_N) \subseteq \text{supp}(\tilde{\P})$ so that $\{e(r_1),e(r_2) \dots\}$ are dense in $\text{supp}(\tilde{\P})$ and Theorem \ref{Thm. one-per} is still applicable for the enlarged market. The martingale constraint $\E_{\Q}[\tilde{r}^{d+i}-1]=0$ for $i=1, \dots, \tilde{d}$ is then equivalent to $\E_{\Q}[f^i_1(S)]=f^i_0(S_0)$.
\end{proof}

\begin{Lem} \label{lem. wasser v2}
Let $\P\in\prob(\R^d_+)$ with finite first moment and $\epsilon>0$. Then for all $x \in \R_+^d$ the ball $B^1_{\varepsilon}(\P)$ in the 1-Wasserstein metric contains $\lambda \delta_x+(1-\lambda) \P$ for some $\lambda \in (0,1)$. In particular, any $\P\in\prob(\R^d_+)$ can be written as a weak limit of probability measures $\P^N$ with supp$(\P^N)=\R_+^d$.
\end{Lem}

\begin{proof}[Proof of Lemma \ref{lem. wasser v2}]
Let $\epsilon >0$ be given. We define $\nu=\lambda \delta_x+(1-\lambda) \mu$ and recall that $\mathcal{L}_1$ denotes the 1-Lipschitz functions $f: \R_+^d \to \R$. Then
\begin{align*}
\mathcal{W}^1(\mu, \nu)&= \sup_{f \in \mathcal{L}_1} \left| \int_{\R_+^d} f(y)d\mu(y)- \int_{\R_+^d} f(y) d\nu(y) \right|=\lambda \sup_{f \in \mathcal{L}_1} \left| \int_{\R_+^d} f(y)-f(x)d\mu(y) \right| \\
&=\lambda \int_{\R_+^d} |x-y|d\mu(y)<\epsilon
\end{align*}
for $\lambda>0$ sufficiently small. 
\end{proof}

\begin{proof}[Proof of Theorem \ref{thm. rob}] Fix $\P\in \prob(\R^d_+)$ and a sequence $\P^N$ converging to $\P$. Let $\{r_1, r_2, \dots \}$ be dense in $\supp(\P)$. Fix $n\geq 1$ and note that, for any $i\geq 1$, weak convergence implies that $\P^{N}(B_{1/n}(r_i))>0$ for all $N$ large enough. In particular there exists $r_i^n \in B_{1/n}(r_i)$ such that $\hat{g}_{\P^N}(r_i^n)\ge g(r_i^n)$. Thus, by the same reasoning as in the proof of Theorem \ref{Thm. one-per} above, 
\begin{align*}
\liminf_{\P^N \Rightarrow \P} \pi^{\P^N}(g)=\liminf_{N \to \infty} \hat{g}_{\P^N}(1) \ge \lim_{n \to \infty} \lim_{k \to \infty}\hat{g}_{\{r_1^n, r_2^n, \dots, r_k^n\}}(1)=\pi^\P(g).
\end{align*}
We conclude using Lemma \ref{lem. wasser v2} since for a sequence with $\supp(\P^N)=\R^d_+$, by continuity of $g$, we have, for all $N\geq 1$, 
$$\pi^{\P^N}(g) = \inf \{x \in \R \ | \ \exists H \in \R^d \text{ s.t. }x+H(r-1) \ge g \text{ on } \R_+^d\}= \sup_{\Q \in \mathcal{M}_{\R_+^d}} \E_{\Q}[g].$$ \\
For the second part of the theorem, assume $\P \in \mathcal{P}(\R_+^d)$ is such that $$\pi^\P(g)< \sup_{\Q \in \mathcal{M}_{\R_+^d}} \E_{\Q}[g]=\inf \{x \in \R \ | \ \exists H \in \R^d \text{ s.t. }x+H(r-1) \ge g \text{ on } \R_+^d\}.$$ 
Take a sequence $(\P^N)_{N \in \N}$, as above, with $\supp(\P^N)=\R^d_+$ and $\P^N \Rightarrow \P$. Fix $\epsilon>0$ such that $$2\epsilon<\pi^{\P^N}(g)-\pi^\P(g).$$ For every $\delta >0$ there exists $N_0 \in \N$ such that for all $N \ge N_0$ we have $d_L(\P^N, \P) \le \delta$. Let $T_N$ be an asymptotically  consistent estimator of $\pi^\P(g)$. Then, for all $N$ large enough
\begin{align*}
d_L(\mathcal{L}_{\P^{N_0}}(T_N), \mathcal{L}_{\P}(T_N)) &\ge d_L(\delta_{\pi^\P(g)}, \delta_{\pi^{\P^{N_0}}(g)})
-d_L(\mathcal{L}_{\P^{N_0}}(T_N), \delta_{\pi^{\P^{N_0}}(g)})\\&-d_L(\delta_{\pi^{\P}(g)}, \mathcal{L}_{\P}(T_N)) 
\ge \epsilon.
\end{align*}
Thus $T_N$ is not robust at $\P$, which shows the claim.
\end{proof}

\begin{proof}[Proof of Proposition \ref{prop:d_H cont}]
Let us fix $\epsilon>0$. As $g$ is uniformly continuous, there exists $\delta>0$, such that for $|r-\tilde{r}|\le \delta$ we have $|g(r)-g(\tilde{r})| \le \epsilon/3$. Let us now consider $\tilde{\P}\in \prob(\R^d_+)$ such that $d_H(\text{supp}(\P), \text{supp}(\tilde{\P}))\le \delta$. Then $\pi^{\tilde{\P}}(g)<\infty$ and there exists $H_{\tilde{\P}} \in \R^d$ such that $$\pi^{\tilde{\P}}(g)+\epsilon/3+H_{\tilde{\P}}(\tilde{r}-1)\ge g(\tilde{r}) \quad \text{for all } \tilde{r}\in\text{supp}(\tilde{\P}).$$ Next we note that by no-arbitrage arguments detailed in \citep[Proof of Prop. 3.5]{cow} there exists $\delta_0>0$ such that for all $\tilde{\P}\in \prob(\R^d_+)$ with $d_H(\text{supp}(\P), \text{supp}(\tilde{\P}))\le \delta_0$ the strategy $|H_{\tilde{\P}}|$ can be chosen to be bounded by a constant $C>0$ which depends on $g$ but does not depend on $\tilde{\P}$. \\
In conclusion consider $\tilde{\P}\in \prob(\R^d_+)$ such that $$d_H(\text{supp}(\P), \text{supp}(\tilde{\P}))< \min\{\epsilon/(3C),\delta, \delta_0\}.$$ Then for every $\tilde{r} \in \text{supp}(\tilde{\P})$ there exists $r \in \text{supp}(\P)$ such that $|H_{\tilde{\P}}||r-\tilde{r}|\le \epsilon/3$ and $|g(r)-g(\tilde{r}) |\le \epsilon/3$. Thus
\begin{align*}
\pi^{\tilde{P}}(g)+\epsilon+H_{\tilde{\P}}(r-1) \ge \pi^{\tilde{\P}}(g)+2\epsilon/3+H_{\tilde{\P}}(\tilde{r}-1)\ge g(\tilde{r})+\epsilon/3\ge  g(r).
\end{align*} 
Thus $\pi^{\P}(g) \le \pi^{\tilde{\P}}(g)+\epsilon.$ Exchanging the roles of $\P$ and $\tilde{\P}$ yields the claim.
\end{proof}

\begin{proof}[Proof of Proposition \ref{prof:nocnvabove}]
We give the following counterexample: Define $\omega_1=(1,1,...)$,  $\P_n = (1-1/n)\delta_{1}+(1/n) \lambda_{[0,2]}$ and $g(r)=|1-r| \wedge 1$. Obviously $\P_n \Rightarrow \delta_{1}$ and $\sup_{\Q \sim \delta_{1}} \E_{\Q}[g]=0$ as well as $\sup_{\Q \sim \P_n, \ \Q \in \mathcal{M}} \E_{\Q}[g]=1$ for all $n \in \N$. By consistency we must have $T_N(\omega_1) \to 0$ as $N \to \infty$, in particular we can assume $T_{N_0}(\omega_1)<1$. Thus 
\begin{align*}
\P_n^{\infty} (T_N \ge 1 \text{ for all }N \ge N_0)&=
1- \P_n^{\infty} (T_N < 1 \text{ for some }N \ge N_0) \\
&\le 
1- \P_n^{N_0}(\{\omega_1\})=1-(1-1/n)^{N_0} \to 0 \\
&\hspace{5.5cm} (n \to \infty).
\end{align*}
\end{proof}

\subsubsection{Proof of Remark \ref{rem:applicable}}
With regards to Remark \ref{rem:applicable} let us now define the following concepts:
\begin{Defn}\label{def:markov}
Suppose that $\{X_n\ | \ n\in \N\}$ is a time-homogeneous Markov chain with initial distribution $\P$ as its invariant measure and transition kernel $K$. Then $\{X_n\}$ is called stationary $\beta$-mixing with exponential decay if there exist $0<\rho<1$ and $c>0$ such that
\begin{align*}
\int \left\|K^n(x, \cdot)-\P(\cdot)\right\|_{\mathrm{TV}}\P(dx) \le c\rho^n \quad \forall n\in \N,
\end{align*}
where $\|\cdot\|_{TV}$ defines the total variation norm.\\
The Markov chain $\{X_n\ | \ n\in \N\}$ is called geometrically ergodic  if there exists a probability measure $\P\in \prob(\R_+^d)$, a constant $0<\rho<1$ and a $\P$-integrable non-negative measurable function $g$ such that 
\begin{align*}
\left\|K^n(x,\cdot)-\P(\cdot)\right\|_{\mathrm{TV}}\le \rho^n g(x) \quad \forall n\in \N, \ \forall x\in \R^d_+.
\end{align*} 
\end{Defn}

The following lemma lists stationarity and ergodicity properties for some common time series models following \citep{mokkadem1990proprietes,boussama1998ergodicite, he1999fourth, carrasco2002mixing, basrak2002regular, lindner2009stationarity, francq2019garch}, where for simplicity we only consider the simplest cases and refer to the references above for a more comprehensive picture:
\begin{Lem}\label{lem:ergodicity}
The following conditions are sufficient for the $\beta$-mixing property with exponential decay (if started from the invariant distribution) for the processes $\{r_n\ | \ n\in \N\}$ and $\{h_n\ | \  n\in \N\}$ defined below, as well as for geometric ergodicity of the process $\{h_n\ | \ n\in \N\}$:
\begin{itemize}
\item Augmented GARCH(1,1) models, where 
\begin{align*}
r_n&=\sqrt{h_n}\eta_n, \quad n=0,1,\dots\\
\Gamma(h_{n+1})&= c(e_n)\Gamma(h_n)+g(e_{n+1}),
\end{align*}
and $\{\eta_n\}$ is an i.i.d. sequence independent of $h_0$ with $\E[\eta_n]=0$, $\E[\eta^2_n]=1$ and continuous positive density with respect to the Lebesgue measure, $\Gamma$ is an increasing function and $e_{n+1}$ is a measurable function of $\eta_n$.
\begin{itemize}
\item LGARCH: $\beta+\alpha<1$, where 
\begin{align*}
h_n=\omega+\beta h_{n-1}+\alpha \eta_{n-1}^2 h_{n-1}
\end{align*}
and $\omega>0$, $\beta\ge 0$, $\alpha\ge 0$. Furthermore, if there exists an integer $s\ge 1$ such that $(\beta+\alpha)^s<1$ or $\E[|\eta_n|^{2s}]<\infty$ and $\beta+\alpha< 1/(\E[|\eta_n|^{2s}])^{1/s}$, then $\E[|r_n|^{2s}]<\infty$.
\item MGARCH: $|\beta|<1$, where 
\begin{align*}
\log(h_n)=\omega+\beta \log(h_{n-1})+\alpha\log(h_{n-1}^2)
\end{align*}
and $\omega>0$, $\beta\ge 0$, $\alpha\ge 0$.
\item EGARCH: $|\beta|<1$, where 
\begin{align*}
\log(h_n)=\omega+\beta\log(h_{n-1})+\alpha(|\eta_{n-1}|+\gamma \eta_{n-1})
\end{align*}
and $\gamma \neq 0$.
\item NGARCH: $\beta +\alpha(1+c^2)<1$, where 
\begin{align*}
h_n=\omega+\beta h_{n-1}+\alpha(\eta_{n-1}-c)^2 h_{n-1}
\end{align*}
and $\omega>0$, $\beta\ge 0$, $\alpha\ge 0$. Furthermore, if there exists an integer $s\ge 1$ such that $\E[(\beta+\alpha(\eta_n-c)^2)^s]<1$ or $\E[|\eta_n|^{2s}]<\infty$ and $\beta+\alpha< 1/(\E[|\eta_n-c|^{2s}])^{1/s}$, then $\E[|r_n|^{2s}]<\infty$.
\item VGARCH: $\beta <1$, where
\begin{align*}
h_n=\omega+\beta h_{n-1}+\alpha(\eta_{n-1}-c)^2
\end{align*}
and $\omega>0$, $\beta, \alpha\ge 0$. Furthermore, if there exists an integer $s\ge 1$ such that $\beta<1$, $\E[|\eta_n|^{2s}]<\infty$, then $\E[|r_n|^{2s}]<\infty$.
\item TSGARCH: $\beta+\alpha \E|\eta_t|<1$, where 
\begin{align*}
\sqrt{h_n}=\omega+\beta \sqrt{h_{n-1}}+\alpha_1 |\eta_{n-1}| \sqrt{h_{n-1}}
\end{align*}
and $\omega >0$, $\beta\ge 0$, $\alpha_1>0$ and $\alpha_1+\alpha_2\ge 0$. Furthermore, if there exists an integer $s\ge 1$ such that $\E[(\beta+\alpha_1|\eta_n|)^s]<\infty$, then $\E[|r_n|^s]<\infty$.
\item GJR-GARCH: $\beta+\alpha_1+\alpha_2\E\max(0,-\eta_n)^2<1$, where 
\begin{align*}
h_n=\omega+\beta h_{n-1}+\alpha_1 \eta_{n-1}^2 h_{n-1}+\alpha_2\max(0, -\eta)^2h_{n-1}
\end{align*}
and $\omega >0$, $\beta\ge 0$, $\alpha_1>0$ and $\alpha_1+\alpha_2\ge 0$. Furthermore, if there exists an integer $s\ge 1$ such that $\E[(\beta+\alpha_1\eta_n^2+\alpha_2\max(0,-\eta_n)^2)^s]<\infty$, then $\E[|r_n|^{2s}]<\infty$.
\item TGARCH: $\beta+\alpha_1|\eta_n|+\alpha_2\E\max(0,-\eta_n)<1,$ where 
\begin{align*}
\sqrt{h_n}=\omega+\beta \sqrt{h_{n-1}}+\alpha_1 |\eta_{n-1}| \sqrt{h_{n-1}}+\alpha_2\max(0, -\eta)\sqrt{h_{n-1}}
\end{align*}
and $\omega >0$, $\beta\ge 0$, $\alpha_1>0$ and $\alpha_1+\alpha_2\ge 0$. Furthermore, if there exists an integer $s\ge 1$ such that $\E[(\beta+\alpha_1|\eta_n|+\alpha_2\max(0,-\eta_n))^s]<\infty$, then $\E[|r_n|^s]<\infty$.
\end{itemize}
\item PGARCH$(a,b)$ models: $\sum_{i=1}^b \alpha_i+\sum_{j=1}^a  \beta_j<1$ and $\E(|\eta_n|^{2\delta})<\infty$, where
\begin{align*}
r_n&=\sqrt{h_n} \eta_n, \quad n=0, 1,\dots\\
h_{n+1}^\delta& = \omega +\sum_{i=1}^b \alpha_i h_{n+1-i}^\delta |\eta_{n+1-i}|^{2\delta}+\sum_{j=1}^a \beta_j h_{n+1-j}^{\delta}
\end{align*}
and $a,b\ge 1$, $\delta>0$, $\omega>0$, $\alpha_i\ge 0$, $i=1, \dots, b$, $\beta_j\ge 0$, $j=1,\dots, a$. This includes GARCH$(a,b)$ for $\delta=1$ and TSGARCH$(a,b)$ for $\delta=1/2$. Furthermore $\E[|r_n|^{2\delta}]<\infty$.
\item Stochastic autoregressive volatility: $\E|u_n|<\infty$, $\E|\beta+\alpha u_n|<\infty$, where 
\begin{align*}
r_n&=\sigma_n z_n,\\
\log \sigma_n &= \omega+\beta \log \sigma_{n-1}+(\gamma+\alpha \log \sigma_{n-1})u_n
\end{align*}
and $\{z_n\}_{n\in \N}$, $\{u_n\}_{n\in \N}$ are mutually independent i.i.d. variables with zero means and unit variances, $\alpha+\beta>0$, $\alpha+\gamma>0$.
\end{itemize}
\end{Lem}

As a direct consequence we obtain the following:

\begin{Cor}\label{cor:ergodicity}
Let the distribution of $\{r_n\}_{n\in \N}$ be given by one of the models satisfying the corresponding conditions in Lemma \ref{lem:ergodicity} for stationarity and let $\P_1=\P$
\footnote{We note that GARCH is traditionally used for logarithmic returns $\log(S_n/S_{n-1})$ and not raw returns $r_n=S_n/S_{n-1}$ in the econometric literature. As $r \mapsto \log(r)$ is continuous this is no restriction for Theorem \ref{Thm. one-per}, however imposing Assumption \ref{Ass 1} for $r_n=\log(S_n/S_{n-1})$ would mean that we had to demand exponential moments of the GARCH models. Even in a parametric setting, statistical inference in a heavy-tailed GARCH environment is notoriously difficult, see e.g. \citep{hall2003inference}. On the other hand, for short time scales we find $r_n\sim 1$ and the approximation $\log(r_n)\sim r_n-1$ is quite accurate.}.
Then the conditions of Theorem \ref{Thm. one-per} are satisfied. Furthermore Assumption \ref{Ass 1}.\ref{Ass 1.1} is satisfied in the following cases:
\begin{itemize}
\item LGARCH: There exists an integer $q> 3p$ such that $(\beta+\alpha)^q<1$ or $\E[|\eta_n|^{2q}]<\infty$ and $\beta+\alpha< 1/(\E[|\eta_n|^{2q}])^{1/q}$.
\item NGARCH: There exists an integer $q> 3p$ such that $\E[(\beta+\alpha(\eta_n-c)^2)^q]<1$ or $\E[|\eta_n|^{2q}]<\infty$ and $\beta+\alpha< 1/(\E[|\eta_n-c|^{2q}])^{1/q}$.
\item VGARCH: There exists an integer $q>3p$ such that $\beta<1$, $\E[|\eta_n|^{2q}]<\infty$.
\item TS-GARCH: There exists an integer $q> 6p$ such that $\E[(\beta+\alpha_1|\eta_n|)^q]<\infty$.
\item GJR-GARCH: There exists an integer $q>3p$ such that $\E[(\beta+\alpha_1\eta_n^2+\alpha_2\max(0,-\eta_n)^2)^q]<\infty$.
\item TGARCH:  There exists an integer $q>6p$ such that $\E[(\beta+\alpha_1|\eta_n|+\alpha_2\max(0,-\eta_n))^q]<\infty$.
\item PGARCH: $\delta>3p$.
\end{itemize}
\end{Cor}

\begin{proof}
The first claim follows directly from Lemma \ref{lem:ergodicity}. To show the second claim it is sufficient to check that there exists $q>6p$ with $\E[|r_n|^q]<\infty$ and \eqref{eq:poincare} holds. Using again Lemma \ref{lem:ergodicity}, existence of $q>6p$ with $\E[|r_n|^q]<\infty$ is satisfied under the conditions listed in the statement of the corollary.
Lastly we note that $$\|K^n(r_1, \cdot)-\P(\cdot)\|_{\mathrm{TV}}= \frac{1}{2} \sup_{\|f\|_{\infty} \le 1} \left|\E[f(r_n)|r_1]-\E[f(r_1)]\right|$$ and thus conclude that for $0<\rho<1$ in the $\beta$-mixing with exponential decay property of $(r_n)$, as given by Lemma \ref{lem:ergodicity}, $r_1\sim \P$ and all $\|f\|_{\infty}\le 1$ we have
\begin{align*}
(\E[(\E[f(r_n)|r_1]-\E[f(r_1)])^2])^{1/2}
&\le (\E[2|\E[f(r_n)|r_1]-\E[f(r_1)]|])^{1/2}\\
&\le \left(\E\left[2\sup_{\|f\|_{\infty}\le 1}|\E[f(r_n)|r_1]-\E[f(r_1)]|\right]\right)^{1/2}\\
&\le (4 \E\left[\|K^n(r_1, \cdot)-\P(\cdot)\|_{\mathrm{TV}}\right])^{1/2}\\
&\le 2c \rho^{n/2},
\end{align*} 
which is summable.
\end{proof}

A test for (strict) stationarity of GARCH models was developed in \citep{francq2012strict}. Applying their test to daily returns of CAC, DAX, DJA, DJI, DJT, DJU, FTSE, Nasdaq, Nikkei, SMI, and SP500, from January 2, 1990 to January 22, 2009 the authors conclude that non-stationarity of the time series is not plausible.

\subsection{Additional results and proofs for Section \ref{sec. impr}} \label{app. 1}

\begin{proof}[Proof of Lemma \ref{Lem d=1}]
It suffices to argue the case $d=1$ since the for $d\geq 2$ the result is shown in \cite[Thm.~1.1]{trillos2014rate}. The Kiefer-Wolfowitz bounds (see Lemma \ref{lem dkw}) yield $$\P^{\infty}\left(\sup_{r \in \R_+}|F_{\P}(r)- F_{\hat{\P}_N}(r)| \ge N^{-1/4}\right) \le \exp(-2\sqrt{N}).$$
Recall also that $A$ is connected, open and bounded and $\P$ satisfies Assumption \ref{Ass. bilip}. It follows that $F_{\P}(x)+\epsilon \leq F_{\P}(x+\alpha \epsilon)$ for $x\in \R$. We thus conclude that with $\P^{\infty}$-probability greater than $(1-\exp(-2\sqrt{N}))$ we have, for $x\in \R$,
\begin{align*}
F_{\P}(x-\alpha N^{-1/4}) \le F_{\P}(x)- N^{-1/4} &\leq F_{\hat{\P}_N}(x)\leq F_{\P}(x)+N^{-1/4}\\
&\leq F_{\P}(x+\alpha N^{-1/4})
\end{align*}
and in particular $\mathcal{W}^{\infty}(\P, \hat{\P}_N)\leq \alpha N^{-1/4}$.
\end{proof}

\begin{proof}[Proof of Corollary \ref{cor:modeluncertainty}]
The ``$\le$"-inequality follows as in the proof of Theorem \ref{Thm wasserstein}. For the ``$\ge"$-inequality take $\Q_N \in \mathcal{M}$ and $\nu_N \in B^p_{\epsilon+\epsilon_N}(\hat{\P}_N)$ such that $\|d\Q_N/ d\nu_N\|_{\infty} < C_N$. 
By assumption there exists a martingale measure $\Q \sim \P$ such that $\|d\Q / d\P \|_{\infty}\le C-\delta$ for some $\delta>0$.  Define $$\delta_N:=\frac{2|C_N-C|(C-\delta)}{C_N-C+\delta}\vee\left( \frac{2\epsilon_NC_N}{\epsilon+2\epsilon_N}-|C_N-C|\right)$$ for all $N\in \N$. For $N\in \N$ large enough choose $$\lambda_N\in \left(\frac{|C_N-C|+\delta_N}{C_N} ,\frac{\delta_N}{C-\delta}\right).$$
Then $\lambda_N \in (0,1)$,  $\lim_{N \to \infty} \lambda_N=0$  and $(1-\lambda_N)\nu_N+\lambda_N\P \in B^p_{\epsilon}(\P)$ on a set of probability at least  $1-\beta_N$. Furthermore $(1-\lambda_N)C_N\le C-\delta_N$ and thus
\begin{align*}
(1-\lambda_N)C_N+\lambda_N (C-\delta)\le C-\delta_N+\lambda_N(C-\delta)< C-\delta_N+\delta_N=C
\end{align*}
for all $N\in \N$.
Then for $\tilde{\Q}_N:=  (1-\lambda_N)\Q_N+\lambda_N \Q$ we have
\begin{align*}
\frac{d\tilde{\Q}_N}{d((1-\lambda_N)\nu_N+\lambda_N \P)}=\frac{d((1-\lambda_N)\Q_N+\lambda_N \Q)}{d((1-\lambda_N)\nu_N+\lambda_N \P)} <C
\end{align*}
and, as $\Q_N, \Q \in \mathcal{M}$ and for any $\tilde{\Q}\in \Mc$
\begin{align*}
\E_{\tilde{\Q}}[|r|] \le \sqrt{d} \E_{\tilde{\Q}}[|\max_{1\leq i\leq d}r_i|] \le \sqrt{d} \sum_{i=1}^d \E_{\tilde{\Q}}[|r_i|]=\sqrt{d} \sum_{i=1}^d \E_{\tilde{\Q}}[r_i]= d^{3/2},
\end{align*}
there exists $M>0$ such that
\begin{align*}
\E_{\Q_N}[g] &\le (1-\lambda_N) \E_{\Q_N}[g]+ \lambda_N \E_{\Q}[g]+ \lambda_N |\E_{\Q_N}[g]- \E_{\Q}[g]| \\
&\le \sup_{\nu \in B_{\epsilon}^p(\P), \ \|d\Q/ d\nu\|_{\infty}<C} \E_{\Q}[g]+\lambda_N M,
\end{align*}
where the last inequality holds with $\P^\infty$-probability at least $1-\beta_N$. This concludes the proof.
\end{proof}

We recall here the continuity property of $AV@R^{\P}_{1/k_N}(g)$:
\begin{Lem}[\cite{pichler2013evaluations}, Cor. 11, p.538]
\label{Lem. Pichler}
For $g \in \mathcal{L}^1$ and $\P, \tilde\P \in \prob(\R_+^d)$ we have
\begin{align*}
\left|AV@R_{1/k_N}^{\P}(g)-AV@R_{1/k_N}^{\tilde\P}(g)\right| \le k_N\mathcal{W}^1(\P, \tilde{\P}).
\end{align*}
\end{Lem}
This is used for the subsequent proof.
\begin{proof}[Proof of Corollary \ref{Lem:AVaR}]
Recall that 
\begin{align*}
\pi_{\hat{\mathcal{Q}}_N}(g) = \sup_{\tilde{\P} \in B^p_{\epsilon_N(\beta_N)}(\hat{\P}_N)} \sup_{\|d\nu/ d\tilde{\P}\|_{\infty} \le k_N} \inf_{H \in \R^d} \E_{\nu}\left[g-H(r-1)\right].
\end{align*}
We now want to interchange the two suprema and the infimum above, which can be done by the same arguments as in \cite[Proofs of Theorem 3.1 and Lemma 3.2]{bartl2016exponential}. Indeed, as NA$(\P)$ holds and as $\P\in B_{\epsilon_N(\beta_N)}^p(\hat{\P}_N)$ with high probability we conclude by \cite[Theorem 1.48, p.29]{follmer2011stochastic} that  there exists $N_0\in \N$ such that 
$$1 \in \text{ri}(\{\E_{\nu}=[r] \ | \ \exists \tilde{\P} \in B^p_{\epsilon_N(\beta_N)}(\hat{\P}_N) \text{ s.t. } \|d\nu/ d\tilde{\P}\|_{\infty} \le k_N\})$$ 
for all $N\ge N_0$, where $\text{ri}(A)$ denotes the relative interior of the convex hull of a set $A\in \mathcal{B}(\R_+^d)$. This implies
\begin{align*}
\pi_{\hat{\mathcal{Q}}_N}(g) &=  \inf_{H \in \R^d} \sup_{\tilde{\P} \in B^p_{\epsilon_N(\beta_N)}(\hat{\P}_N)} \sup_{\|d\nu/ d\tilde{\P}\|_{\infty} \le k_N}\E_{\nu}\left[g-H(r-1)\right]\\
&= \inf_{H \in \R^d} \sup_{\tilde{\P} \in B^p_{\epsilon_N(\beta_N)}(\hat{\P}_N)} AV@R^{\tilde{\P}}_{1/k_N}\left[g-H(r-1)\right].
\end{align*}
The first inequality in \eqref{eq. avar1} is trivial, while the second inequality follows from 2-Lipschitz-continuity of $g(r)-H(r-1)$ for $|H|\le 1$ and continuity of $\tilde{\P} \mapsto AV@R_{1/k_N}^{\tilde{\P}}(g-H(r-1))$ w.r.t. to $\Wc^1$, see \cite{pichler2013evaluations} or Lemma \ref{Lem. Pichler}.
\end{proof}

\begin{proof}[Proof of Theorem \ref{thm:penalty}]
Clearly 
\begin{align*}
&\lim_{N \to \infty} \sup_{\Q \in \mathcal{M}} \left(\E_{\Q}[g]-C_N \left(\inf_{\hat{\Q} \sim \hat{\P}_N, \ \hat{\Q}\in \mathcal{M}} \left\| \frac{d\hat{\Q}}{d\Q} \right\| _{\infty}-1\right) \right)\\
&\ge\ \lim_{N \to \infty}\sup_{\Q \sim \hat{\P}_N, \ \Q \in \mathcal{M}} \E_{\Q}[g]= \sup_{\Q \sim \P, \ \Q \in \mathcal{M}} \E_{\Q}[g] \hspace{0.5cm} \P^{\infty}\text{-a.s.}
\end{align*}
It remains to establish the reverse inequality. Let us first assume $C_N=C$. Fix $\Q \in \mathcal{M}$ and consider $\hat{\Q} \sim \hat{\P}_N$, $\hat{\Q}\in \Mc$ such that $\hat{\Q} \ll \Q$. Define $\epsilon \geq 0$ through
\begin{align*}
 \Q(r \notin \text{supp}(\hat{\P}_N))+\sum_{i=1}^N (\Q(r_i)-\hat{\Q}(r_i))^+=\epsilon.
\end{align*}
Note that we can write $\Q=f(r)\hat{\Q}+\nu$ for some Radon-Nykodym derivative $f$ and some measure $\nu$ singular to $\hat{\Q}$, where $\nu(\R_+^d)\le \epsilon$. Let $J := \{i \in \{0, \dots, N\} \ | \ \hat{\Q}(r_i) \ge \Q(r_i) \}$. Then we must have 
\begin{align*}
\sum_{i \in J} \hat{\Q}(r_i)- \Q(r_i)=\epsilon.
\end{align*}
In particular 
\begin{align*}
\left\|\frac{1}{f} \right\|_{\infty}=\left\|\frac{d\hat{\Q}}{d\Q} \right\|_{\infty}\ge \frac{1}{1-\epsilon},
\end{align*}
because otherwise
\begin{align*}
\sum_{i \in J} &\hat{\Q}(r_i)- \Q(r_i)<  \frac{\epsilon}{1-\epsilon} \sum_{i \in J} \Q(r_i)\\
&=\frac{\epsilon}{1-\epsilon}\left(1-\Q(r\notin \text{supp}(\hat{\P}_N))-\sum_{i=1}^N (\Q(r_i)-\hat{\Q}(r_i))^+-\sum_{i \in \{1, \dots, N\}\setminus J} \hat{\Q}(r_i)\right)\\
&\le \epsilon,
\end{align*}
a contradiction. Furthermore, by the definition of $\epsilon$, $\E_{\Q}[g]- \E_{\hat{\Q}}[g]\le C \epsilon$. Thus
\begin{align*}
\left(\E_{\Q}[g]-C \left(\left\|\frac{d\hat{\Q}}{d\Q} \right\|_{\infty} -1\right)\right) -\E_{\hat{\Q}}[g] \le C\left(\epsilon-\frac{1}{1-\epsilon}+1\right)=\frac{-C\epsilon^2}{1-\epsilon}\leq 0,
\end{align*}
so there is no gain from shifting mass and in particular
\begin{align*}
\sup_{\hat{\Q} \sim \hat{\P}_N, \ \hat{\Q} \in \mathcal{M}} \left(\E_{\Q}[g]-C \left(\left\|\frac{d\hat{\Q}}{d\Q} \right\|_{\infty} -1\right)\right) \le  \sup_{\hat{\Q} \sim \hat{\P}_N, \ \hat{\Q} \in \mathcal{M}} \E_{\hat{\Q}}[g].
\end{align*}
For $C_N\to C$ we have
\begin{align*}
 \left(\E_{\Q}[g]-C_N \left(\left\|\frac{d\hat{\Q}}{d\Q} \right\|_{\infty} -1\right)\right) -\E_{\hat{\Q}}[g] &\le C\epsilon-C_N\left(\frac{1}{1-\epsilon}-1\right)\\
 &=C\epsilon-C_N\left(\frac{\epsilon}{1-\epsilon}\right),
\end{align*}
which is non-negative for $\epsilon\le 1-C_N/C$. Finally, the case
when $\hat{\Q}$ is not absolutely continuous with respect to $\Q$ is trivial since then $\| d\hat{\Q}/d\Q\|_{\infty}= \infty$.
\end{proof}

\subsubsection{Proof of Theorem \ref{Thm wasserstein}}
For simplicity of exposition we first prove Theorem \ref{Thm wasserstein} under Assumption \ref{Ass 1}.\ref{Ass 1.2}. We adopt the notation of Section \ref{sec. impr}.

\begin{Lem}[\cite{fournier2015rate} Theorem 2]\label{Lem. was}
Under Assumption \ref{Ass 1}.\ref{Ass 1.2} we have
\begin{align*}
\P^{\infty}(\mathcal{W}^p(\P, \hat{\P}_N) \ge \epsilon ) \le 
\left\{\begin{array}{ll}
					c_1 \exp(-c_2 N \epsilon^{\min(\max(d/p,2),a/(2p))}) &\text{ if }\epsilon \le 1, \\
					c_1\exp(-c_2N\epsilon^{a/(2p)}) &\text{ if } \epsilon >1
			\end{array}
				\right. 
\end{align*}
for $N \ge 1$, $d \neq 2p$ and $\epsilon>0$, where $c_1, c_2$ are positive constants that only depend on $p$, $d$, $a$, $c$ and $\E_{\P}[\exp(c|r|^a)]$. Thus for some confidence level $\beta \in (0,1)$ we can choose
\begin{align}\label{eq:def_epsilonN}
\epsilon_N(\beta) := \left\{\begin{array}{ll}
					\left(\frac{\log(c_1\beta^{-1})}{c_2 N}\right)^{1/\min(\max(d/p,2),a/(2p))} &\text{ if }N \ge \frac{\log(c_1\beta^{-1})}{c_2}, \\
					\left(\frac{\log(c_1\beta^{-1})}{c_2 N}\right)^{(2p)/a} &\text{ if } N<\frac{\log(c_1\beta^{-1})}{c_2}
			\end{array}
				\right.
\end{align}
which yields $\P^{\infty} (\mathcal{W}^p(\P, \hat{\P}_N) \ge \epsilon_N(\beta)) \le \beta$.
\end{Lem}

\begin{Lem}\label{lem. comp}
Fix $N \in \N$ and $p\ge 1$. Let $\Q_n \in D^p_{\epsilon_N(\beta_N),k_N}(\P)$ such that $\Q_n \Rightarrow \Q \in \mathcal{P}(\R_+^d)$ for $n \to \infty$. Then $|E_{\Q}[r-1]| \le K k_N \epsilon_N(\beta_N)$ for some $K>0$. $D^p_{\epsilon_N(\beta_N),k_N}(\P)$ is weakly compact for $p>1$. In general $D^1_{\epsilon_N(\beta_N),k_N}(\P)$ is not weakly closed. 
\end{Lem}

\begin{proof}
Take a sequence $\Q_n \in D^p_{\epsilon_N(\beta_N),k_N}(\P)$ such that $\Q_n \Rightarrow \Q$. For every $n \in \N$ there exists $\nu_n \in B^p_{\epsilon_N(\beta_N)}(\P)$ such that $\|d\Q_n / d\nu_n \|_{\infty} \le k_N$. First observe that for any $n$ we have
\begin{equation}\label{eq:wassballbound}
\E_{\Q_n}[r\land K]\leq k_N\left(\E_{\P}[r\land K]+\epsilon_N(\beta_N)\right)\leq k_N\left(\E_{\P}[r]+\epsilon_N(\beta_N)\right)<\infty.
\end{equation}
It follows, by weak convergence and monotone convergence theorem, that $\E_{\Q}[r]\leq k_N\left(\E_{\P}[r]+\epsilon_N(\beta_N)\right)<\infty$. Next, we show $|\E_{\Q}[r-1]| \le \text{const}\cdot k_N \epsilon_N(\beta_N)$. As 
\begin{align*}
|\E_{\Q}[r-1]| \le \sqrt{d} \max_{1\leq i\leq d}|\E_{\Q}[r_i-1]| 
\leq \sqrt{d} \sum_{i=1}^d |\E_{\Q}[r_i-1]|
\end{align*}
it is enough to consider the case $d=1$. Then
\begin{align}\label{eq. est}
|\E_{\Q}[r-1]|\leq &  \Big|\E_{\Q}[(r-1) \wedge K]-\E_{\Q_n}[(r-1)\wedge K]\Big|\nonumber\\
&+\Big| \E_{\Q}\left[(r-1-K)\mathds{1}_{\{r\ge K+1\}}\right]-\E_{\Q_n}\left[(r-1-K)\mathds{1}_{\{ r\ge K+1\}}\right]\Big|
\end{align}
Consider now the terms on the RHS. 
The third term can be made arbitrarily small by taking large $K$ since $\Q$ admits first moment. The fourth term can be bounded, in analogy to \eqref{eq:wassballbound}, as follows:
\begin{align*}
\E_{\Q_n}[(r-K)\mathds{1}_{\{r\ge K\}}] &\le k_N  \E_{\P}[(r-K)\mathds{1}_{\{r\ge K\}}]+k_N\epsilon_N(\beta_N)\leq 2 k_N \epsilon_N(\beta_N),
\end{align*}
where we took $K$ large enough. Finally, for a fixed $K$, the difference between the first two terms can be made small by taking $n$ large due to weak convergence of measures.
The bound  $|\E_{\Q}[r-1]| \le \text{const}\cdot k_N \epsilon_N(\beta_N)$ follows. 

If $p>1$, the ball $B^p_{\epsilon_N}(\P)$ is weakly compact (see \cite[Def. 6.8., p.96]{villani2008optimal}), in particular uniformly integrable, and it follows that also the fourth term on the RHS of \eqref{eq. est} converges to zero, uniformly in $n$, as $K\to \infty$
\begin{align*}
\E_{\Q_n}[(r-1-K)\mathds{1}_{\{ r\ge K+1\}}]&\le k_N \limsup_{K \to \infty} \sup_{n}\E_{\nu_n}[r\mathds{1}_{r\geq K}]=0.
\end{align*}
Further, possibly on a subsequence, $(\nu_n)_{n \in \N}$ converges weakly to a limit $\nu\in B^p_{\epsilon_N}(\P)$. By regularity of probability measures it is sufficient to test $d\Q/ d\nu$ against bounded continuous functions and we easily conclude that $\|d\Q/ d\nu\|_{\infty} \le k_N$. This shows compactness of $D^p_{\epsilon_N(\beta_N),k_N}(\P)$.

Consider now $p=1$. We give the following counterexample: take $\P=\delta_1$ and set for $r_n \ge 2$
\begin{align*}
\nu_n=\frac{\epsilon_N(\beta_N)}{2}\ \delta_0+\left(1-\frac{r_n \ \epsilon_N(\beta_N)}{2(r_n-1)}\right) \ \delta_1+\frac{\epsilon_N(\beta_N)}{2(r_n-1)}\ \delta_{r_n}.
\end{align*}
Furthermore let
\begin{align*}
\Q_n=\frac{\epsilon_N(\beta_N)}{2\sqrt{\epsilon_N(\beta_N)}} \ \delta_0+\left(1-\frac{r_n \epsilon_N(\beta_N)}{2(r_n-1)\sqrt{\epsilon_N(\beta_N)}}\right)\ \delta_1+\frac{\epsilon_N(\beta_N)}{2(r_n-1)\sqrt{\epsilon_N(\beta_N)}} \ \delta_{r_n}.
\end{align*}
Then $\mathcal{W}^1(\nu_n, \delta_1) \le \epsilon_N(\beta_N)$
and
\begin{align*}
\left\|\frac{d\Q_N}{d\nu_N} \right\|_{\infty}\le \frac{1}{\sqrt{\epsilon_N(\beta_N)}}.
\end{align*}
Assume $\lim_{n \to \infty} r_n = \infty$. Then
\begin{align*}
\Q_n \Rightarrow\frac{\sqrt{\epsilon_N(\beta_N)}}{2} \ \delta_0+\left(1-\frac{ \sqrt{\epsilon_N(\beta_N)}}{2}\right)\ \delta_1 \notin \mathcal{M}.
\end{align*}
\end{proof}

Taking the closure of $D^1_{\epsilon_N(\beta_N),k_N}(\P)$ would ensure compactness, but consistency of the estimator $\sup_{\Q \in \hat{\Qc}_N} \E_{\Q}[g]$ in Theorem \ref{Thm wasserstein} would be lost in general since the closure might include non-martingale measures. To see this, take for instance $g(r)=(r-1)$ in the example in the proof of Lemma \ref{lem. comp}.

\begin{proof}[Proof of Theorem \ref{Thm wasserstein} under Assumption \ref{Ass 1}.\ref{Ass 1.2}.]
Let us assume that Assumption \ref{Ass 1}.\ref{Ass 1.2} is satisfied.
Note that the ``$\ge$"-inequality follows from Lemma \ref{lem. rasonyi}. Indeed, Lemma \ref{lem. rasonyi} implies that for all $N \in \N$ there exists a martingale measure $\Q_N \sim \P$ with $\| d\Q_N /d\P \|_{\infty} \le k_N$ such that
\begin{align*}
\sup_{\Q \sim \P, \ \Q \in \mathcal{M}} \E_{\Q}[g] \le \E_{\Q_N}[g]+a_N, 
\end{align*}
with $a_N \to 0$ as $N \to \infty$. With $\P^{\infty}$-probability $(1-\beta_N)$ we have $\P \in B^p_{\epsilon_N(\beta_N)}(\hat{\P}_N)$ and hence $\Q_N\in \hat{\Qc}_N$. This gives
\begin{align*}
\P^{\infty}\left(\sup_{\Q \sim \P, \ \Q \in \mathcal{M}} \E_{\Q}[g]-a_N \le \sup_{\Q \in \hat{\Qc}_N} \E_{\Q}[g]\right)\ge 1-\beta_N,\quad N\geq 1.
\end{align*}
We recall that NA$(\P)$ gives $\pi^\P(g)=\sup_{\Q \sim \P, \ \Q \in \mathcal{M}} \E_{\Q}[g]$ and hence, for any $\epsilon>0$ and $N$ large enough, $\P^{\infty}(\pi_{\hat{\Qc}_N}(g)-\pi^\P(g)\leq -\epsilon)\leq \beta_N\to 0$ as $N\to \infty$. 
\\
For the ``$\le$"-inequality, we assume $g$ is Lipschitz continuous bounded from below and we take $H \in \R^d$ such that
\begin{align*}
\pi^\P(g)+H(r-1) \ge g \quad \P\text{-a.s.}
\end{align*}
Take a sequence $\Q_N\in \hat{\Qc}_N$ with $\E_{\Q_N}[g]\geq \pi_{\hat{\Qc}_N}(g)-a_N$. By definition, there exist $\nu_N\in B^p_{\epsilon_N(\beta_N)}(\hat{\P}_N)$ such that $\|d\Q_N /d\nu_N\|_{\infty}\le k_N$. In particular, with $\P^{\infty}$-probability $(1-\beta_N)$ we have $\nu_N \in B^p_{2\epsilon_N(\beta_N)}(\P)$. Let us define
\begin{align*}
A= \{ \pi^\P(g)+H(r-1)-g \ge 0\}.
\end{align*}
Then
\begin{align}\label{eq. lip1}
\nonumber\E_{\Q_N}[g] &\le \E_{\Q_N}[(\pi^\P(g)+H(r-1))\mathds{1}_A+g\mathds{1}_{A^c}] \\
&=\pi^\P(g)+\E_{\Q_N}[H(r-1)]+ \E_{\Q_N}\left[(g-H(r-1)-\pi^\P(g))\mathds{1}_{A^c}\right]
\end{align}
and the second term on the RHS vanishes since $\Q_N$ is a martingale measure. To treat the last term on the RHS consider the function
\begin{align*}
\tilde{g}:=(g-H(r-1)-\pi^\P(g))\vee 0
\end{align*}
which is non-negative, $C$-Lipschitz for some $C>0$ and $\{\tilde{g}>0\}=A^c$. Since $\P(A^c)=0$ we have in particular
\begin{align*}
\left| \int_{A^c}\tilde{g}d\nu_N \right|=\left|\int_{A^c} \tilde{g}d\nu_N-\int_{A^c}\tilde{g}d\P\right|=\left|\int \tilde{g}d\nu_N-\int\tilde{g}d\P\right|
\end{align*}
which by the Kantorovitch-Rubinstein duality \eqref{eq:KRduality} is dominated by 
$C\mathcal{W}^1(\nu_N, \P)\le C\mathcal{W}^p(\nu_N, \P)$. We conclude that, for any $\epsilon>0$, 
$$\P^\infty\left(\pi_{\hat{\Qc}_N}(g)-\pi^\P(g)\geq \epsilon\right) \leq \P^\infty\left(a_N+Ck_N\mathcal{W}^p(\nu_N, \P)\geq \epsilon\right)\leq \beta_N
$$
for $N$ large enough since $\epsilon_Nk_N\to 0$. This establishes the convergence of $\pi_{\hat{\Qc}_N}(g)$ to $\pi^\P(g)$ in $\P^\infty$-probability. Further, whenever $\sum_{N=1}^\infty\beta_N<\infty$, a simple application of Borel-Cantelli lemma, similarly as in \cite[Lemma 3.7]{esfahani2015data}, shows that the convergence holds $\P^\infty$-a.s.
This concludes the proof in the case of Lipschitz continuous $g$ bounded from below and under Assumption \ref{Ass 1}.\ref{Ass 1.2}. \\

It remains to argue the ``$\le$"-inequality when $g$ is bounded and continuous. We fix a small $\delta>0$ and define
\begin{align*}
\tilde{\mathcal{Q}}_N:=\left\{\Q\in \mathcal{P}(\R^d_+)\ | \ \exists \tilde\P \in B^p_{\epsilon_N(\beta_N)}(\hat{\P}_N) \text{ such that } \|d\Q/d\tilde\P\|_{\infty}\le k_N/(1-\delta)\right\}.
\end{align*}
and set 
\begin{align*}
\pi_{\hat{\mathcal{Q}}_N,[0,K]^d}(g):= \sup_{\Q \in \tilde{\mathcal{Q}}_N, \Q\in \mathcal{M}, \ \text{supp}(\Q)\subseteq[0,K]^d} \E_{\Q}(g).
\end{align*}
Similarly to Corollary \ref{Lem:AVaR} we see that for $K>1$ large enough
\begin{align*}
&\pi_{\hat{\mathcal{Q}}_N,[0,K]^d}(g)\\
&\qquad =
\inf \Bigg\{x \in \R \ | \ \exists H \in \R^d \text{ s.t. }\sup_{\Q\in \tilde{\mathcal{Q}}_N, \ \text{supp}(\Q) \subseteq [0,K]^d}\E_{\Q}[g(r)-H(r-1)-x] \le 0\Bigg\}.\nonumber
\end{align*}
Thus there exists a sequence $H_N\in \R^d$ such that
\begin{align}\label{eq:avar_j}
\sup_{\Q\in \tilde{\mathcal{Q}}_N, \ \text{supp}(\Q) \subseteq [0,K]^d}\E_{\Q}[g(r)-H_N(r-1)-\pi_{\hat{\mathcal{Q}}_N,[0,K]^d}(g)] \le 1/N
\end{align}
for all $N \in \N$. Take $K$ large enough so that $d/K<\delta$. For notational simplicity we assume that $\P$ has full support. Recall that $g$ is bounded, $|g|\leq C$. We now show that $H_N^i$ is bounded. Let us first look at the lower bound: for this, we take $i\in \{1,\dots,d\}$ and suppose that $H_N^i\le 0$ (otherwise we trivially bound $H_N^i$ from below by $0$). We work on the set $\{\P\in B^p_{\epsilon_N(\beta_N)}(\hat{\P}_N)\}$, which has $\P^\infty$-probability at least $1-\beta_N$. Now we take $N$ large enough such that 
$$k_N\cdot \P\left(s_j r_j\ge 1 \text{ for all }j\neq i,\  K/2<r_i< K\right)\ge 1$$
for all $(s_1, \dots,s_{i-1},s_{i+1},\dots, s_d) \in \{-1,1\}^{d-1}$.
Then defining $s_j=-\text{sign}(H^j_N)$ for all $j\neq i$ and
$$\frac{d\Q^{s_1,\dots,s_{i-1}, s_{i+1},s_{d}}}{d\P}:=\frac{ \mathds{1}_{\left\{s_j r_j\ge 1\text{ for all }j\neq i,\  K/2<r_i< K\right\}}}{\P(s_j r_j\ge 1 \text{ for all }j\neq i,\  K/2<r_i< K)}$$ we have by \eqref{eq:avar_j}
\begin{align*}
-H_N^i (K/2-1)\le -H_N \E_{\Q}[r-1]\le \frac{1}{N}-\E_{\Q}[g(r)]+\pi_{\hat{\mathcal{Q}}_N}(g)\le 2 C+1
\end{align*}
and thus $H_N^i \ge -(2(C+1))/(K/2-1)$. 
On the other hand, assuming $H_N^i\ge 0$ and setting
$$\frac{d\Q^{s_1,\dots,s_{i-1}, s_{i+1},s_{d}}}{d\P}:=\frac{ \mathds{1}_{\left\{s_j r_j\ge 1\text{ for all }j\neq i,\  0<r_i< 1/2\right\}}}{\P(s_j r_j\ge 1 \text{ for all }j\neq i,\  0<r_i< 1/2)}$$ we got by the same arguments as above
\begin{align*}
H_N^i/2 \le -H_N \E_{\Q}[r-1]\le \frac{1}{N}-\E_{\Q}[g(r)]+\pi_{\hat{\mathcal{Q}}_N}(g)\le 2 C+1
\end{align*}
for a possibly larger $N$. In conclusion thus $ -(2(C+1))/(K/2-1)\le H_N^i \le 4C+2$. Now we set $$\tilde{H}^i_N:=(C-\pi_{\hat{\mathcal{Q}}_N,[0,K]^d}(g))/(K-1)\vee H^i_N, \qquad i\in \{1, \dots, d\}.$$ Obviously $\tilde{H}^i_N\ge 0$ and 
\begin{align}\label{eq:late}
\pi_{\hat{\mathcal{Q}}_N,[0,K]^d}(g)+\tilde{H}^{i}_N(r^{i}-1)\ge C
\end{align}
for $r^i \ge K$. Furthermore as $(C-\pi_{\hat{\mathcal{Q}}_N,[0,K]^d}(g))/(K-1) \le 2C/(K-1)$ and $H^i_N \ge -2(C+1)/(K/2-1)$ we have for $r^i\le 1$ 
\begin{align}\label{eq:late2}
0\ge (\tilde{H}^i_N-H^i_N)(r^i-1)\ge -(\tilde{H}^i_N-H^i_N)\ge -4(C+1)/(K/2-1).
\end{align}
Consider now $\Q_N\in \hat{\mathcal{Q}}_N$. Note that by Markov's inequality we have 
$$\Q_N([0,K]^d)\geq 1-\frac{d}{K}\ge 1-\frac{\delta}{4C+2}\ge 1-\delta.$$ 
Then 
\begin{align*}
& \E_{\Q_N}\left[g(r)-\tilde{H}_N(r-1)-\pi_{\hat{\mathcal{Q}}_N,[0,K]^d}(g)-4d(C+1)/(K/2-1)-d^2(4C+2)/K \right]\\
&\le  \E_{\Q_N}\left[\left(g(r)-\tilde{H}_N(r-1)-\pi_{\hat{\mathcal{Q}}_N,[0,K]^d}(g)-4d(C+1)/(K/2-1)\right)\mathds{1}_{[0,K]^d}\right]\\
&\qquad+\E_{\Q_N}\left[\left(g(r)-\tilde{H}_N (r-1)-\pi_{\hat{\mathcal{Q}}_N,[0,K]^d}(g)-d^2(4C+2)/K \right)\mathds{1}_{([0,K]^d)^c}\right]\\
&\le \E_{\Q_N}\left[\left(g(r)-H_N(r-1)-\pi_{\hat{\mathcal{Q}}_N,[0,K]^d}(g)\right)\mathds{1}_{[0,K]^d}\right]\\
&\le \sup_{\Q\in \tilde{\mathcal{Q}}_{N}, \ \text{supp}(\Q) \subseteq [0,K]^d}\E_{\Q}\left[g(r)-H_N(r-1)-\pi_{\hat{\mathcal{Q}}_N,[0,K]^d}(g)\right]\leq \frac{1}{N}.
\end{align*}
where we used  \eqref{eq:late} and \eqref{eq:late2} as well as $|H_N^j|\Q_N(([0,K]^d)^c)\le d(4C+2)/K$ for $j\neq i$ in the second inequality. In the last inequality we used \eqref{eq:avar_j}.
By definition of $\pi_{\hat{\mathcal{Q}}_N}(g)$, and since $\E_{\Q_N}[\tilde H_N(r-1)]=0$, this shows 
\begin{align*}
 \pi_{\hat{\mathcal{Q}}_N}(g)\le \frac{4d(C+1)}{K/2-1}+\frac{d^2(4C+2)}{K}+\frac{1}{ N}+\pi_{\hat{\mathcal{Q}}_{N}, [0,K]^d}(g).
\end{align*}
We obtain the same result in the case of $\P$ without full support with a possible addition of the term constant times $\epsilon_N k_N$ on the RHS. Note that we either have $r_i=1$ $\P$-a.s., or else by NA$(P)$, $\P$ puts mass on $r_i$ on either side of $1$. If $r_i$ is bounded under $\P$, we can take $\tilde H_N^i=H_N^i$ and the arguments remain the same but we have the additional error term as $\Q_N$ can put small mass on unbounded $r_1$.

Lastly we fix $\epsilon>0$, $K>1$ such that $$\frac{4d(C+1)}{K/2-1}+\frac{d^2(4C+2)}{K} \le \epsilon$$ and take a sequence $(\Q_N)_{N \in \N}$ which satisfies $\pi_{\hat{\mathcal{Q}}_N,[0,K]^d}(g) \le \epsilon +\E_{\Q_N}[g]$ for all $N \in \N$. Note that all $(\Q_N)_{N \in \N}$ are martingale measures supported on $[0,K]^d$ .
As $g$ is uniformly continuous on $[0,K+1]^d$, there exists a constant $L>0$ such that $|g(r)-g(\tilde{r})|\le \epsilon+L|r-\tilde{r}|$ for all $r, \tilde{r}\in [0,K+1]^d$. Define
$$\hat{g}(r):=\sup_{u\in [0,K]^d}\left( g(u)-(L\vee 2C)|u-r|-\epsilon\right)$$
and note that $g(r)-\epsilon\le \hat{g}(r)$ on $[0,K]^d$ as well as $\hat{g}(r)\le g(r)$ on $[0,K+1]^d$ and $\hat{g}(r)\le -C$ on $([0,K+1]^d)^c$. In conclusion  $\hat{g}(r)\le g(r)$. For fixed $u\in [0,K]^d$ we write 
\begin{align*}
g(u)-(L\vee 2C)|u-r|-\epsilon\le g(u)-(L\vee 2C)|u-r'|-\epsilon+(L\vee 2C)|r-r'|
\end{align*}
Taking suprema over $u\in [0,K]^d$ on both sides we conclude
\begin{align*}
\hat{g}(r)\le \hat{g}(r')+(L\vee 2C)|r-r'|.
\end{align*}
Exchanging the roles of $r$ and $r'$ finally yields $|\hat{g}(r)-\hat{g}(r')|\le (L\vee 2C)|r-r'|$ for $r,r'\in \R_+^d$. Thus we can argue as in the proof of Theorem \ref{Thm wasserstein} for Lipschitz continuous claims (and using the same notation) to obtain
\begin{align*}
\pi_{\hat{\mathcal{Q}}_N}(g)-\pi^{\P}(g)&\le \pi_{\hat{\mathcal{Q}}_N,[0,K]^d}(g)-\pi^{\P}(g)+\epsilon+\frac{1}{N}\\
&\le \E_{\Q_N}[g]-\pi^{\P}(g)+2\epsilon+\frac{1}{N}\\
&\le \E_{\Q_N}[\hat{g}]-\pi^{\P}(g)+3\epsilon+\frac{1}{N}\\
&\le  \pi^{\P}(\hat{g})-\pi^{\P}(g)+3\epsilon+\hat{C}k_N\mathcal{W}^p(\hat{\P}_N,\P)+\frac{1}{N}\\
&\le \pi^{\P}(g)-\pi^{\P}(g)+3\epsilon+\hat{C}k_N\mathcal{W}^p(\hat{\P}_N,\P)+\frac{1}{N}.
\end{align*}
for some $\hat{C}>0$. In particular
\begin{align*}
\P^{\infty}\left(\pi_{\hat{\mathcal{Q}}_N}(g)-\pi^{\P}(g)\ge \epsilon\right)\le \P^{\infty}\left(3\epsilon+\hat{C}k_N\mathcal{W}^p(\hat{\P}_N,\P)+\frac{1}{N} \ge 4\epsilon\right)\le \beta_N.
\end{align*}
As $\epsilon$ was arbitrary and $\epsilon_N k_N\to 0$ the claim follows.
\end{proof}

We have used the following lemma:
\begin{Lem}[\cite{rasonyi2002note}, Cor. 3.3] \label{lem. rasonyi}
For a measurable function $g$ bounded from below
\begin{align*}
\sup_{\| \frac{d\Q}{d\P}\|_{\infty}< \infty, \ \Q \in \mathcal{M}} \E_{\Q}[g]=\sup_{\Q \sim\P, \ \Q \in \mathcal{M}} \E_{\Q}[g].
\end{align*}
\end{Lem}

Before we finish the proof of Theorem \ref{Thm wasserstein} we first give the convergence rates for $\hat{\P}_N$ under Assumption \ref{Ass 1}.\ref{Ass 1.1}.\footnote{According to \cite[Comments after Theorems 14 \& 15]{fournier2015rate} the $L^2$-$L^2$-decay property stated in \citep[Theorem 15]{fournier2015rate} is actually too strong as only functions bounded by one are considered, as given in our assumptions.} This is a slight modification of the result in \citep{fournier2015rate} but the proof is essentially the same as that of Theorem 15 therein and is hence omitted. 
\begin{Lem}[\citep{fournier2015rate}, Proof of Theorem 15]\label{lem:fournilin}
Under Assumption \ref{Ass 1}.\ref{Ass 1.1} there exists a constant $C>0$ such that
\begin{align*}
\E_{\P}&[ \mathcal{W}^p(\P, \hat{\P}_N)^p] \le \kappa_N\\
&:=C \begin{cases}
N^{-1/2}+N^{-(q_s-p)/q_s} &\text{if }p>d_s/(2s) \text{ and }q_s\neq 2p,\\
N^{-1/2}\log(1+N)+N^{-(q_s-p)/q_s}  &\text{if }p=d_s/(2s) \text{ and }q_s\neq 2p,\\
N^{-p/d}+N^{-(q_s-p)/q_r} &\text{ if }p\in (0,d_s/2) \text{ and }q_s\neq d_s/(d_s-p)
\end{cases}
\end{align*}
and $q_s:= q(s-2)/(2s)$ and $d_s:= d(3s+2)/(2s)$.
\end{Lem}

\begin{proof}[Proof of Theorem \ref{Thm wasserstein} under Assumption \ref{Ass 1}.\ref{Ass 1.1}]
By Assumption \ref{Ass 1}.\ref{Ass 1.1} we have for $\epsilon\ge 0$
\begin{align*}
\P (\mathcal{W}^p(\P, \hat{\P}_N)\ge \epsilon) \le \kappa_N/\epsilon^p
\end{align*}
using Markov's inequality, so in particular we can choose
\begin{align}\label{eq:fournilin}
\epsilon_N(\beta):=\left(\frac{\kappa_N}{\beta}\right)^{1/p}.
\end{align} 
The rest of the proof now follows as under Assumption \ref{Ass 1}.\ref{Ass 1.2}.
\end{proof}

\subsection{Additional results and proofs for Section \ref{sec:robust}}\label{app:robust}

Before we prove Theorem \ref{Thm. infi}, we recall the following result:
\begin{Lem}[\cite{shorack2009empirical} Theorem 26.1, p. 828 \& Ex. 26.2, p.833]\label{Lem. shorack}
Let $\mathcal{C}$ denote all closed balls in $\R_+^d$. Then 
\begin{align*}
\P^{\infty}\left(\lim_{N \to \infty} \sup_{\overline{B} \in \mathcal{C}} |\hat{\P}_N(\overline{B})-\P(\overline{B})|=0\right)=1 \quad \text{ uniformly for all }\P \in \mathcal{P}(\R_+^d)
\end{align*}
\end{Lem}

\begin{proof}[Proof of Theorem \ref{Thm. infi}] 
Let $\epsilon>0$ and fix $\P^0 \in \mathcal{P}(\R_+^d)$ such that NA$(\P^0)$ holds. 
As $g$ is uniformly continuous, its $\P$-concave envelope operator is continuous at $\P_0$, see Proposition \ref{prop:d_H cont}, and we can take $\delta$ small enough so that for all $\bar \P,\tilde \P\in \prob(\R^d_+)$ with $d_H(\supp(\bar{\P}),\supp(\tilde \P))\leq 2\delta$ we have $|\pi^{\bar\P}(g) -\pi^{\tilde{\P}}(g)|\leq \epsilon/9$.

We first argue that we can restrict to a compact set. Indeed we have
\begin{align}\label{eq. start}
|\pi^{\hat{\P}^0_N}(g)-\pi^{\hat{\P}^1_N}(g)| &\le |\pi^{\hat{\P}^0_N}(g)-\pi^{\P^0}(g)|+|\pi^{\P^0}(g)-\pi^{\P^1}(g)|\\
&+|\pi^{\P^1}(g)-\pi^{\hat{\P}^1_N}(g)|\nonumber.
\end{align}
We choose $\P^1$ such that $\mathcal{W}^{\infty}(\P^1, \P^0)\le \delta/4$, i.e. $\P^1(B) \le \P^0(B^{\delta/4})$ and $\P^0(B) \le \P^1(B^{\delta/4})$ for all $B \in \mathcal{B}(\R_+^d)$. Thus 
\begin{align*}
\text{supp}(\P^0) \subseteq \text{supp}(\P^1)^{\delta/2}  \quad \text{and} \quad \text{supp}(\P^1)  \subseteq \text{supp}(\P^0)^{\delta/2}.
\end{align*}
Note also that, by the choice of $\delta$, we have
\begin{align}\label{eq:boundondiffofpi}
|&\pi^{\P^0(\cdot|[0,L]^d)}(g)-\pi^{\P^1(\cdot|[0,L]^d)}(g) | \le \epsilon/9 \quad \text{and} \quad |\pi^{\P^0}(g)-\pi^{\P^1}(g) | \le \epsilon/9
\end{align} 
By a monotone limit argument, see the proof of Theorem \ref{Thm. one-per}, we have
\begin{align*}
\sup_{L \in \N} \pi^{\P^0(\cdot|[0,L]^d)}(g)
 =\pi^{\P^0}(g),
\end{align*}
so there for all $L$ large enough 
\begin{align*}
&| \pi^{\P^0(\cdot|[0,L]^d)}(g)-\pi^{\P^0}(g)| \le \epsilon/9.
\end{align*}
Fix such $L$ and set $K:=[0,L]^d$. Note that by \eqref{eq:boundondiffofpi} we now have
\begin{align*}
\pi^{\P^1}(g)-\pi^{\P^1(\cdot|K)}(g) &\le \pi^{\P^1}(g)-\pi^{\P^0}(g)+ \pi^{\P^0}(g)-\pi^{\P^0(\cdot|K)}(g)\\
&+ \pi^{\P^0(\cdot|K)}(g)-\pi^{\P^1(\cdot|K)}(g)\le \epsilon/9+\epsilon/9+\epsilon/9 = \epsilon/3.
\end{align*}
In particular for $i=0,1$
\begin{align}\label{eq. boundab}
0 \le \pi^{\P^i}(g)-\pi^{\hat{\P}^i_N}(g) &= \pi^{\P^i}(g)-\pi^{\P^i(\cdot|K)}(g)+\pi^{\P^i(\cdot|K)}(g)- \pi^{\hat{\P}^i_N}(g)\\&\le \epsilon/3+\pi^{\P^i(\cdot|K)}(g)- \pi^{\hat{\P}^i_N}(g)\nonumber\\
&\le \epsilon/3+\pi^{\P^i(\cdot|K)}(g)- \pi^{\hat{\P}^i_N(\cdot|K)}(g)\nonumber.
\end{align}
We proceed to bound the difference of the last two terms on the RHS. Let us write $\text{supp}(\hat{\P}^i_N)=\{r_1^i, \dots, r_N^i\}$, $i=0,1$. We now argue that 
\begin{align}\label{eq. aim}
\P^{\infty}(d_H(\{r_1^i, \dots, r_N^i\}\cap K, \text{supp}(\P^i)\cap K)> 2\delta) \le \epsilon/2,\quad i=0,1,
\end{align}
for all $N\geq N_*$ for some $N_*$ independent of $\P^1$. To this end, take $M \in \N$ and deterministic points $\tilde{r}_1, \dots, \tilde{r}_M \in \text{supp}(\P^0)\cap K$ such that 
$$\cup_{k=1}^M B_{\delta/2}(\tilde{r}_k) \supseteq  \text{supp}(\P^0)^{\delta/2} \cap K\supseteq \text{supp}(\P^1) \cap K, $$
where $B_\delta(\tilde r)=\{r\in \R^d: |r-\tilde r|<\delta\}$. 
Then as $M$ is finite, there exists $N_0 \in \N$ such that for all $N \ge N_0$
\begin{align*}
\P^{\infty}(\{\forall k\in \{1, \dots, M\}\ \exists j \in \{1, \dots, N\} \text{ s.t. } |\tilde{r}_k-r^0_j|\le \delta\}) \ge 1-\epsilon/2.
\end{align*}
Set 
\begin{align*}
\alpha&:= \min_{k = 1, \dots, M} \P^{0}(B_{\delta/4}(\tilde{r}_k))>0.
\end{align*}
Then $\P^1(B_{\delta/2}(\tilde{r}_k)) \ge \P^0(B_{\delta/4}(\tilde{r}_k)) \ge \alpha$ for all $k=1, \dots, M$. By Lemma \ref{Lem. shorack}  there exists $N_* \ge N_0$ such that for all $N \ge N_*$ and all $\P \in \mathcal{P}(\R_+^d)$
\begin{align*}
\P^{\infty}\left( \sup_{\overline{B} \in \mathcal{C}} |\hat{\P}_N(\overline{B})-\P(\overline{B})| \ge \alpha \right) < \epsilon/2,
\end{align*}
If there exists $k \in \{1, \dots, M\}$ such that for all $j\in \{1, \dots, N\}$ $|\tilde{r}_k- r_j^1| \ge \delta$ then $$|\hat{\P}_N^1(\overline{B}_{\delta/2}(\tilde{r}_k))-\P^1(\overline{B}_{\delta/2}(\tilde{r}_k))|=\P^1(\overline{B}_{\delta/2}(\tilde{r}_k)\ge \alpha,$$ in particular $\sup_{\overline{B} \in \mathcal{C}} |\hat{\P}^1_N(\overline{B})-\P^1(\overline{B})| \ge\alpha$. Thus
\begin{align*}
\P^{\infty} \left( \forall k \in \{1, \dots, M\} \ \exists j \in \{1, \dots, N\} \ \text{s.t. } |\tilde{r}_k - r^1_j| \le \delta \}\right)\ge 1- \epsilon/2.
\end{align*}
On the other hand, by the choice of $\{\tilde{r}_1, \dots \tilde{r}_M\}$ for any $i\in \{0,1\}$ and any $j \in \{1, \dots, N\}$ with $r_j^i \in K$ there exists $k \in \{1, \dots, M\}$ such that $|r^i_j-\tilde{r}_k | \le \delta.$ Note that $\{\tilde{r}_1, \dots \tilde{r}_M \} \subseteq \text{supp}(\P^0) \cap K$ and we conclude that 
\eqref{eq. aim} holds. It then follows from our choice of $\delta$ that for all $N\geq N_*$ we have
\begin{align*}
\P^{\infty}\left(\pi^{\P^0(\cdot|K)}(g)- \pi^{\hat{\P}^0_N(\cdot|K)}(g) \le \epsilon/9,\ \pi^{\P^1(\cdot|K)}(g)- \pi^{\hat{\P}^1_N(\cdot|K)}(g) \le \epsilon/9 \right) > 1- \epsilon.
\end{align*}
Hence, by \eqref{eq. boundab}, we deduce that
\begin{align*}
\P^{\infty}\left( \pi^{\P^0}(g)-\pi^{\hat{\P}^0_N}(g) \le 4\epsilon/9, \ \pi^{\P^1}(g)-\pi^{\hat{\P}^1_N}(g) \le 4\epsilon/9 \right)>1-\epsilon.
\end{align*} 
Combining the above with \eqref{eq. start}  and \eqref{eq:boundondiffofpi} we have 
\begin{align*}
\P^{\infty}\left(\left|\pi^{\hat{\P}^0_N}(g)-\pi^{\hat{\P}^1_N}(g)\right|> \epsilon\right) < \epsilon.
\end{align*}
Using Strassen's theorem (\cite[Theorem 2.13, p. 30]{huber1996robust}) we deduce that $$d_L(\mathcal{L}_{\P^1}(\hat{\pi}_N), \mathcal{L}_{\P^0} (\hat{\pi}_N)) \le \epsilon.$$
This concludes the proof.
\end{proof}

\begin{Cor}\label{cor:Winfinity}
 Let $\P \in \mathcal{P}(\R_+^d)$ such that NA$(\P)$ holds.
\begin{enumerate}
\item[(i)] Let $g$ be continuous and $\mathfrak{P}\subseteq \mathcal{P}(\R_+^d)$. If $\P \in \mathfrak{P}$ and for all $\delta>0$, there exists a compact set $K\subseteq \R_+^d$ such that 
\begin{align*}
\sup_{\tilde \P \in \mathfrak{P}} \left(\pi^{\tilde\P}(g)-\pi^{\tilde\P(\cdot|K)}(g)\right) \le \delta,
\end{align*}
then $\hat{\pi}_N(g)$ is robust at $\P$ wrt. $\mathcal{W}^{\infty}$ on $\mathfrak{P}$.
\item[(ii)] Let $g$ be a continuous function of linear growth and $\mathfrak{P} \subseteq \mathcal{P}(\R_+)$. If $\P \in \mathfrak{P}$, $\mathfrak{P}$ is uniformly integrable and for all $\delta>0$  there exists $C>0$ such that
\begin{align}\label{eq. easy}
\sup_{\tilde \P \in \mathfrak{P}}\left(\pi^{\tilde \P}(g)-\sup_{\|d\Q/ d\tilde \P \|_{\infty}\le C, \ \Q \in \mathcal{M}} \E_{\Q}[g] \right)\le \delta,
\end{align}
then $\hat{\pi}_N$ is robust at $\P$ wrt. $\mathcal{W}^{\infty}$ on $\mathfrak{P}$.
\end{enumerate}
\end{Cor}

\begin{proof}
\textit{(i)}: Let $\epsilon>0$. By assumption we can find a compact set $K\subseteq \R_+^d$ such that 
\begin{align*}
 \sup_{\tilde \P \in \mathfrak{P}} \left(\pi^{\tilde\P}(g) -\pi^{\tilde{\P}(\cdot|K)}(g)\right)\le \epsilon/3.
\end{align*}
The rest of the proof follows as in the proof of Theorem \ref{Thm. infi} above using uniform continuity of $g$ on $K$.\\
\textit{(ii)}: Let $\epsilon>0$ and choose $C >0$ such that 
\begin{align*}
\sup_{\tilde{\P} \in \mathfrak{P}}\left(\pi^{\tilde{\P}}(g)-\sup_{\|d\Q/ d\tilde{\P} \|_{\infty}\le C, \ \Q \in \mathcal{M}} \E_{\Q}[g] \right)\le \epsilon/3,
\end{align*}
By assumption there exists a constant $D >0$ such that $g(r) \le D(1+|r|)$ for $|r|$ large enough. As $\mathfrak{P}$ is uniformly integrable, there exists $L>0$ such that
\begin{align*}
\sup_{\tilde\P \in \mathfrak{P}} \sup_{\Q \in \mathcal{M}, \ \|d\Q/ d\tilde\P \|_{\infty}\le C} \E_{\Q}[g\mathds{1}_{\{|r|\ge L\}}] &\le \sup_{\tilde\P \in \mathfrak{P}} C\E_{\tilde\P}[D(1+|r|)\mathds{1}_{\{|r|\ge L\}}] \le \epsilon/3.
\end{align*}
Note that there exists $A>0$ such that the hedging strategies 
\begin{align*}
\{ H \in \R^d \ | \ \pi^{\tilde\P(\cdot|[0,L]^d)}(g)+H(r-1) \ge g(r) \quad \tilde\P(\cdot|[0,L]^d)\text{-a.s.} \}
\end{align*}
contain an element bounded by some constant $A>0$ for all $L>0$ large enough: Otherwise there exist sequences $(L_n)_{n \in \N}$, $(H_n)_{n \in \N}$ and $(\tilde{\P}_n)_{n \in \N}$ such that $H_n \to \infty$. Note that by uniform integrability of $\tilde{\P}_n$ we have $\tilde{\P}_n \Rightarrow \tilde{\P}$ and also $\tilde{\P}_n( \cdot |[0,L_n]^d) \Rightarrow \tilde{\P}$, where NA$(\tilde{\P})$ holds. Take $\tilde{H}_n := H_n /|H_n|$, then after possibly taking a subsequence $\tilde{H}_n \to \tilde{H}$ with $|\tilde{H}|=1$ and $\tilde{H}(r-1) \ge 0$ $\tilde{\P}$-a.s., which leads to $\tilde{H}=0$ by NA$(\tilde{\P})$, a contradiction. Take $L>0$ such that 
\begin{align*}
\sup_{\tilde\P \in \mathfrak{P}} \sup_{\Q \in \mathcal{M}, \ \|d\Q/ d\tilde\P \|_{\infty}\le C} \E_{\Q}[A|r|\mathds{1}_{\{|r|\ge L\}}] &\le \sup_{\tilde\P \in \mathfrak{P}} C\E_{\tilde\P}[A|r|\mathds{1}_{\{|r|\ge L\}}] \le \epsilon/3.
\end{align*}
Then for all $\tilde{\P} \in \mathfrak{P}$
\begin{align*}
\pi^{\tilde{\P}}(g)-\epsilon 
&\le \sup_{\|d\Q/ d\tilde{\P} \|_{\infty}\le C, \ \Q \in \mathcal{M}} \E_{\Q}[g] -2\epsilon/3 
\\
&\le \sup_{\|d\Q/ d\tilde{\P} \|_{\infty}\le C, \ \Q \in \mathcal{M}} \E_{\Q}[g \mathds{1}_{\{|r| \le L \}}]- \epsilon/3\\
&\le
\sup_{\|d\Q/ d\tilde{\P} \|_{\infty}\le C, \ \Q \in \mathcal{M}} \E_{\Q}[g \mathds{1}_{\{|r| \le L \}}-A|r|\mathds{1}_{\{|r|\ge L\}}] \\
&\le
\sup_{\Q \sim \tilde{\P}, \ \Q \in \mathcal{M}} \E_{\Q}[g \mathds{1}_{\{|r| \le L \}}-A|r|\mathds{1}_{\{|r|\ge L\}}] \\
&=  \pi^{\tilde{\P}}(g \mathds{1}_{\{|r| \le L \}}-A|r|\mathds{1}_{\{|r|\ge L\}}) \le \pi^{\tilde{\P}(\cdot|[0,L]^d)}(g) \le \pi^{\tilde{\P}}(g). 
\end{align*}
Thus again we can restrict to $K=[0,L]^d$ as before and proceed as in the proof of Theorem \ref{Thm. infi}.
\end{proof}

\begin{proof}[Proof of Corollary \ref{cor:Wcinfty_rob}]
Recall that $\P$ is compactly supported, say $\text{supp}(\P)\subseteq \overline{B_R(0)}$ for some $R>0$, so that we can assume that $g$ is uniformly continuous. Note that we also have  that $\text{supp}(\hat{\P}_N)\subseteq \overline{B_R(0)}$, $\P^{\infty}$-a.s.
Theorem \ref{Thm. infi} then shows robustness of $\hat{\pi}_N$. Thus, for any $\epsilon>0$, there exists $\delta>0$ and $N_0 \in \N$ such that for all $N \ge N_0$ and all $\tilde{\P} \in \mathcal{P}(\R_+^d)$ we have:
\begin{align*}
\mathcal{W}^{\infty}(\tilde\P, \P)\le \delta \hspace{0.5cm} \Rightarrow \hspace{0.5cm} d_L(\mathcal{L}_{\tilde\P}(\hat{\pi}_N), \mathcal{L}_{\P} (\hat{\pi}_N)) \le \epsilon/3.
\end{align*}
Observe that $\mathcal{W}^{\infty}(\tilde\P, \P)\le \delta$ implies $d_H(\supp(\P),\supp(\tilde{\P}))\le \delta$. Proposition \ref{prop:d_H cont} states that the map $\tilde{\P} \mapsto \pi^{\tilde{\P}}(g)$ is continuous in the pseudo-metric $d_H(\text{supp}(\P), \text{supp}(\tilde{\P})$. By \cite[Prop. C.2]{bertsekas1978stochastic} the collection of closed subsets of $\overline{B_{R+1}(0)}$ equipped with the Hausdorff-metric is compact, so in particular there exists $N_1 \ge N_0$ such that for all $N \ge N_1$ and for all $\tdP \in B^{\infty}_{l_N}(\hat{\P}_N)$ we have $|\pi^{\tdP}(g)-\pi^{\hat{\P}_N}(g)|\le \epsilon/3$. This concludes the proof.
\end{proof}

\subsection{Additional results for Section \ref{sec:riskmeasures}}
\label{app:riskmeasures}

In this Section we present some simulations complementing Figure \ref{fig. avar} in the main article. More specifically we compare estimates for the quantity  $\pi^{\mathrm{AV@R}_{0.95}^{\P}}(g)$, where we set $g(r)=(r-1)^+$ (Figure \ref{fig. avar2}) for historical gold price (WGC/GOLD DAILY USD) returns and Apple (AAPL) returns, and $g(r)=|r-1|$ (Figure \ref{fig. avar3}) for historical S\&P500 and DAX30 returns. Finally, Figure \ref{fig. avar4} reproduces the case of S\&P500 and DAX30 returns and $g(r)=(r-1)^+$ from Figure \ref{fig. avar} but this time the GARCH(1,1)-estimator uses log-returns instead of simple returns. While before it ignored market crisis now it overreacts to it. We use $50$ or $100$ data points to build the estimates and plot the average of the last $5$ or $10$ running estimates. 

\begin{figure}[h!]
\centering
  \includegraphics[width=\textwidth]{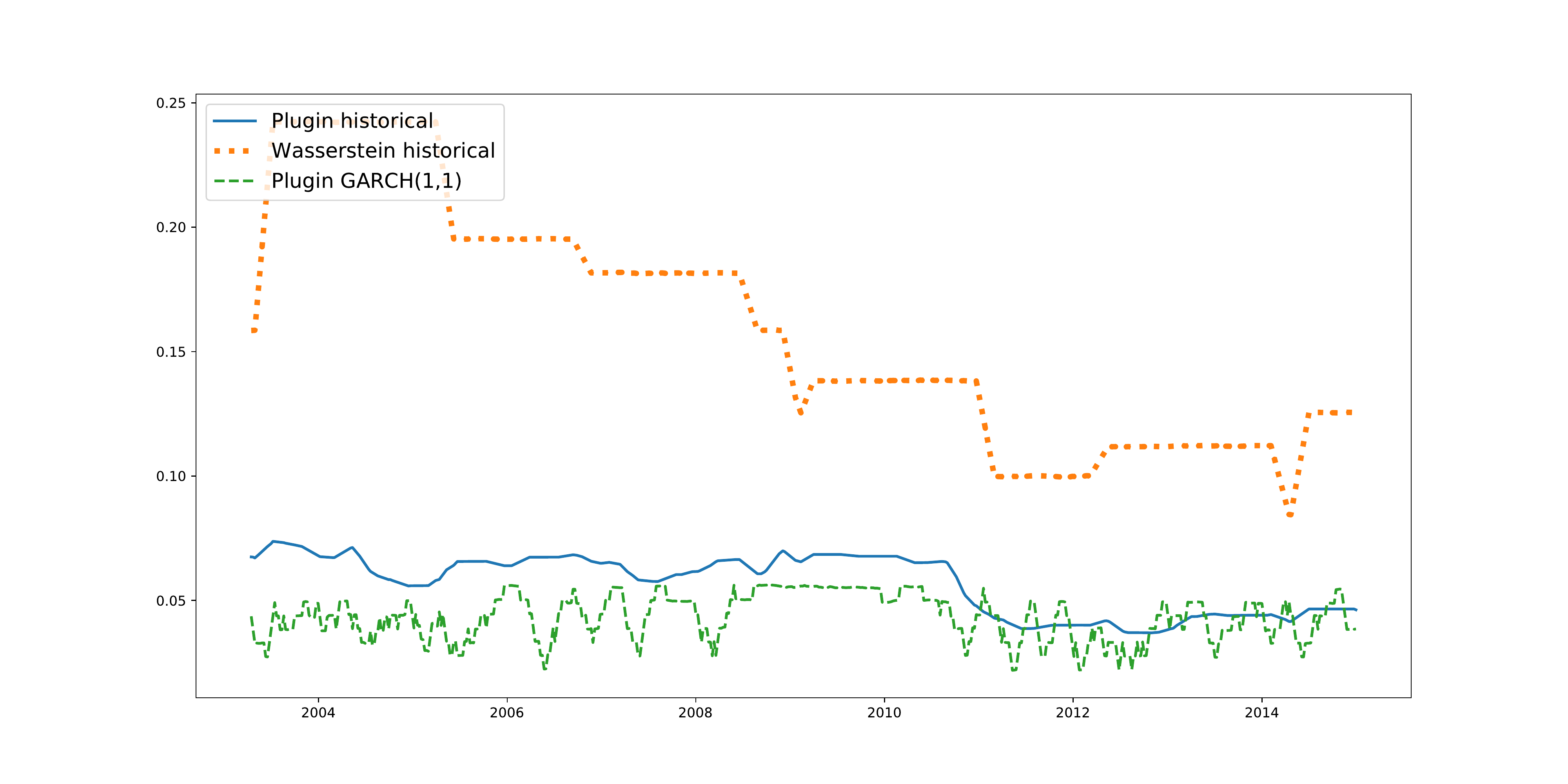}
  \includegraphics[width=\textwidth]{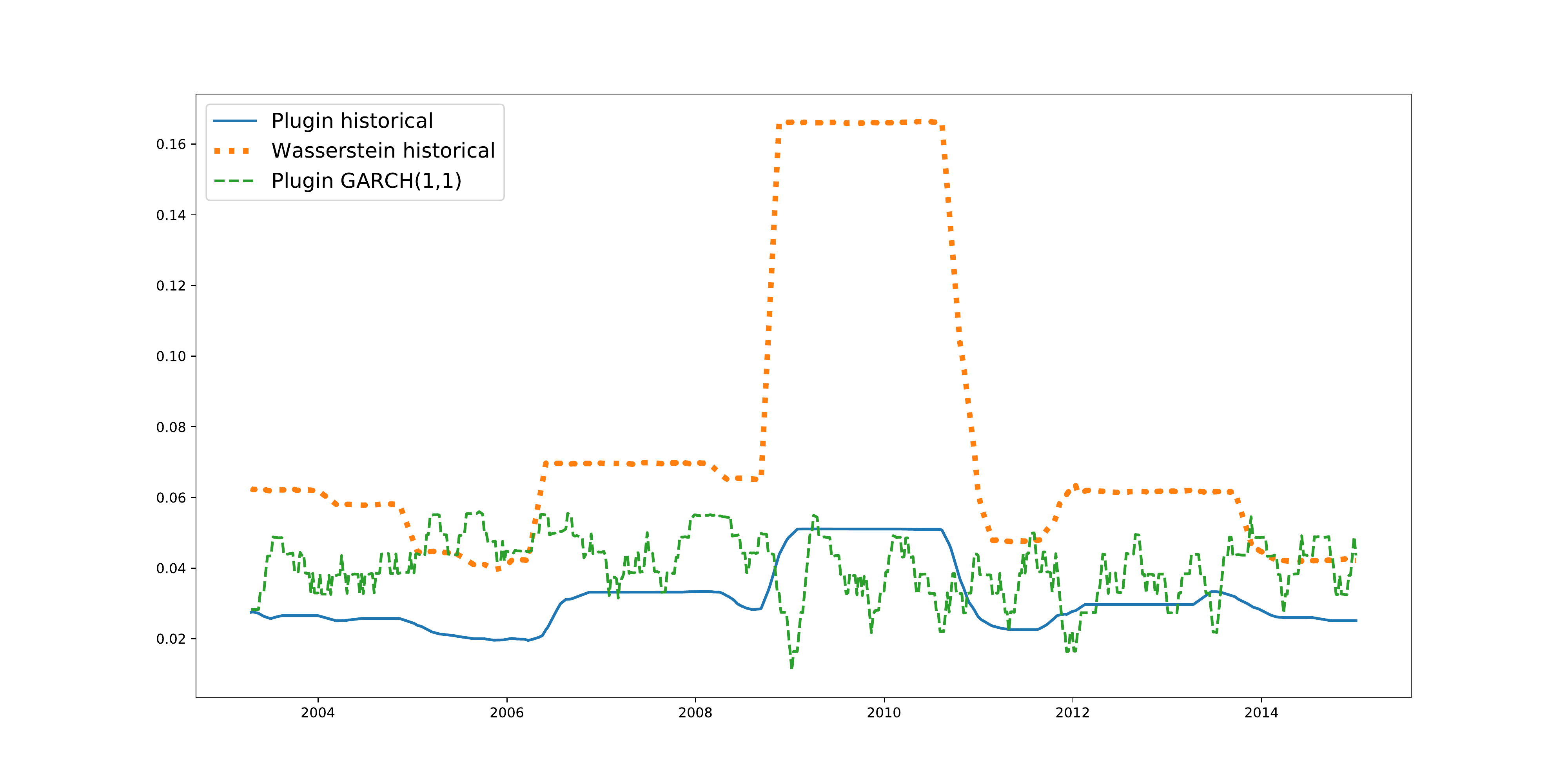}
   \includegraphics[width=\textwidth]{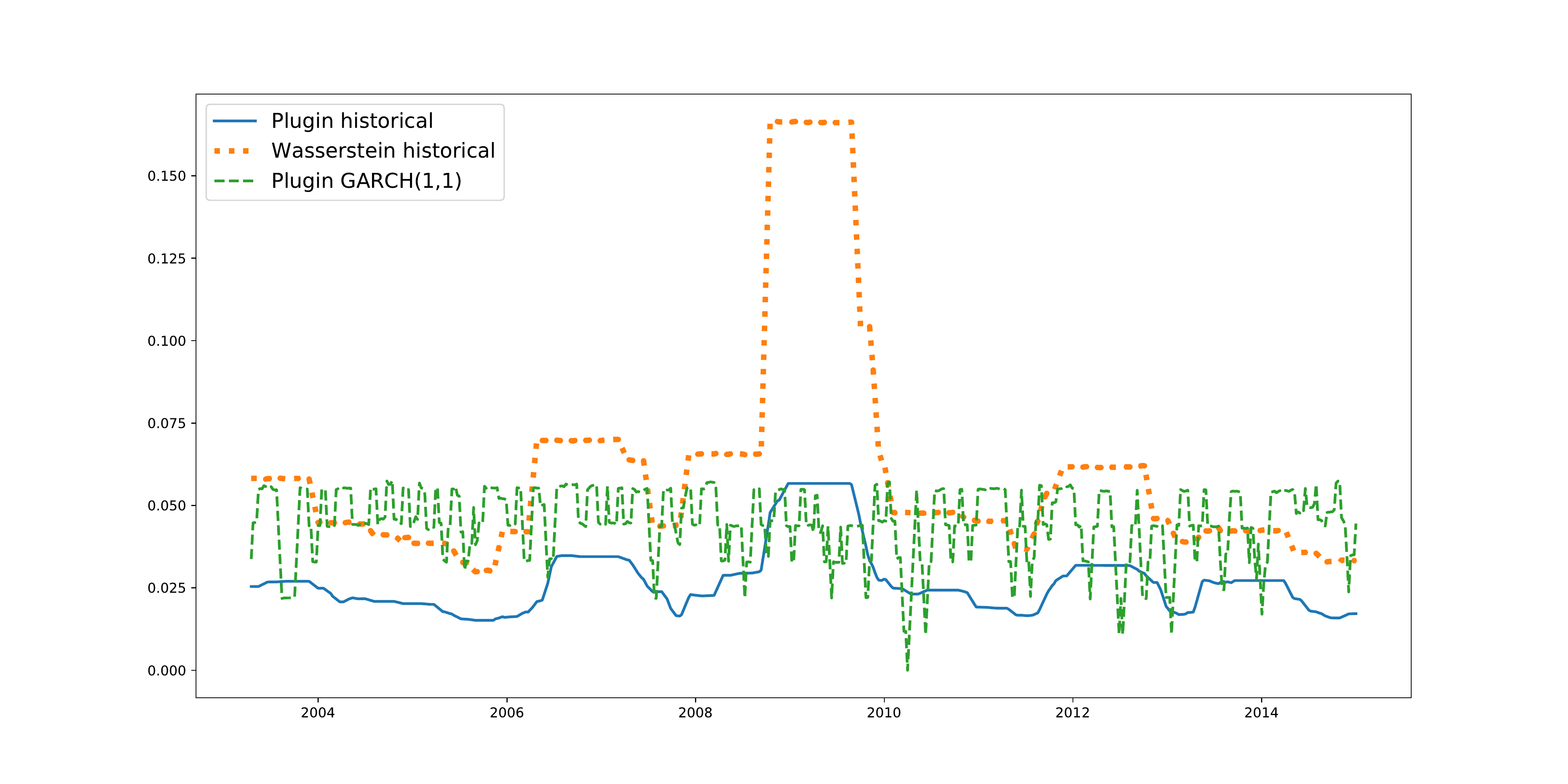}
\caption{Comparison of estimates for $\pi^{\mathrm{AV@R}_{0.95}^{\P}}((r-1)^+)$. The first two panes show estimates with a rolling window of $100$ data points and we plot the average of the last $10$ estimates. The first pane uses Apple returns while the second one uses gold returns. The last pane uses gold returns for estimates with a rolling window of $50$ data points where we plot the average of the last $5$ estimates.} \label{fig. avar2}
\end{figure}

\begin{figure}[h!]
\centering
  \includegraphics[width=\textwidth]{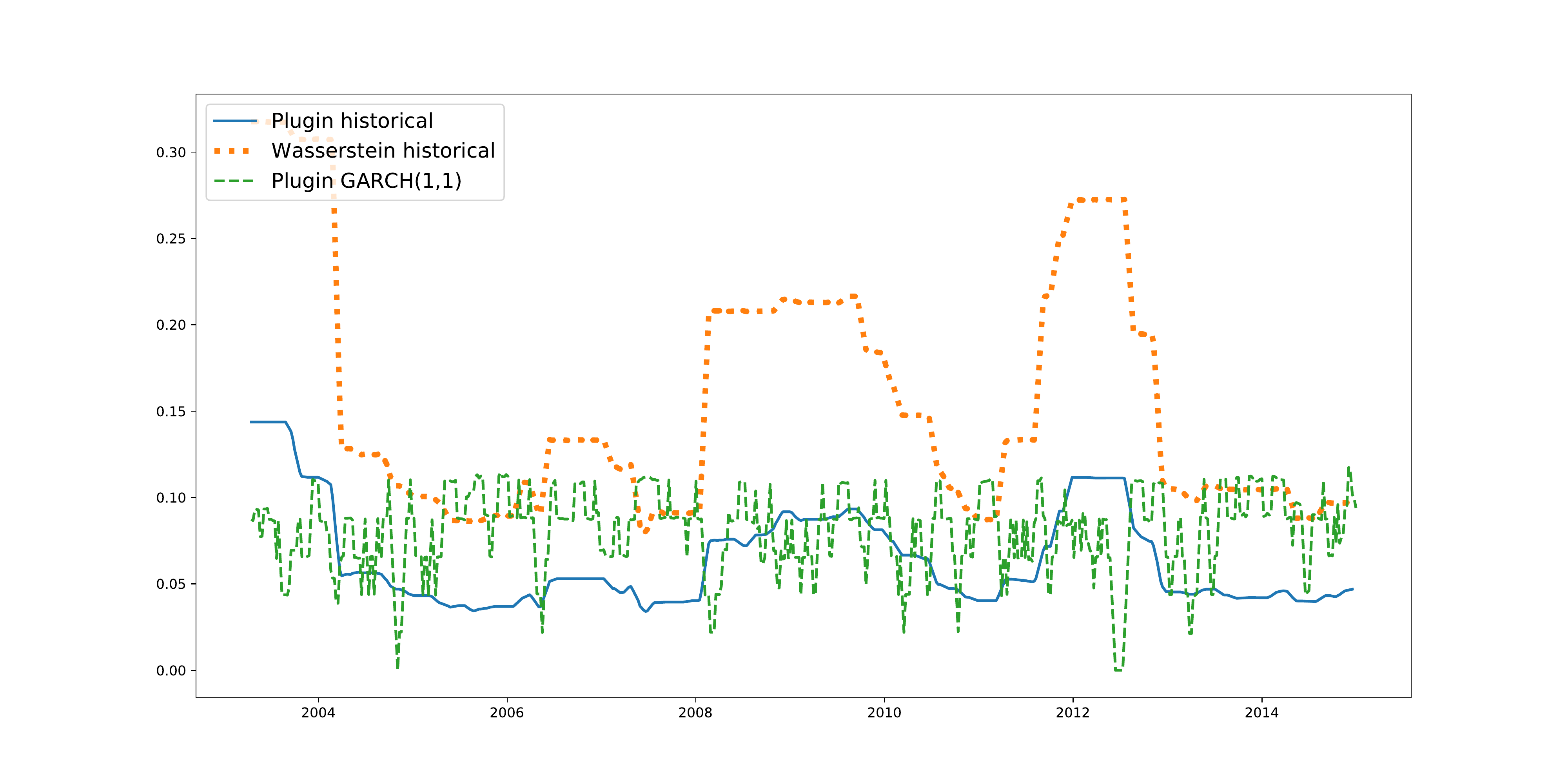}
  \includegraphics[width=\textwidth]{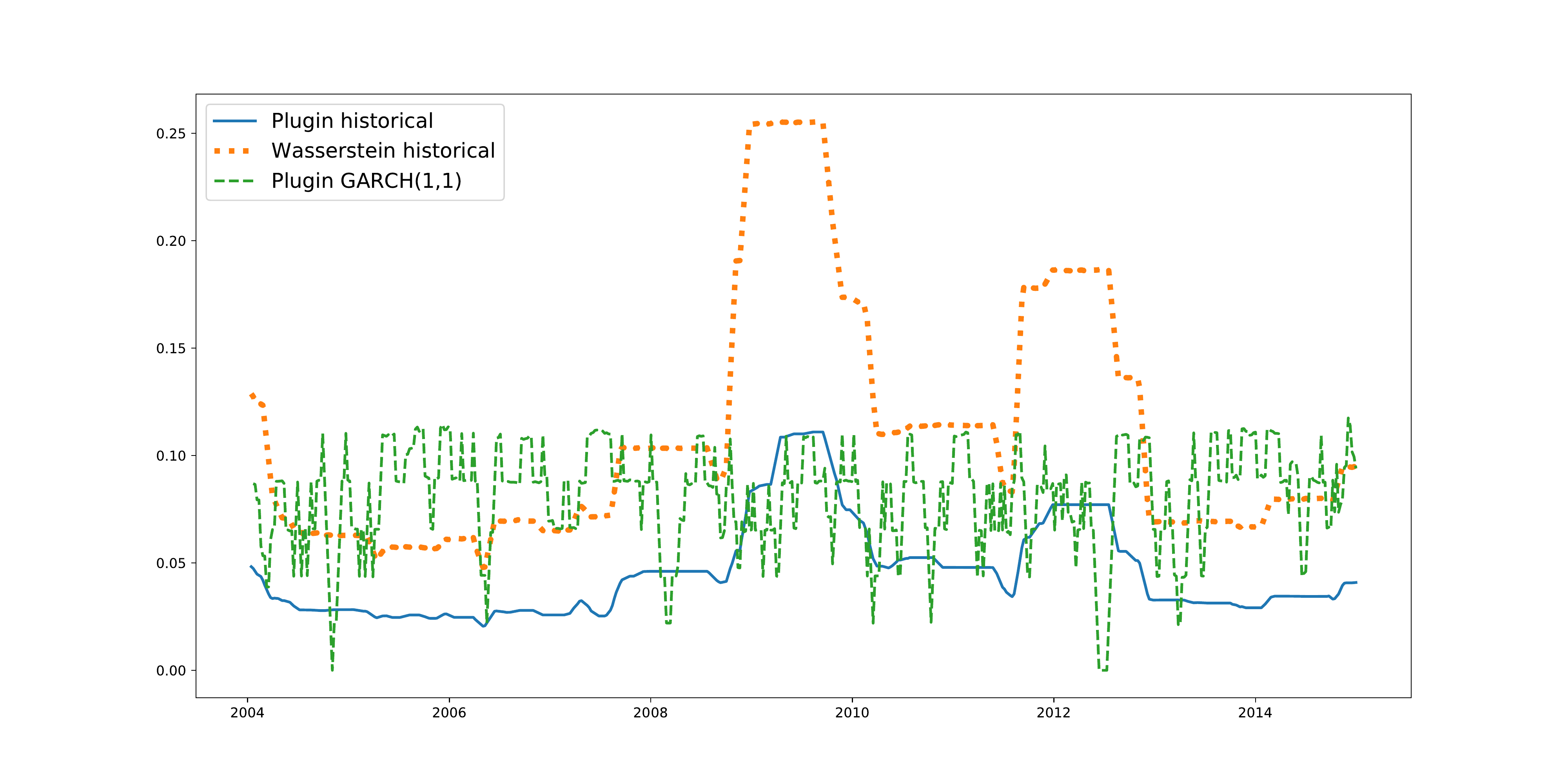}
\caption{Comparison of estimates for $\pi^{\mathrm{AV@R}_{0.95}^{\P}}(|r-1|)$. Estimates use a rolling window of $50$ data points and we plot the average of the last $5$ estimates. The first pane uses DAX30  returns while the second one uses S\&P500 returns.} \label{fig. avar3}
\end{figure}

\begin{figure}[h!]
\centering
  \includegraphics[width=\textwidth]{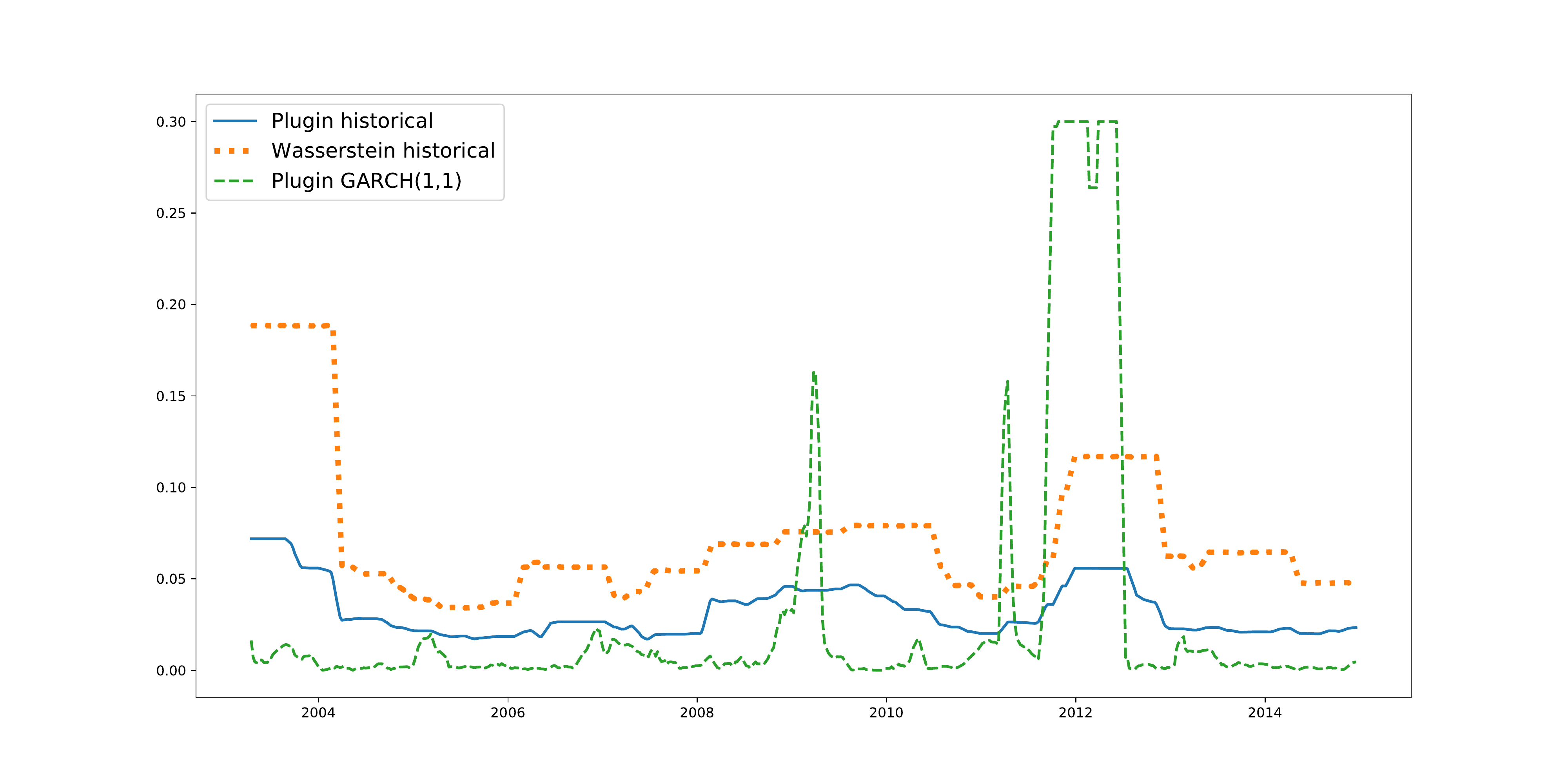}
  \includegraphics[width=\textwidth]{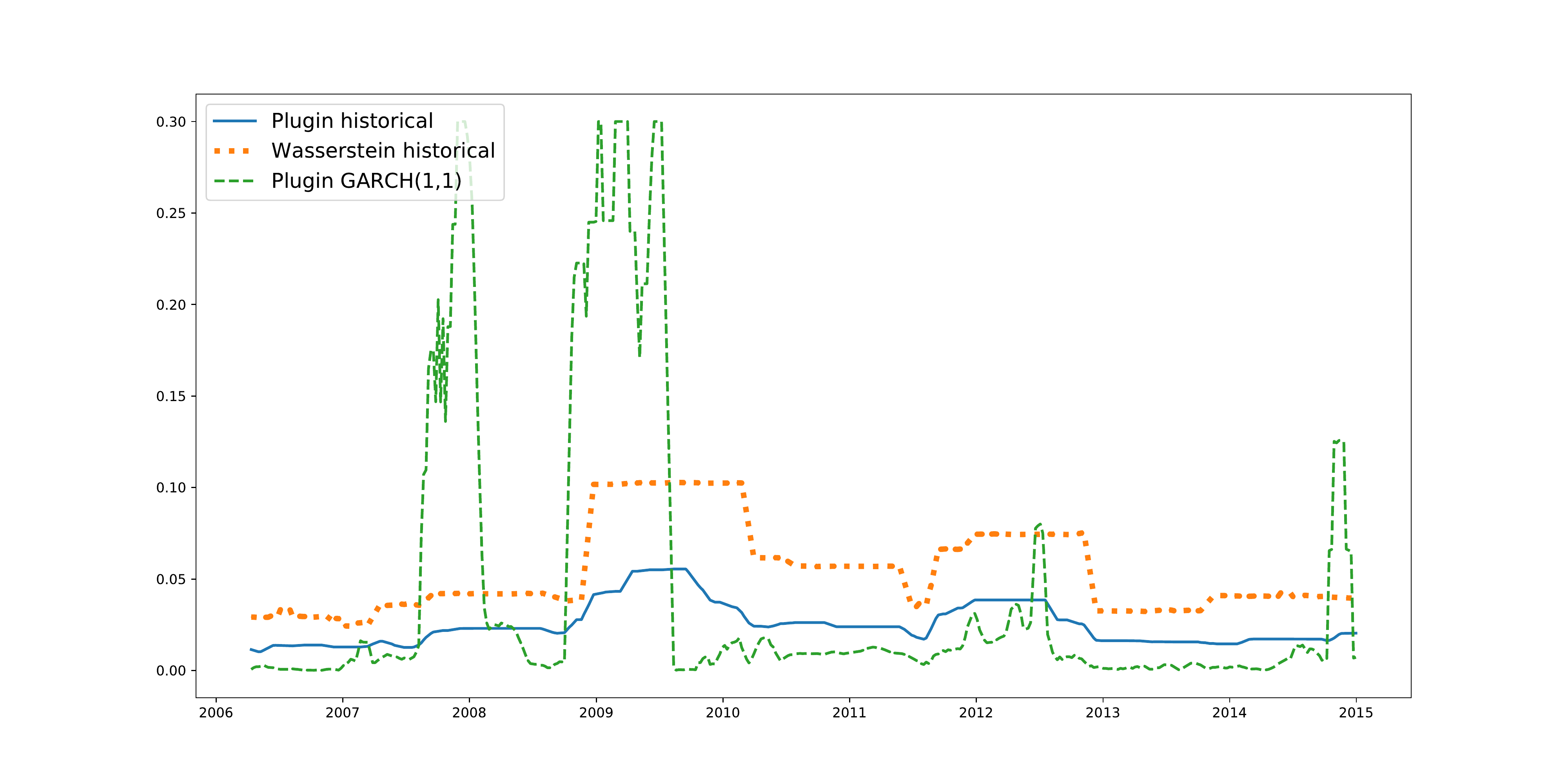}
\caption{Comparison of estimates for $\pi^{\mathrm{AV@R}_{0.95}^{\P}}((r-1)^+)$ estimated on log-returns for the GARCH(1,1) model. Estimates use a rolling window of $50$ data points and we plot the average of the last $5$ estimates. The first pane uses DAX30  returns while the second one uses S\&P500 returns.} \label{fig. avar4}
\end{figure}

\clearpage

\bibliographystyle{abbrvnat}
\bibliography{bib}

\end{document}